\documentclass[a4paper,reqno]{amsart}
\usepackage[latin1]{inputenc}
\usepackage{amssymb} 
\usepackage{latexsym}
\input amssym.def
\input amssym.tex
\newtheorem{definition}{Definition}[section]
\newtheorem{lemma}[definition]{Lemma}
\newtheorem{theorem}[definition]{Theorem}
\newtheorem{proposition}[definition]{Proposition}

\newtheorem{remark}[definition]{Remark}

\def\emph#1{{\bfseries\itshape{#1}}}
 1   
 1  

\def\a{\mathrm{ a}}
\def\b{\mathrm{b}}
\def\R{\mathbb{R}}               

\newcommand{\pd}[2]{\frac{\partial #1}{\partial #2}}

\def\quand{\qquad\text{and}\qquad}
\def\lcf{\lbrack\! \lbrack}
\def\rcf{\rbrack\! \rbrack}

\newcount\ancho
\newcount\anchom
\newcount\anchoa
\newcount\anchob
\newcommand{\ltilde}[2]{\ancho=#1 \anchom=\ancho \divide\anchom by 2
            \anchoa=\ancho \divide\anchoa by 4
        \anchob=\anchom \advance\anchob by \anchoa
      $\kern-5pt \begin{array}[b]{c}
                 \begin{picture}(1,1)(\anchom,0)
         \qbezier(0,2)(\anchoa,5)(\anchom,2)
         \qbezier(\anchom,2)(\anchob,-1)(\ancho,4)
         \qbezier(0,2)(\anchoa,4.5)(\anchom,1.8)
         \qbezier(\anchom,1.8)(\anchob,-1.5)(\ancho,4)
      \end{picture} \\[-4pt]
       \mbox{#2}
       \end{array} \kern-9pt$
       }

\newcommand\prol{\@ifstar{\@proldf}{\@prolpf}}  
\def\@prolpf{\@ifnextchar[{\@prolpf@wrt}{\@prolpf@}}
\def\@prolpf@wrt[#1]#2{\@ifnextchar[{\@prolpf@wrt@at{#1}{#2}}{\@prolpf@wrt@{#1}{#2}}}
\def\@prolpf@wrt@at#1#2[#3]{\prolsymbol^{#1}_{#3}#2}
\def\@prolpf@wrt@#1#2{\prolsymbol^{#1}#2}
\def\@prolpf@#1{\@ifnextchar[{\@prolpf@at{#1}}{\@prolpf@@{#1}}}
\def\@prolpf@at#1[#2]{\prolsymbol_{#2}#1}
\def\@prolpf@@#1{\prolsymbol#1}
\def\@proldf{\@ifnextchar[{\@proldf@wrt}{\@proldf@}}
\def\@proldf@wrt[#1]#2{\@ifnextchar[{\@proldf@wrt@at{#1}{#2}}{\@proldf@wrt@{#1}{#2}}}
\def\@proldf@wrt@at#1#2[#3]{\prolsymbol^{*#1}_{#3}#2}
\def\@proldf@wrt@#1#2{\prolsymbol^{*#1}#2}
\def\@proldf@#1{\@ifnextchar[{\@proldf@at{#1}}{\@proldf@@{#1}}}
\def\@proldf@at#1[#2]{\prolsymbol^*_{#2}#1}
\def\@proldf@@#1{\prolsymbol^*#1}
\def\prolsymbol{\mathcal{T}}
\makeatother

\begin{document}

\title[ Nonholonomic Systems on Lie affgebroids]{A general framework for nonholonomic mechanics:\newline Nonholonomic Systems on Lie affgebroids}

\author[D. Iglesias]{David Iglesias}
\address{D.\ Iglesias:
Instituto de Matem\'aticas y F{\'\i}sica Fundamental, Consejo
Superior de Investigaciones Cient\'{\i}ficas, Serrano 113 bis,
28006 Madrid, Spain} \email{iglesias@imaff.cfmac.csic.es}

\author[J.\ C.\ Marrero]{Juan C.\ Marrero}
\address{Juan C.\ Marrero:
Departamento de Matem\'atica Fundamental, Facultad de
Ma\-te\-m\'a\-ti\-cas, Universidad de la Laguna, La Laguna,
Tenerife, Canary Islands, Spain} \email{jcmarrer@ull.es}

\author[D.\ Mart\'{\i}n de Diego]{D. Mart\'{\i}n de Diego}
\address{D.\ Mart\'{\i}n de Diego:
Instituto de Matem\'aticas y F{\'\i}sica Fundamental, Consejo
Superior de Investigaciones Cient\'{\i}ficas, Serrano 123, 28006
Madrid, Spain} \email{d.martin@imaff.cfmac.csic.es}

\author[D. Sosa]{Diana Sosa}
\address{Diana Sosa:
Departamento de Matem\'atica Fundamental, Facultad de
Ma\-te\-m\'a\-ti\-cas, Universidad de la Laguna, La Laguna,
Tenerife, Canary Islands, Spain} \email{dnsosa@ull.es}



\begin{abstract}
\noindent This paper presents a geometric description of
Lagrangian and Hamiltonian systems on Lie affgebroids subject to
affine nonholonomic constraints. We define the notion of
nonholonomically constrained system, and characterize regularity
conditions that guarantee that the dynamics of the system can be
obtained as a suitable projection of the unconstrained dynamics.
It is shown that one can define an almost aff-Poisson bracket on
the constraint AV-bundle, which plays a prominent role in the
description of nonholonomic dynamics. Moreover, these developments
give a  general description of nonholonomic systems and the
unified treatment permits to study nonholonomic systems after or
before reduction in the same framework. Also, it is not necessary
to distinguish between  linear or affine constraints and the
methods are valid for explicitly time-dependent systems.

\end{abstract}



\maketitle

\tableofcontents

\setcounter{section}{0}

\section{Introduction}
Nonholonomic constraints is one of the more fascinating and
applied  topics actually. First, there are many open questions
related with this subject: characterization of the integrability,
construction of geometric integrators, stabilization,
controllability... But morever, there is a wide range of
applications of this kind of systems in engineering, robotics...
(see \cite{B,Cort,NF} and references therein).

During the last years, many authors have studied in detail the
geometry of nonholonomic systems. Some of them developed a
geometric formalism for the most typical nonholonomic systems, the
ones determined by a mechanical Lagrangian, that is,
\[
L(v_q)=\frac{1}{2}g(v_q, v_q)-V(q), \ v_q\in T_qQ,
\]
where $V: Q \to \R$ is the potential energy on the configuration
space $Q$, $g$ is a Riemannian metric on $Q$ and, additionally,
the system is subjected to linear constraints on the velocities,
expressed as a nonintegrable distribution on $Q$
\cite{Cort,Lewis}. The usual formalism for these systems was the
use of adequate  projections of the Levi-Civita connection
associated to $g$ to obtain the equations of motion of the system.
Other authors preferred to work on the tangent bundle of the
configuration space, which permits to introduce more general
Lagrangians and different type of constraints (linear or
nonlinear). Usually these systems are specified by a Lagrangian
function $L: TQ\to \R$ and a constraint submanifold $M$ of $TQ$
\cite{LeMaMa,LD}. Moreover, also a description in terms of a
nonholonomic bracket was introduced \cite{CaLeMa,KoMa,VaMa} with
considerable applications to reduction. In this sense, reduction
of nonholonomic systems was intensively studied u\-sing different
geometric techniques. For instance, introducing modified Ereshmann
connections (nonholonomic connection) \cite{BKMM}, using a
symplectic distribution on the constraint submanifold
\cite{BaSn,Sn} and distinguishing the different casuistic
depending on the `position' of the Lie group of symmetries acting
on the system and the nonholonomic distribution or in a Poisson
context \cite{Mar} (see \cite{Cort} and references therein).
Moreover, using jet bundle techniques, many authors studied the
geometry of nonholonomic systems admitting an extension to
explicitly time-dependent nonholonomic systems
\cite{CLMM,LMdM,LeMaMa2,SaCaSa,SaSaCa,SaSaCa2} and ready for
extension to nonholonomic field theories \cite{VCLM}.

Today we find an scenario with a very rich  theory but with an
important lack: a general framework  unifying the different
casuistic (unreduced and reduced equations, systems subjected to
linear or affine constraints, time-dependent  or time-independent
systems...).

The aim of the present paper is cover this lack and present a
geometrical framework covering the different cases. Obviously, in
order to reach this objective it will be necessary to use some new
and sophisticated  techniques, in particular, the concept of a Lie
affgebroid (an affine version of a Lie algebroid) and appropriate
versions of  Lagrangian and nonholonomic mechanics in this
setting.

The  previous and other motivations coming from other fields
(topology, algebraic geometry...) have recently caused a  lot of
interest in the study of Lie algebroids \cite{Ma}, which in our
setting can be thought of as `generalized tangent bundles', since
it allows to consider, in a unified formalism, mechanical systems
on Lie algebras, integrable distributions, tangent bundles,
quotients of tangent bundles by Lie groups when an action is
considered... In this sense,  Lie algebroids can be used to give
unified geometric descriptions of Hamiltonian and Lagrangian
Mechanics. In \cite{We}, A. Weinstein introduced Lagrangian
systems on a Lie algebroid $E\to M$ by means of the linear Poisson
structure on the dual bundle $E^*$ and a Legendre-type map from
$E$ to $E^*$. Motivated by this description, E. Mart\'{\i}nez
\cite{EMar} developed a geometric formalism on Lie algebroids
extending the Klein's formalism in ordinary Lagrangian Mechanics
on tangent bundles. This line of research has been followed in
\cite{LMM,M2}. In fact, a geometric description of Lagrangian and
Hamiltonian dynamics on Lie algebroids, in terms of Lagrangian
submanifolds of symplectic Lie algebroids, was given in
\cite{LMM}. More recently, in \cite{CoLeMaMa} (see also
\cite{CoMa,Mest,MeLa}) a comprehensive treatment of Lagrangian
systems on Lie algebroids subject to linear nonholonomic
constraints was developed. The proposed formalism allows us to
treat in a unified way a variety of situations for
time-independent Lagrangian systems subject to linear nonholonomic
constraints (systems with symmetry, nonholonomic LL systems,
nonholonomic LR systems,...).

On the other hand, in \cite{GGrU,MMeS} an affine version of the
notion of a Lie algebroid structure was introduced. The resultant
geometric object is called a Lie affgebroid structure. Lie
affgebroid structures may be used to develop a time-dependent
version of Lagrange and Hamilton equations on Lie algebroids (see
\cite{GGU2,IMPS,M,MMeS,SaMeMa}). In fact, one may obtain Lagrange
and Hamilton equations on Lie affgebroids using a cosymplectic
formalism \cite{IMPS,M}. In the same setting of Lie affgebroids,
one also may obtain the Hamilton equations using the notion of an
aff-Poisson structure on an AV-bundle (see \cite{GGU2,IMPS}). An
AV-bundle is an affine bundle of rank $1$ modelled on the trivial
vector bundle and an aff-Poisson bracket on an AV-bundle may be
considered as an affine version of a Poisson bracket on a manifold
(see \cite{GGrU}).

The aim of this paper is develop a geometric description of
Lagrangian systems subject to affine constraints using the Lie
affgebroid theory. The general geometric framework proposed covers
the most interesting previous methods in the literature.

The paper is organized as follows. In Section II, we recall
several constructions (which will be useful in the sequel) about
the geometric description of Lagrangian and Hamiltonian Mechanics
on Lie affgebroids and the equivalence between both formalisms in
the hyperregular case. In Section III, we introduce the
Lagrange-d'Alembert equations for an affine nonholonomic
Lagrangian system on a Lie affgebroid. In Section IV, we discuss
the existence and uniqueness of solutions for this type of
systems. Moreover, we prove that in the regular case the
nonholonomic dynamics can be obtained by different projections
from the unconstrained dynamics. In Section V, we develop the
Hamiltonian description and we discuss the equivalence between the
Lagrangian and Hamiltonian formalism. We also introduce the
nonholonomic bracket which gives the evolution of an observable
for the nonholonomic dynamics. The nonholonomic bracket is an
almost aff-Poisson bracket (an aff-Poisson bracket which doesn't
satisfy, in general, "the Jacobi identity") on the constraint
AV-bundle. In Section VI, we apply the results obtained in the
paper to two particular cases: a linear nonholonomic Lagrangian
system on a Lie algebroid and an standard affine nonholonomic
Lagrangian system on the $1$-jet bundle of local sections of a
fibration $\tau: M \to \R$. As a consequence, we directly deduce
some results obtained in \cite{CoLeMaMa,LMdM,LeMaMa2} (see also
\cite{CLMM}). In addition, we also analyze (in the Lie affgebroid
setting) the reduced equations for a well-known example of an
affine nonholonomic mechanical system: a homogeneous rolling ball
without sliding on a rotating table with time-dependent angular
velocity. For this purpose, we will use a particular class of Lie
affgebroids, namely, Atiyah affgebroids.

The paper ends with our conclusions and a description of future
research directions.

\section{Lagrangian and Hamiltonian formalism on Lie affgebroids} \label{LaLa}
In this section we recall some well-known facts concerning the
geometry of  Lie algebroids and Lie affgebroids.

\subsection{Lie algebroids}
Let $E$ be a vector bundle of rank $n$ over the manifold $M$ of
dimension $m$ and $\tau_{E}:E\rightarrow M$ be the vector bundle
projection. Denote by $\Gamma(\tau_{E})$ the
$C^{\infty}(M)$-module of sections of $\tau_{E}:E\rightarrow M$. A
{\it Lie algebroid structure} $(\lcf\cdot,\cdot\rcf_{E},\rho_{E})$
on $E$ is a Lie bracket $\lcf\cdot,\cdot\rcf_{E}$ on the space
$\Gamma(\tau_{E})$ and a bundle map $\rho_{E}:E\rightarrow TM$,
called {\it the anchor map}, such that if we also denote by
$\rho_{E}:\Gamma(\tau_{E})\rightarrow\mathfrak{X}(M)$ the
homomorphism of $C^{\infty}(M)$-modules induced by the anchor map
then $\lcf X,fY\rcf_{E}=f\lcf X,Y\rcf_{E}+\rho_{E}(X)(f)Y,$ for
$X,Y\in\Gamma(\tau_{E})$ and $f\in C^{\infty}(M)$. The triple
$(E,\lcf\cdot,\cdot\rcf_{E},\rho_{E})$ is called a {\it Lie
algebroid over $M$} (see \cite{Ma}). In such a case, the anchor
map $\rho_{E}:\Gamma(\tau_{E})\rightarrow\mathfrak{X}(M)$ is a
homomorphism between the Lie algebras
$(\Gamma(\tau_{E}),\lcf\cdot,\cdot\rcf_{E})$ and
$(\mathfrak{X}(M),[\cdot,\cdot])$.

If $(E,\lcf\cdot,\cdot\rcf_{E},\rho_{E})$ is a Lie algebroid, one
may define a cohomology operator which is called {\it the
differential of $E$},
$d^E:\Gamma(\wedge^k\tau_{E}^*)\longrightarrow\Gamma(\wedge^{k+1}\tau_{E}^*)$,
as follows
\begin{equation}\label{dE}
\begin{array}{lcl}
(d^E\mu)(X_0,\dots,X_k)&=&\displaystyle\sum_{i=0}^k(-1)^i\rho_{E}(X_i)(\mu(X_0,\dots,\widehat{X_i},\dots,X_k))\\
 &+&\displaystyle\sum_{i<j}(-1)^{i+j}\mu(\lcf
 X_i,X_j\rcf_{E},X_0,\dots,\widehat{X_i},\dots,\widehat{X_j},\dots,X_k),
\end{array}
\end{equation}
for $\mu\in\Gamma(\wedge^k\tau_{E}^*)$ and
$X_0,\dots,X_k\in\Gamma(\tau_{E})$. Moreover, if
$X\in\Gamma(\tau_{E})$, one may introduce, in a natural way, {\it
the Lie derivative with respect to $X$}, as the operator
${\mathcal
L}_X^E:\Gamma(\wedge^k\tau_{E}^*)\longrightarrow\Gamma(\wedge^{k}\tau_{E}^*)$
given by ${\mathcal L}_X^E=i_X\circ d^E+d^E\circ i_X.$

If $E$ is the standard Lie algebroid $TM$ then the differential
$d^E=d^{TM}$ is the usual exterior differential associated with
$M$. On the other hand, if $E$ is a real Lie algebra ${\mathfrak
g}$ of finite dimension then ${\mathfrak g}$ is a Lie algebroid
over a single point and the differential $d^{\mathfrak g}$ is the
algebraic differential of the Lie algebra.

Now, suppose that $(E,\lcf\cdot,\cdot\rcf_{E},\rho_{E})$ and
$(E',\lcf\cdot,\cdot\rcf_{E'},\rho_{E'})$ are Lie algebroids over
$M$ and $M'$, respectively, and that $F:E\to E'$ is a vector
bundle morphism over the map $f:M\to M'.$ Then $(F,f)$ is said to
be a {\it Lie algebroid morphism} if
$$d^{E}((F,f)^*\phi')=(F,f)^*(d^{E'}\phi'),\mbox{ for }\phi'\in\Gamma(\wedge^k(\tau_{E'})^*) \mbox{ and for all } k.$$
Note that $(F,f)^*\phi'$ is the section of the vector bundle
$\wedge^k E^*\rightarrow M$ defined by
$$((F,f)^*\phi')_{x}(\a_1,\dots,\a_k)=\phi'_{f(x)}(F(\a_1),\dots,F(\a_k)),$$
for $x\in M$ and $\a_1,\dots,\a_k\in E_{x}$. If $(F,f)$ is a Lie
algebroid morphism, $f$ is an injective immersion and
$F_{|E_x}:E_x\rightarrow E'_{f(x)}$ is injective, for all $x\in
M$, then $(E,\lcf\cdot,\cdot\rcf_{E},\rho_{E})$ is said to be a
{\it Lie subalgebroid} of
$(E',\lcf\cdot,\cdot\rcf_{E'},\rho_{E'})$.

If we take local coordinates $(x^i)$ on $M$ and a local basis
$\{e_\alpha\}$ of sections of $E$, then we have the corresponding
local coordinates $(x^i,y^\alpha)$ on $E$, where $y^\alpha(\a)$ is
the $\alpha$-th coordinate of $\a\in E$ in the given basis. Such
coordinates determine local functions $\rho_\alpha^i$,
$C_{\alpha\beta}^{\gamma}$ on $M$ which contain the local
information of the Lie algebroid structure, and accordingly they
are called the {\it structure functions of the Lie algebroid.}
They are given by
\[
\rho_{E}(e_\alpha)=\rho_\alpha^i\frac{\partial }{\partial
x^i}\;\;\;\mbox{ and }\;\;\; \lcf
e_\alpha,e_\beta\rcf_{E}=C_{\alpha\beta}^\gamma e_\gamma,
\]
and they satisfy the following relations
\begin{align*}
  \rho^j_\alpha\pd{\rho^i_\beta}{x^j} -
  \rho^j_\beta\pd{\rho^i_\alpha}{x^j} = \rho^i_\gamma
  C^\gamma_{\alpha\beta}
  \quand
  \sum_{\mathrm{cyclic}(\alpha,\beta,\gamma)} \left[\rho^i_\alpha\pd{
      C^\nu_{\beta\gamma}}{x^i} + C^\mu_{\beta\gamma}
    C^\nu_{\alpha\mu}\right]=0.
\end{align*}

If $f\in C^\infty(M)$, we have that
$$d^E f=\frac{\partial f}{\partial x^i}\rho_\alpha^i e^\alpha,$$
where $\{e^\alpha\}$ is the dual basis of $\{e_\alpha\}$.
Moreover, if $\theta\in \Gamma(\tau_{E}^*)$ and
$\theta=\theta_\gamma e^\gamma$ it follows that
$$d^E \theta=(\frac{\partial \theta_\gamma}{\partial
x^i}\rho^i_\beta-\frac{1}{2}\theta_\alpha
C^\alpha_{\beta\gamma})e^{\beta}\wedge e^\gamma.$$

\subsection{The prolongation of a Lie algebroid over a
fibration}\label{sec1.1.1}

In this section, we will recall the definition of the Lie
algebroid structure on the prolongation of a Lie algebroid over a
fibration (see \cite{HM,LMM}).

Let $(E,\lcf\cdot,\cdot\rcf_{E},\rho_{E})$ be a Lie algebroid of
rank $n$ over a manifold $M$ of dimension $m$ with vector bundle
projection $\tau_{E}:E\rightarrow M$ and $\pi:M'\rightarrow M$ be
a fibration.

We consider the subset ${\mathcal T}^EM'$ of $E\times TM'$ and the
map $\tau_{E}^{\pi}:{\mathcal T}^EM'\rightarrow M'$ defined by
$${\mathcal T}^EM'=\{(\b,v')\in E\times
TM'/\rho_{E}(\b)=(T\pi)(v')\},\;\;\;\;\tau_{E}^{\pi}(\b,v')=\pi_{M'}(v'),$$
where $T\pi:TM'\rightarrow TM$ is the tangent map to $\pi$ and
$\pi_{M'}:TM'\rightarrow M'$ is  the canonical projection. Then,
$\tau_{E}^{\pi}:{\mathcal T}^EM'\rightarrow M'$ is a vector bundle
over $M'$ of rank $n+\dim M'-m$ which admits a Lie algebroid
structure $(\lcf\cdot,\cdot\rcf_{E}^{\pi},\rho_{E}^{\pi})$
characterized by
$$\lcf(X\circ\pi,U'),(Y\circ\pi,V')\rcf_{E}^{\pi}=(\lcf X,Y\rcf_{E}\circ\pi
,[U',V']),\;\;\rho_{E}^{\pi}(X\circ\pi,U')=U',$$
for all  $X,Y\in \Gamma(\tau_{E})$  and $U', V'$ vector fields
which are $\pi$-projectable to $\rho_{E}(X)$ and $\rho_{E}(Y)$,
res\-pec\-tively.
 $({\mathcal
T}^EM',\lcf\cdot,\cdot\rcf_{E}^{\pi},\rho_{E}^{\pi})$ is called
{\it the prolongation of the Lie algebroid $E$ over the fibration
$\pi$ or the $E$-tangent bundle to $M'$} (for more details, see
\cite{HM,LMM}).

An element $Z\in{\mathcal T}^EM'$ is said to be {\it vertical} if
its projection onto the first factor is zero. Therefore, it is of
the form $(0,v_{x'})$, with $v_{x'}$ a tangent vector to $M'$ at
$x'$ which is $\pi$-vertical.

Next, we consider a particular case of the above construction. Let
$E$ be a Lie algebroid over a manifold $M$ with vector bundle
projection $\tau_{E}:E\rightarrow M$ and ${\mathcal T}^EE^*$ be
the prolongation of $E$ over the projection
$\tau_{E}^*:E^*\rightarrow M$. ${\mathcal T}^EE^*$ is a Lie
algebroid over $E^*$ and we can define a canonical section
$\lambda_E$ of the vector bundle $({\mathcal T}^EE^*)^*\rightarrow
E^*$ as follows. If $\a^*\in E^*$ and $(\b,v)\in{\mathcal
T}^E_{\a^*}E^*$ then
\begin{equation}\label{lambdaE}
\lambda_E(\a^*)(\b,v)=\a^*(\b).
\end{equation}
$\lambda_E$ is called the {\it Liouville section} associated with
the Lie algebroid $E.$

Now, one may consider the nondegenerate $2$-section
$\Omega_E=-d^{{\mathcal T}^EE^*}\lambda_E$ of $({\mathcal
T}^EE^*)^*\rightarrow E^*$. It is clear that $d^{{\mathcal
T}^EE^*}\Omega_E=0$. In other words, $\Omega_E$ is a symplectic
section. $\Omega_E$ is called {\it the canonical symplectic
section} associated with the Lie algebroid $E$. Using the
symplectic section $\Omega_E$ one may introduce a linear Poisson
structure $\Pi_{E^*}$ on $E^*$, with linear Poisson bracket
$\{\cdot, \cdot\}_{E^*}$ given by
\begin{equation}\label{2.2'}
\{F,G\}_{E^*}=\Omega_E(X_F^{\Omega_E},X_G^{\Omega_E}), \mbox{ for
} F,G\in C^\infty(E^*), \end{equation} where $X_F^{\Omega_E}$ and
$X_G^{\Omega_E}$ are the Hamiltonian sections associated with $F$
and $G$, that is, $i_{X_F^{\Omega_E}}\Omega_E=d^{{\mathcal
T}^EE^*}F$ and $i_{X_G^{\Omega_E}}\Omega_E=d^{{\mathcal
T}^EE^*}G.$

Suppose that $(x^i)$ are local coordinates on an open subset $U$
of $M$ and that $\{e_{\alpha}\}$ is a local basis of sections of
the vector bundle $\tau_{E}^{-1}(U)\rightarrow U$ as above. Then,
$\{\tilde{e}_{\alpha},\bar{e}_{\alpha}\}$ is a local basis of
sections of the vector bundle
$(\tau_{E}^{\tau_{E}^*})^{-1}((\tau_{E}^*)^{-1}(U))\rightarrow
(\tau_{E}^*)^{-1}(U)$, where $\tau_{E}^{\tau_{E}^*}:{\mathcal
T}^EE^*\rightarrow E^*$ is the vector bundle projection and
$$\tilde{e}_{\alpha}(\a^*)=\Big(e_{\alpha}(\tau_{E}^*(\a^*)),\rho_{\alpha}^i\displaystyle\frac{\partial}{\partial
x^i}_{|\a^*}\Big),\;\;\;\bar{e}_{\alpha}(\a^*)=\Big(0,\displaystyle\frac{\partial}{\partial
y_{\alpha}}_{|\a^*}\Big).
$$
Here, $(x^i,y_{\alpha})$ are the local coordinates on $E^*$
induced by the local coordinates $(x^i)$ and the dual basis
$\{e^{\alpha}\}$ of $\{e_{\alpha}\}$. Moreover, we have that
$$
\begin{array}{c}
\lcf\tilde{e}_{\alpha},\tilde{e}_{\beta}\rcf_{E}^{\tau_{E}^*}\kern-2pt=C_{\alpha\beta}^{\gamma}\tilde{e}_{\gamma},\;\;
\lcf\tilde{e}_{\alpha},\bar{e}_{\beta}\rcf_{E}^{\tau_{E}^*}\kern-2pt=\lcf\bar{e}_{\alpha},\bar{e}_{\beta}\rcf_{E}^{\tau_{E}^*}\kern-2pt=0,\;
\rho_{E}^{\tau_{E}^*}(\tilde{e}_{\alpha})\kern-3pt=\rho_{\alpha}^i\displaystyle\frac{\partial}{\partial
x^i},\;\rho_{E}^{\tau_{E}^*}(\bar{e}_{\alpha})\kern-3pt=\displaystyle\frac{\partial}{\partial
y_{\alpha}},
\end{array}$$
and
\begin{equation}\label{formas}
\lambda_E(x^i,y_{\alpha})=y_{\alpha}\tilde{e}^{\alpha},\;\;\;\Omega_E(x^i,y_{\alpha})=\tilde{e}^{\alpha}\wedge\bar{e}^{\alpha}+\displaystyle\frac{1}{2}C_{\alpha\beta}^{\gamma}y_{\gamma}\tilde{e}^{\alpha}\wedge\tilde{e}^{\beta},
\end{equation}
\begin{equation}\label{eq2.3'}
\Pi_{E^*}=\rho_\alpha^i\frac{\partial}{\partial
x^i}\wedge\frac{\partial}{\partial y_\alpha
}-\displaystyle\frac{1}{2}C_{\alpha\beta}^\gamma
y_\gamma\frac{\partial}{\partial
y_\alpha}\wedge\frac{\partial}{\partial y_\beta}.
\end{equation}
(For more details, see \cite{LMM,M2}).

\subsection{Lie affgebroids}\label{secaff}

Let $\tau_{\mathcal A}:{\mathcal A}\rightarrow M$ be an affine
bundle with associated vector bundle $\tau_V:V\rightarrow M$.
Denote by $\tau_{{\mathcal A}^+}:{\mathcal A}^+={{\mathcal
A}\hspace{-0.05cm}f\hspace{-0.075cm}f}({\mathcal A},\R)\rightarrow
M$ the dual bundle whose fibre over $x\in M$ consists of affine
functions on the fibre ${\mathcal A}_x$. Note that this bundle has
a distinguished section $1_{\mathcal A}\in\Gamma(\tau_{{\mathcal
A}^+})$ corresponding to the constant function $1$ on ${\mathcal
A}$. We also consider the bidual bundle
$\tau_{\widetilde{{\mathcal A}}}:\widetilde{{\mathcal
A}}\rightarrow M$ whose fibre at $x\in M$ is the vector space
$\widetilde{{\mathcal A}}_x=({\mathcal A}_x^+)^*$. Then,
${\mathcal A}$ may be identified with an affine subbundle of
$\widetilde{{\mathcal A}}$ via the inclusion $i_{\mathcal
A}:{\mathcal A}\rightarrow\widetilde{{\mathcal A}}$ given by
$i_{\mathcal A}(\a)(\varphi)=\varphi(\a)$, which is an injective
affine map whose associated linear map is denoted by
$i_V:V\rightarrow\widetilde{{\mathcal A}}$. Thus, $V$ may be
identified with a vector subbundle of $\widetilde{{\mathcal A}}$.
Using these facts, one can prove that there is a one-to-one
correspondence between affine functions on ${\mathcal A}$ and
linear functions on $\widetilde{{\mathcal A}}$. On the other hand,
there is an obvious one-to-one correspondence between affine
functions on ${\mathcal A}$ and sections of ${\mathcal A}^+$.

A {\it
 Lie affgebroid structure} on ${\mathcal A}$ consists of a Lie algebra structure
 $\lcf\cdot,\cdot\rcf_V$ on the space
 $\Gamma(\tau_V)$ of the sections of $\tau_V:V\rightarrow M$, a $\R$-linear action
 $D:\Gamma(\tau_{\mathcal A})\times\Gamma(\tau_V)\rightarrow\Gamma(\tau_V)$ of
 the sections of ${\mathcal A}$ on $\Gamma(\tau_V)$ and an affine map
 $\rho_{\mathcal A}:{\mathcal A}\rightarrow TM$, the {\it anchor map}, satisfying the following
 conditions:
  \begin{enumerate}
\item[$ \bullet$] $D_X\lcf\bar{Y},\bar{Z}\rcf_V=\lcf
D_X\bar{Y},\bar{Z}\rcf_V+\lcf\bar{Y},D_X\bar{Z}\rcf_V,$
\item[$\bullet$]
$D_{X+\bar{Y}}\bar{Z}=D_X\bar{Z}+\lcf\bar{Y},\bar{Z}\rcf_V,$
\item[$\bullet$] $D_X(f\bar{Y})=fD_X\bar{Y}+\rho_{\mathcal
A}(X)(f)\bar{Y},$
\end{enumerate}
\noindent for $X\in\Gamma(\tau_{\mathcal A})$,
$\bar{Y},\bar{Z}\in\Gamma(\tau_V)$ and $f\in C^{\infty}(M)$ (see
\cite{GGrU,MMeS}).

If $(\lcf\cdot,\cdot\rcf_V,D,\rho_{\mathcal A})$ is a Lie
affgebroid structure on an affine bundle ${\mathcal A}$ then
$(V,\lcf\cdot,\cdot\rcf_V,\rho_V)$ is a Lie algebroid, where
$\rho_V:V\rightarrow TM$ is the vector bundle map associated with
the affine morphism $\rho_{\mathcal A}:{\mathcal A}\rightarrow
TM$.

A Lie affgebroid structure on an affine bundle $\tau_{\mathcal
A}:{\mathcal A}\rightarrow M$ induces a Lie algebroid structure
$(\lcf\cdot,\cdot\rcf_{\widetilde{{\mathcal
A}}},\rho_{\widetilde{{\mathcal A}}})$ on the bidual bundle
$\widetilde{{\mathcal A}}$ such that $1_{\mathcal
A}\in\Gamma(\tau_{{\mathcal A}^+})$ is a $1$-cocycle in the
corresponding Lie algebroid cohomology, that is,
$d^{\widetilde{{\mathcal A}}}1_{\mathcal A}=0$. Indeed, if
$X_0\in\Gamma(\tau_{\mathcal A})$ then for every section
$\widetilde{X}$ of $\widetilde{{\mathcal A}}$ there exists a
function $f\in C^{\infty}(M)$ and a section
$\bar{X}\in\Gamma(\tau_V)$ such that $\widetilde{X}=fX_0+\bar{X}$
and
$$\begin{array}{rcl}
\rho_{\widetilde{{\mathcal A}}}(fX_0+\bar{X})&=&f\rho_{\mathcal A}(X_0)+\rho_V(\bar{X}),\\
\lcf fX_0+\bar{X},gX_0+\bar{Y}\rcf_{\widetilde{{\mathcal
A}}}&=&(\rho_V(\bar{X})(g)-\rho_V(\bar{Y})(f)+f\rho_{\mathcal
A}(X_0)(g)\\&&-g\rho_{\mathcal A}(X_0)(f))X_0
+\lcf\bar{X},\bar{Y}\rcf_V+fD_{X_0}\bar{Y}-gD_{X_0}\bar{X}.
\end{array}$$

Conversely, let $(U,\lcf\cdot,\cdot\rcf_U,\rho_U)$ be a Lie
algebroid over $M$ and $\phi:U\rightarrow\R$ be a $1$-cocycle of
$(U,\lcf\cdot,\cdot\rcf_U,\rho_U)$ such that $\phi_{|U_x}\neq 0$,
for all $x\in M$. Then, ${\mathcal A}=\phi^{-1}\{1\}$  is an
affine bundle over $M$ which admits a Lie affgebroid structure in
such a way that $(U,\lcf\cdot,\cdot\rcf_U,\rho_U)$ may be
identified with the bidual Lie algebroid $(\widetilde{{\mathcal
A}},\lcf\cdot,\cdot\rcf_{\widetilde{{\mathcal
A}}},\rho_{\widetilde{{\mathcal A}}})$ to ${\mathcal A}$ and,
under this identification, the $1$-cocycle $1_{\mathcal
A}:\widetilde{{\mathcal A}}\rightarrow\R$ is just $\phi$. The
affine bundle $\tau_{\mathcal A}:{\mathcal A}\rightarrow M$ is
modelled on the vector bundle $\tau_V:V=\phi^{-1}\{0\}\rightarrow
M$. In fact, if $i_V:V\rightarrow U$ and $i_{\mathcal A}:{\mathcal
A}\rightarrow U$ are the canonical inclusions, then
$$\begin{array}{rcl} i_V\circ\lcf\bar{X},\bar{Y}\rcf_V&=&\lcf
i_V\circ\bar{X},i_V\circ\bar{Y}\rcf_U,\;\;i_V\circ D_X\bar{Y}=\lcf
i_{\mathcal A}\circ X,i_V\circ\bar{Y}\rcf_U,\\
\rho_{\mathcal A}(X)&=&\rho_{U}(i_{\mathcal A}\circ X),
\end{array}$$
 for $\bar{X},\bar{Y}\in\Gamma(\tau_V)$ and $X\in\Gamma(\tau_{\mathcal A}).$

Let $\tau_{\mathcal A}:{\mathcal A}\to M$ be a Lie affgebroid
modelled on the Lie algebroid $\tau_V:V\to M$. Suppose that
$(x^i)$ are local coordinates on an open subset $U$ of $M$ and
that $\{e_0,e_{\alpha}\}$ is a local basis of sections of
$\tau_{\widetilde{{\mathcal A}}}:\widetilde{{\mathcal A}}\to M$ in
$U$ which is adapted to the $1$-cocycle $1_{\mathcal A}$, i.e.,
such that $1_{\mathcal A}(e_0)=1$ and $1_{\mathcal
A}(e_{\alpha})=0,$ for all $\alpha.$ Note that if
$\{e^0,e^{\alpha}\}$ is the dual basis of $\{e_0,e_{\alpha}\}$
then $e^0=1_{\mathcal A}$. Denote by $(x^i,y^0,y^{\alpha})$ the
corresponding local coordinates on $\widetilde{{\mathcal A}}$.
Then, the local equation defining the affine subbundle ${\mathcal
A}$ (respectively, the vector subbundle $V$) of
$\widetilde{{\mathcal A}}$ is $y^0=1$ (respectively, $y^0=0$).
Thus, $(x^i,y^{\alpha})$ may be considered as local coordinates on
${\mathcal A}$ and $V$.

The standard example of a Lie affgebroid may be constructed as
follows. Let $\tau:M\to\R$ be a fibration and
$\tau_{1,0}:J^1\tau\to M$ be the $1$-jet bundle of local sections
of $\tau:M\to\R$. It is well known that $\tau_{1,0}:J^1\tau\to M$
is an affine bundle modelled on the vector bundle
$\pi=(\pi_M)_{|V\tau}:V\tau\to M$, where $V\tau$ is the vertical
bundle of $\tau:M\to\R$. Moreover, if $t$ is the usual coordinate
on $\R$ and $\eta$ is the closed $1$-form on $M$ given by
$\eta=\tau^*(dt)$ then we have the following identification
$J^1\tau\cong\{v\in TM/\eta(v)=1\}$ (see, for instance,
\cite{Sa}). Note that $V\tau=\{v\in TM/\eta(v)=0\}.$ Thus, the
bidual bundle $\widetilde{J^1\tau}$ to the affine bundle
$\tau_{1,0}:J^1\tau\to M$ may be identified with the tangent
bundle $TM$ to $M$ and, under this identification, the Lie
algebroid structure on $\pi_M:TM\to M$ is the standard Lie
algebroid structure and the $1$-cocycle $1_{J^1\tau}$ on
$\pi_M:TM\to M$ is just the $1$-form $\eta$.

\subsection{The Lagrangian formalism on Lie
affgebroids}\label{Section2.4}

In this section, we will develop a geometric framework, which
allows to write the Euler-Lagrange equations associated with a
Lagrangian function $L$ on a Lie affgebroid ${\mathcal A}$ in an
intrinsic way (see \cite{MMeS}).

Suppose that $(\tau_{\mathcal A}:{\mathcal A}\rightarrow M,
\tau_V:V\rightarrow M,(\lcf\cdot,\cdot\rcf_V,D,\rho_{\mathcal
A}))$  is a Lie affgebroid over $M$. Then, the bidual bundle
$\tau_{\widetilde{{\mathcal A}}}:\widetilde{{\mathcal A}}\to M$ to
${\mathcal A}$ admits a Lie algebroid structure
$(\lcf\cdot,\cdot\rcf_{\widetilde{{\mathcal
A}}},\rho_{\widetilde{{\mathcal A}}})$ in such a way that the
section $1_{\mathcal A}$ of the dual bundle ${\mathcal A}^+$ is a
$1$-cocycle.

Now, we consider the Lie algebroid prolongation (${\mathcal
T}^{\widetilde{{\mathcal A}}}{\mathcal A},$
$\lcf\cdot,\cdot\rcf_{\widetilde{{\mathcal A}}}^{\tau_{\mathcal
A}},\rho_{\widetilde{{\mathcal A}}}^{\tau_{\mathcal A}})$ of the
Lie algebroid $(\widetilde{{\mathcal
A}},\lcf\cdot,\cdot\rcf_{\widetilde{{\mathcal A}}},$ $
\rho_{\widetilde{{\mathcal A}}})$ over the fibration
$\tau_{{\mathcal A}}:{\mathcal A}\rightarrow M$ with vector bundle
projection $\tau^{\tau_{{\mathcal A}}}_{\widetilde{{\mathcal
A}}}:{\mathcal T}^{\widetilde{{\mathcal A}}}{\mathcal
A}\rightarrow {\mathcal A}$.

If $(x^i)$ are local coordinates on an open subset $U$ of $M$ and
$\{e_0,e_{\alpha}\}$ is a local basis of sections of the vector
bundle $\tau^{-1}_{\widetilde{{\mathcal A}}}(U)\rightarrow U$
adapted to $1_{\mathcal A}$, then
$\{\tilde{T}_0,\tilde{T}_{\alpha},\tilde{V}_{\alpha}\}$ is a local
basis of sections of the vector bundle $(\tau^{\tau_{{\mathcal
A}}}_{\widetilde{{\mathcal A}}})^{-1}(\tau^{-1}_{{{\mathcal
A}}}(U))\rightarrow\tau^{-1}_{{{\mathcal A}}}(U)$, where
\begin{equation}\label{tildeT}
\tilde{T}_0(\a)=\Big(e_0(\tau_{\mathcal
A}(\a)),\rho_0^i\displaystyle\frac{\partial}{\partial
x^i}_{|\a}\Big),\;\;\;
\tilde{T}_{\alpha}(\a)=\Big(e_{\alpha}(\tau_{\mathcal
A}(\a)),\rho_{\alpha}^i\displaystyle\frac{\partial}{\partial
x^i}_{|\a}\Big),\;\;\;
\tilde{V}_{\alpha}(\a)=\Big(0,\displaystyle\frac{\partial}{\partial
y^{\alpha}}_{|\a}\Big),
\end{equation}
$(x^i,y^{\alpha})$ are the local coordinates on ${\mathcal A}$
induced by the local coordinates $(x^i)$ and the basis
$\{e_{\alpha}\}$ and $\rho_0^i,\;\rho_{\alpha}^i$ are the
components of the anchor map $\rho_{\widetilde{{\mathcal A}}}$. We
also have that
\begin{equation}\label{corTVtilde}
\begin{array}{c}
\lcf\tilde{T}_0,\tilde{T}_{\alpha}\rcf_{\widetilde{{\mathcal
A}}}^{\tau_{{\mathcal
A}}}=C_{0\alpha}^{\gamma}\tilde{T}_{\gamma},\;\;\;
\lcf\tilde{T}_{\alpha},\tilde{T}_{\beta}\rcf_{\widetilde{{\mathcal A}}}^{\tau_{{\mathcal A}}}=C_{\alpha\beta}^{\gamma}\tilde{T}_{\gamma},\\[8pt]
\lcf\tilde{T}_0,\tilde{V}_{\alpha}\rcf_{\widetilde{{\mathcal A}}}^{\tau_{{\mathcal A}}}=\lcf\tilde{T}_{\alpha},\tilde{V}_{\beta}\rcf_{\widetilde{{\mathcal A}}}^{\tau_{{\mathcal A}}}=\lcf\tilde{V}_{\alpha},\tilde{V}_{\beta}\rcf_{\widetilde{{\mathcal A}}}^{\tau_{{\mathcal A}}}=0,\\[10pt]
\rho_{\widetilde{{\mathcal A}}}^{\tau_{{\mathcal
A}}}(\tilde{T}_0)=\rho_0^i\displaystyle\frac{\partial}{\partial
x^i},\;\;\;\rho_{\widetilde{{\mathcal A}}}^{\tau_{{\mathcal
A}}}(\tilde{T}_{\alpha})=\rho_{\alpha}^i\displaystyle\frac{\partial}{\partial
x^i},\;\;\;\rho_{\widetilde{{\mathcal A}}}^{\tau_{{\mathcal
A}}}(\tilde{V}_{\alpha})=\displaystyle\frac{\partial}{\partial
y^{\alpha}},
\end{array}
\end{equation}
where $C_{0\beta}^{\gamma}$ and $C_{\alpha\beta}^{\gamma}$ are the
structure functions of the Lie bracket
$\lcf\cdot,\cdot\rcf_{\widetilde{{\mathcal A}}}$ with respect to
the basis $\{e_0,e_{\alpha}\}$.  Note that, if
$\{\tilde{T}^0,\tilde{T}^{\alpha},\tilde{V}^{\alpha}\}$ is the
dual basis of
$\{\tilde{T}_0,\tilde{T}_{\alpha},\tilde{V}_{\alpha}\}$, then
$\tilde{T}^0$ is globally defined and it is a $1$-cocycle. We will
denote by $\phi_0$ the $1$-cocycle $\tilde{T}^0$. Thus, we have
that
$$\phi_0(\a)(\tilde{\b},X_\a)=1_{\mathcal A}(\tilde{\b}),\mbox{ for
}(\tilde{\b},X_\a)\in{\mathcal  T}^{\widetilde{{\mathcal
A}}}_\a{\mathcal A}.$$

One may also consider {\it the vertical endomorphism} $S:{\mathcal
T}^{\widetilde{{\mathcal A}}}{\mathcal A}\rightarrow{\mathcal
T}^{\widetilde{{\mathcal A}}}{\mathcal A}$, as a section of the
vector bundle ${\mathcal T}^{\widetilde{{\mathcal A}}}{\mathcal
A}\otimes({\mathcal T}^{\widetilde{{\mathcal A}}}{\mathcal
A})^*\to {\mathcal A}$, whose local expression is (see
\cite{MMeS})
\begin{equation}\label{verend}
S=(\tilde{T}^{\alpha}-y^{\alpha}\phi_0)\otimes\tilde{V}_{\alpha}.
\end{equation}

A section $\xi$ of $\tau_{\widetilde{{\mathcal
A}}}^{\tau_{\mathcal A}}:{\mathcal T}^{\widetilde{{\mathcal
A}}}{\mathcal A}\to {\mathcal A}$ is said to be a {\it second
order differential equation} (SODE) on ${\mathcal A}$ if
$\phi_0(\xi)=1$ and $S\,\xi=0$. If
$\xi\in\Gamma(\tau_{\widetilde{{\mathcal A}}}^{\tau_{\mathcal
A}})$ is a SODE then
$\xi=\tilde{T}_0+y^{\alpha}\tilde{T}_{\alpha}+\xi^{\alpha}\tilde{V}_{\alpha},$
where $\xi^{\alpha}$ are  local functions on ${\mathcal A}$, and
$$\rho^{\tau_{\mathcal A}}_{\widetilde{{\mathcal A}}}(\xi)=(\rho_0^i+y^{\alpha}\rho_{\alpha}^i)\displaystyle\frac{\partial}{\partial
x^i}+\xi^{\alpha}\displaystyle\frac{\partial}{\partial
y^{\alpha}}.$$

Now, a curve $\gamma:I\subseteq\R\rightarrow {\mathcal A}$ in
${\mathcal A}$ is said to be {\it admissible} if
$\rho_{\widetilde{{\mathcal A}}}\circ i_{\mathcal
A}\circ\gamma=\dot{\widehat{(\tau_{\mathcal A}\circ\gamma)}}$ or,
equivalently, $(i_{\mathcal
A}(\gamma(t)),\dot{\gamma}(t))\in{\mathcal
T}^{\widetilde{{\mathcal A}}}_{\gamma(t)}{\mathcal A}$, for all
$t\in I$, $i_{\mathcal A}:{\mathcal
A}\rightarrow\widetilde{{\mathcal A}}$ being the canonical
inclusion. We will denote by $Adm({\mathcal A})$ the space of
admissible curves on ${\mathcal A}$. Thus, if
$\gamma(t)=(x^i(t),y^{\alpha}(t)),$ for all $t\in I$, then
$\gamma$ is an admissible curve if and only if
$$\displaystyle\frac{dx^i}{dt}=\rho_0^i+\rho_{\alpha}^iy^{\alpha},\makebox[1cm]{for}i\in\{1,\dots,m\}.$$

It is clear that if $\xi$ is a SODE then the integral curves of
the vector field $\rho^{\tau_{\mathcal A}}_{\widetilde{{\mathcal
A}}}(\xi)$ are admissible.

On the other hand, let $L:{\mathcal A}\rightarrow\R$ be a
Lagrangian function. Then, we introduce {\it the
Poincar\'{e}-Cartan $1$-section}
$\Theta_L\in\Gamma((\tau^{\tau_{\mathcal A}}_{\widetilde{{\mathcal
A}}})^*)$ and {\it the Poincar\'{e}-Cartan $2$-section}
$\Omega_L\in\Gamma(\wedge^2(\tau_{\widetilde{{\mathcal
A}}}^{\tau_{\mathcal A}})^*)$ a\-sso\-cia\-ted with $L$ defined by
\begin{equation}\label{PC}
\Theta_L=L\phi_0+(d^{{\mathcal  T}^{\widetilde{{\mathcal
A}}}{\mathcal A}}L)\circ S,\;\;\;\Omega_L=-d^{{\mathcal
T}^{\widetilde{{\mathcal A}}}{\mathcal A}}\Theta_L.
\end{equation}

From (\ref{corTVtilde}), (\ref{verend}) and (\ref{PC}), we obtain
that
\begin{equation}\label{PCloc}
\begin{array}{lcl}
\Theta_L&=&(L-y^{\alpha}\displaystyle\frac{\partial L}{\partial
y^{\alpha}})\phi_0+\displaystyle\frac{\partial L}{\partial
y^{\alpha}}\tilde{T}^{\alpha},\\[8pt]
\Omega_L&=&\Big(i_{\xi_0}(d^{{\mathcal T}^{\widetilde{{\mathcal
A}}}{\mathcal A}}(\displaystyle\frac{\partial L}{\partial
y^{\alpha}}))-\displaystyle\frac{\partial L}{\partial
y^{\gamma}}(C_{0\alpha}^{\gamma}+C_{\beta\alpha}^{\gamma}y^{\beta})-\rho_{\alpha}^i\displaystyle\frac{\partial
L}{\partial
x^i}\Big)\theta^{\alpha}\wedge\phi_0\\[8pt]
&+&\displaystyle\frac{\partial^2 L}{\partial y^{\alpha}\partial
y^{\beta}}\theta^{\alpha}\wedge\psi^{\beta}
+\displaystyle\frac{1}{2}\Big(\rho_{\beta}^i\displaystyle\frac{\partial^2
L}{\partial x^i\partial
y^{\alpha}}-\rho_{\alpha}^i\displaystyle\frac{\partial^2
L}{\partial x^i\partial y^{\beta}}+\displaystyle\frac{\partial
L}{\partial
y^{\gamma}}C_{\alpha\beta}^{\gamma}\Big)\theta^{\alpha}\wedge\theta^{\beta},
\end{array}
\end{equation}
where $\theta^{\alpha}=\tilde{T}^{\alpha}-y^{\alpha}\phi_0$,
$\psi^{\alpha}=\tilde{V}^{\alpha}-\xi_0^{\alpha}\phi_0$ and
$\xi_0=\tilde{T}_0+y^{\alpha}\tilde{T}_{\alpha}+\xi_0^{\alpha}\tilde{V}_{\alpha}$
is an arbitrary SODE.

Now, a curve $\gamma:I=(-\epsilon,\epsilon)\subseteq\R\rightarrow
{\mathcal A}$ in ${\mathcal A}$ is {\it a solution of the
Euler-Lagrange equations} associated with $L$ if and only if
$\gamma$ is admissible and $i_{(i_{\mathcal
A}(\gamma(t)),\dot{\gamma}(t))}\Omega_L(\gamma(t))=0$, for all
$t$.

If $\gamma(t)=(x^i(t),y^{\alpha}(t))$ then $\gamma$ is a solution
of the Euler-Lagrange equations if and only if
\begin{equation}\label{EL}
\displaystyle\frac{dx^i}{dt}=\rho_0^i+\rho_{\alpha}^iy^{\alpha},\;\;\;\displaystyle\frac{d}{dt}\left(\displaystyle\frac{\partial
L}{\partial
y^{\alpha}}\right)=\rho_{\alpha}^i\displaystyle\frac{\partial
L}{\partial
x^i}+(C_{0\alpha}^{\gamma}+C_{\beta\alpha}^{\gamma}y^{\beta})\displaystyle\frac{\partial
L}{\partial y^{\gamma}},
\end{equation}
for $i\in\{1,\dots,m\}$ and $\alpha\in\{1,\dots,n\}$.

If we denote by ${\mathcal C}({\mathcal A}^+)$ the set of curves
in ${\mathcal A}^+$, we can define the {\it Euler-Lagrange
operator} $\delta L:Adm({\mathcal A})\to {\mathcal C}({\mathcal
A}^+)$ by
\begin{equation}\label{E-Loperator}\delta
L_{\gamma(t)}(\tilde{\a})=\Omega_L(\gamma(t))((i_{\mathcal
A}(\gamma(t)),\dot{\gamma}(t)),(\tilde{\a},v_{\gamma(t)})),
\end{equation}
for $\gamma\in Adm({\mathcal A})$ and
$\tilde{\a}\in\widetilde{{\mathcal A}}_{\tau_{\mathcal
A}(\gamma(t))}$, where $v_{\gamma(t)}\in T_{\gamma(t)}{\mathcal
A}$ is such that $(\tilde{\a},v_{\gamma(t)})\in{\mathcal
T}^{\widetilde{{\mathcal A}}}_{\gamma(t)}{\mathcal A}$, for all
$t$. It is easy to prove that the map $\delta L$ doesn't depend on
the chosen tangent vector $v_{\gamma(t)}$. From (\ref{PCloc}), we
deduce that its local expression is
$$\delta L=\Big(-\displaystyle\frac{d}{dt}\left(\displaystyle\frac{\partial
L}{\partial
y^{\alpha}}\right)+\rho_{\alpha}^i\displaystyle\frac{\partial
L}{\partial
x^i}+(C_{0\alpha}^{\gamma}+C_{\beta\alpha}^{\gamma}y^{\beta})\displaystyle\frac{\partial
L}{\partial y^{\gamma}}\Big)(e^\alpha-y^\alpha e^0),$$ where
$\{e^0,e^\alpha\}$ is the dual basis of $\{e_0,e_\alpha\}$. Then
the Euler-Lagrange differential equations read as $\delta L=0.$

The Lagrangian $L$ is {\it regular} if and only if the matrix
$(W_{\alpha\beta})=\Big(\displaystyle\frac{\partial^2L}{\partial
y^{\alpha}\partial y^{\beta}}\Big)$ is regular or, in a intrinsic
way, if  the pair $(\Omega_L,\phi_0)$ is a cosymplectic structure
on ${\mathcal T}^{\widetilde{{\mathcal A}}}{\mathcal A}$, that is,
\begin{eqnarray*}
\Big\{ \underbrace{\Omega_L\wedge \ldots \wedge
\Omega_L}_{(n}\wedge \phi_0 \Big\}(\a)\not=0,\quad d^{{\mathcal
T}^{\widetilde{{\mathcal A}}}{\mathcal A}}\Omega_L=0\;\;\mbox{ and
}\;\; d^{{\mathcal T}^{\widetilde{{\mathcal A}}}{\mathcal
A}}\phi_0=0,\mbox{ for all }a\in{\mathcal A}.
\end{eqnarray*}
Note that the first condition is equivalent to the fact that the
map $\flat_L:{\mathcal T}^{\widetilde{{\mathcal A}}}{\mathcal
A}\to({\mathcal T}^{\widetilde{{\mathcal A}}}{\mathcal A})^*$
defined by
$$\flat_L(X)=i_X\Omega_L+\phi_0(X)\phi_0$$
is an isomorphism of vector bundles.

Now, suppose that the Lagrangian $L$ is regular and let
$R_L\in\Gamma(\tau_{\widetilde{{\mathcal A}}}^{\tau_{\mathcal
A}})$ be the Reeb section of the cosymplectic structure
$(\Omega_L,\phi_0)$ characterized by the following conditions
$$i_{R_L}\Omega_L=0\makebox[1cm]{and}i_{R_L}\phi_0=1.$$

Then, $R_L$ is the unique Lagrangian SODE associated with $L$,
that is, the integral curves of the vector field
$\rho^{\tau_{\mathcal A}}_{\widetilde{{\mathcal A}}}(R_L)$ are
solutions of the Euler-Lagrange equations associated with $L$. In
such a case, $R_L$ is called {\it the Euler-Lagrange section
associated with $L$} and its local expression is
$$R_L=\tilde{T}_0+y^{\alpha}\tilde{T}_{\alpha}+W^{\alpha\beta}\Big(\rho_{\beta}^i\displaystyle\frac{\partial
L}{\partial
x^i}-(\rho_0^i+y^{\gamma}\rho_{\gamma}^i)\frac{\partial^2L}{\partial
x^i\partial
y^{\beta}}+(C_{0\beta}^{\gamma}+y^{\mu}C_{\mu\beta}^{\gamma})\frac{\partial
L}{\partial y^{\gamma}}\Big)\tilde{V}_{\alpha},$$
where $(W^{\alpha\beta})$ is the inverse matrix of
$(W_{\alpha\beta})$.

\subsection{The Hamiltonian formalism}\label{sec3.1}

In this section, we will develop a geometric framework, which
allows to write the Hamilton equations associated with a
Hamiltonian section on a Lie affgebroid (see \cite{IMPS,M}).

Suppose that $\left(\tau_{\mathcal A}:{\mathcal A}\rightarrow M,
\tau_V:V\rightarrow M, (\lcf\cdot,\cdot\rcf_V,D,\rho_{\mathcal
A})\right)$ is a Lie affgebroid. Now, consider the prolongation
$({\mathcal T}^{\widetilde{{\mathcal
A}}}V^*,\lcf\cdot,\cdot\rcf_{\widetilde{{\mathcal
A}}}^{\tau_{V}^*},\rho_{\widetilde{{\mathcal A}}}^{\tau_{V}^*})$
of the bidual Lie algebroid $(\widetilde{{\mathcal
A}},\lcf\cdot,\cdot\rcf_{\widetilde{{\mathcal A}}},$ $
\rho_{\widetilde{{\mathcal A}}})$ over the fibration
$\tau_{V}^*:V^*\rightarrow M$.

Let $(x^i)$ be local coordinates on an open subset $U$ of $M$ and
$\{e_0,e_{\alpha}\}$ be a local basis of sections of the vector
bundle $\tau^{-1}_{\widetilde{{\mathcal A}}}(U)\rightarrow U$
adapted to $1_{\mathcal A}$  and
$$\begin{array}{rclrclccrclrcl}
\lcf e_0,e_{\alpha}\rcf_{\widetilde{{\mathcal
A}}}=C_{0\alpha}^{\gamma}e_{\gamma},\;\;\lcf
e_{\alpha},e_{\beta}\rcf_{\widetilde{{\mathcal
A}}}=C_{\alpha\beta}^{\gamma}e_{\gamma},\;\;
\rho_{\widetilde{{\mathcal
A}}}(e_0)=\rho_0^i\displaystyle\frac{\partial}{\partial x^i},\;\;
\rho_{\widetilde{{\mathcal
A}}}(e_{\alpha})=\rho_{\alpha}^i\displaystyle\frac{\partial}{\partial
x^i}.
\end{array}
$$

Denote by $(x^i,y^0,y^{\alpha})$ the induced local coordinates on
$\widetilde{{\mathcal A}}$ and by $(x^i,y_0,y_{\alpha})$ the dual
coordinates on the dual vector bundle $\tau_{{\mathcal
A}^+}:{\mathcal A}^+\to M$ to $\widetilde{{\mathcal A}}$. Then,
$(x^i,y_{\alpha})$ are local coordinates on $V^*$ and
$\{\tilde{e}_0,\tilde{e}_{\alpha},\bar{e}_{\alpha}\}$ is a local
basis of sections of the vector bundle
$\tau^{\tau_{V}^*}_{\widetilde{{\mathcal A}}}:{\mathcal
T}^{\widetilde{{\mathcal A}}}V^*\rightarrow V^*$, where
$$\tilde{e}_0(\psi)=\Big(e_0(\tau_V^*(\psi)),\rho_0^i\displaystyle\frac{\partial}{\partial
x^i}_{|\psi}\Big),\;\;\tilde{e}_{\alpha}(\psi)=\Big(e_{\alpha}(\tau_V^*(\psi)),\rho_{\alpha}^i\displaystyle\frac{\partial}{\partial
x^i}_{|\psi}\Big),\;\;\bar{e}_{\alpha}(\psi)=\Big(0,\displaystyle\frac{\partial}{\partial
y_{\alpha}}_{|\psi}\Big).$$ Using this local basis one may
introduce local coordinates $(x^i,y_{\alpha};z^0,$
$z^{\alpha},v_{\alpha})$ on ${\mathcal  T}^{\widetilde{{\mathcal
A}}}V^*$. A direct computation proves that
$$
\begin{array}{l}
\lcf\tilde{e}_0,\tilde{e}_\beta\rcf_{\widetilde{{\mathcal
A}}}^{\tau_V^*}=C_{0\beta}^\gamma\tilde{e}_\gamma,\;\;\;
\lcf\tilde{e}_\alpha,\tilde{e}_\beta\rcf_{\widetilde{{\mathcal A}}}^{\tau_V^*}=C_{\alpha\beta}^\gamma\tilde{e}_\gamma,\\[8pt]
\lcf\tilde{e}_0,\bar{e}_\alpha\rcf_{\widetilde{{\mathcal
A}}}^{\tau_V^*}=\lcf\tilde{e}_\alpha,\bar{e}_\beta\rcf_{\widetilde{{\mathcal
A}}}^{\tau_V^*}=\lcf \bar{e}_\alpha,
\bar{e}_\beta\rcf_{\widetilde{{\mathcal A}}}^{\tau_V^*}=0,\\[8pt]
\rho_{\widetilde{{\mathcal
A}}}^{\tau_V^*}(\tilde{e}_0)=\rho_0^i\displaystyle\frac{\partial
}{\partial x^i},\;\;\;\rho_{\widetilde{{\mathcal
A}}}^{\tau_V^*}(\tilde{e}_\alpha)=\rho_\alpha^i\displaystyle\frac{\partial
}{\partial x^i},\;\;\; \rho_{\widetilde{{\mathcal
A}}}^{\tau_V^*}(\bar{e}_\alpha)=\displaystyle\frac{\partial
}{\partial y_\alpha},
\end{array}
$$
for all $\alpha$ and $\beta$. Thus, if
$\{\tilde{e}^0,\tilde{e}^\alpha,\bar{e}^\alpha\}$ is the dual
basis of $\{\tilde{e}_0,\tilde{e}_\alpha,\bar{e}_\alpha\}$ then
$$\begin{array}{rcl} d^{{\mathcal T}^{\widetilde{{\mathcal
A}}}V^*}f&=& \rho_0^i\displaystyle\frac{\partial f}{\partial
x^i}\tilde{e}^0+ \rho_\alpha^i\displaystyle\frac{\partial
f}{\partial x^i}\tilde{e}^\alpha + \displaystyle\frac{\partial
f}{\partial
y_\alpha} \bar{e}^\alpha,\\[8pt]
d^{{\mathcal T}^{\widetilde{{\mathcal
A}}}V^*}\tilde{e}^\gamma&=&\displaystyle
-\frac{1}{2}C_{0\alpha}^\gamma\tilde{e}^0\wedge\tilde{e}^\alpha
\displaystyle -\frac{1}{2} C_{\alpha\beta}^{\gamma}
\tilde{e}^\alpha \wedge \tilde{e}^\beta, \makebox[0.75cm]{}
d^{{\mathcal T}^{\widetilde{{\mathcal
A}}}V^*}\tilde{e}^0=d^{{\mathcal T}^{\widetilde{{\mathcal
A}}}V^*}\bar{e}^\gamma=0,
\end{array}$$
for $f\in C^\infty(V^*).$

Let $\mu:{\mathcal A}^+\rightarrow V^*$ be the canonical
projection given by $\mu(\varphi)=\varphi^l$, for $\varphi\in
{\mathcal A}^+_x$, with $x\in M$, where $\varphi^l\in V^*_x$ is
the linear map associated with the affine map $\varphi$ and
$h:V^*\rightarrow {\mathcal A}^+$ be a {\it Hamiltonian section}
of $\mu$.

Now, we consider the Lie algebroid prolongation ${\mathcal
T}^{\widetilde{{\mathcal A}}}{\mathcal A}^+$ of the Lie algebroid
$\widetilde{{\mathcal A}}$ over $\tau_{{\mathcal A}^+}:{\mathcal
A}^+\to M$ with vector bundle projection $\tau^{\tau_{{\mathcal
A}^+}}_{\widetilde{{\mathcal A}}}:{\mathcal
T}^{\widetilde{{\mathcal A}}}{\mathcal A}^+\rightarrow {\mathcal
A}^+$ (see Section \ref{sec1.1.1}). Then, we may introduce the map
${\mathcal T}h:{\mathcal T}^{\widetilde{{\mathcal
A}}}V^*\rightarrow{\mathcal T}^{\widetilde{{\mathcal A}}}{\mathcal
A}^+$ defined by ${\mathcal
T}h(\tilde{\a},X_{\alpha})=(\tilde{\a},(T_{\alpha}h)(X_{\alpha})),$
for $(\tilde{\a},X_{\alpha})\in {\mathcal
T}_\alpha^{\widetilde{{\mathcal A}}}V^*$, with $\alpha\in V^*.$ It
is easy to prove that the pair $({\mathcal  T}h,h)$ is a Lie
algebroid morphism between the Lie algebroids
$\tau_{\widetilde{{\mathcal A}}}^{\tau_V^*}:{\mathcal
T}^{\widetilde{{\mathcal A}}}V^*\rightarrow V^*$ and
$\tau_{\widetilde{{\mathcal A}}}^{\tau_{{\mathcal A}^+}}:{\mathcal
T}^{\widetilde{{\mathcal A}}}{\mathcal A}^+\rightarrow {\mathcal
A}^+$.

\noindent Next, denote by $\lambda_h$ and $\Omega_h$ the sections
of the vector bundles $({\mathcal T}^{\widetilde{{\mathcal
A}}}V^*)^*\rightarrow V^*$ and $\Lambda^2({\mathcal
T}^{\widetilde{{\mathcal A}}}V^*)^*\rightarrow V^*$ given by
\begin{equation}\label{Omegah}
\lambda_h=({\mathcal T}h,h)^*(\lambda_{\widetilde{{\mathcal
A}}}),\;\;\Omega_h=({\mathcal
T}h,h)^*(\Omega_{\widetilde{{\mathcal A}}}),
\end{equation}
where $\lambda_{\widetilde{{\mathcal A}}}$ and
$\Omega_{\widetilde{{\mathcal A}}}$ are the Liouville section and
the canonical symplectic section, respectively, associated with
the Lie algebroid $\widetilde{{\mathcal A}}.$ Note that
$\Omega_h=-d^{{\mathcal T}^{\widetilde{{\mathcal
A}}}V^*}\lambda_h.$

On the other hand, let $\eta:{\mathcal T}^{\widetilde{{\mathcal
A}}}V^*\rightarrow\R$ be the section of $({\mathcal
T}^{\widetilde{{\mathcal A}}}V^*)^*\rightarrow V^*$ defined by
\begin{equation}\label{eta}
\eta(\tilde{\a},X_{\nu})=1_{\mathcal A}(\tilde{\a}),
\end{equation}
 for
$(\tilde{\a},X_{\nu})\in {\mathcal T}_\nu^{\widetilde{{\mathcal
A}}}V^*$, with $\nu\in V^*$. Note that if $pr_1:{\mathcal
T}^{\widetilde{{\mathcal A}}}V^*\to \widetilde{{\mathcal A}}$ is
the canonical projection on the first factor then
$(pr_1,\tau_V^*)$ is a morphism between the Lie algebroids
$\tau_{\widetilde{{\mathcal A}}}^{\tau_V^*}:{\mathcal
T}^{\widetilde{{\mathcal A}}}V^*\rightarrow V^*$ and
$\tau_{\widetilde{{\mathcal A}}}:\widetilde{{\mathcal A}}\to M$
and $(pr_1,\tau_V^*)^*(1_{\mathcal A})=\eta$. Thus, since
$1_{\mathcal A}$ is a $1$-cocycle of $\tau_{\widetilde{{\mathcal
A}}}:\widetilde{{\mathcal A}}\rightarrow M$, we deduce that $\eta$
is a $1$-cocycle of the Lie algebroid $\tau_{\widetilde{{\mathcal
A}}}^{\tau_V^*}:{\mathcal T}^{\widetilde{{\mathcal
A}}}V^*\rightarrow V^*.$

Suppose that
$h(x^i,y_{\alpha})=(x^i,-H(x^j,y_{\beta}),y_{\alpha})$ and that
$\{\tilde{e}^0,\tilde{e}^{\alpha},\bar{e}^{\alpha}\}$ is the dual
basis of $\{\tilde{e}_0,\tilde{e}_{\alpha},$ $\bar{e}_{\alpha}\}$.
Then $\eta=\tilde{e}^0$ and, from (\ref{formas}), (\ref{Omegah})
and the definition of the map ${\mathcal T}h$, it follows that
\begin{equation}\label{Omegahloc}
\Omega_h=\tilde{e}^{\alpha}\wedge\bar{e}^{\alpha}+\frac{1}{2}C_{\alpha\beta}^{\gamma}y_{\gamma}\tilde{e}^{\alpha}\wedge\tilde{e}^{\beta}+(\rho_{\alpha}^i\frac{\partial
H}{\partial
x^i}-C_{0\alpha}^{\gamma}y_{\gamma})\tilde{e}^{\alpha}\wedge\tilde{e}^0+\frac{\partial
H }{\partial y_{\alpha}}\bar{e}^{\alpha}\wedge\tilde{e}^0.
\end{equation}

Thus, it is easy to prove that the pair $(\Omega_h,\eta)$ is a
cosymplectic structure on the Lie algebroid
$\tau_{\widetilde{{\mathcal A}}}^{\tau_V^*}:{\mathcal
T}^{\widetilde{{\mathcal A}}}V^*\rightarrow V^*$.

\begin{remark} {\em Let ${\mathcal  T}^VV^*$ be the prolongation of
the Lie algebroid $V$ over the projection $\tau_V^*:V^*\to M$.
Denote by $\lambda_V$ and $\Omega_V$ the Liouville section and the
canonical symplectic section, respectively, of $V$ and by
$(i_V,Id):{\mathcal  T}^VV^*\to{\mathcal T}^{\widetilde{{\mathcal
A}}}V^*$ the canonical inclusion. Then, using (\ref{lambdaE}),
(\ref{Omegah}), (\ref{eta}) and the fact that $\mu\circ h=Id$, we
obtain that
$$(i_V,Id)^*(\lambda_h)=\lambda_V,\;\;\;(i_V,Id)^*(\eta)=0.$$
Thus, since $(i_V,Id)$ is a Lie algebroid morphism, we also deduce
that
$$(i_V,Id)^*(\Omega_h)=\Omega_V.$$\hfill$\diamondsuit$}
\end{remark}

Now, let $R_h\in\Gamma(\tau_{\widetilde{{\mathcal
A}}}^{\tau_V^*})$ be the Reeb section of the cosymplectic
structure $(\Omega_h,\eta)$ characterized by the following
conditions
$$i_{R_h}\Omega_h=0\makebox[1cm]{and}i_{R_h}\eta=1.$$
With respect to the basis
$\{\tilde{e}_0,\tilde{e}_{\alpha},\bar{e}_{\alpha}\}$ of
$\Gamma(\tau_{\widetilde{{\mathcal A}}}^{\tau_V^*})$, $R_h$ is
locally expressed as follows:
\begin{equation}\label{rh}
R_h=\tilde{e}_0+\frac{\partial H}{\partial
y_{\alpha}}\tilde{e}_{\alpha}-(C_{\alpha\beta}^{\gamma}y_{\gamma}\frac{\partial
H}{\partial y_{\beta}}+\rho^i_{\alpha}\frac{\partial H}{\partial
x^i}-C_{0\alpha}^{\gamma}y_{\gamma})\bar{e}_{\alpha}.
\end{equation}

Thus, the vector field $\rho_{\widetilde{{\mathcal
A}}}^{\tau_{V}^*}(R_h)$ is locally given by
\begin{equation}\label{eq3.7'}
\rho_{\widetilde{{\mathcal
A}}}^{\tau_{V}^*}(R_h)=(\rho_0^i+\rho_{\alpha}^i\frac{\partial H
}{\partial y_{\alpha}})\frac{\partial}{\partial
x^i}+\Big(-\rho_{\alpha}^i\frac{\partial H}{\partial
x^i}+y_\gamma(C_{0\alpha}^{\gamma}+C_{\beta\alpha}^{\gamma}\frac{\partial
H}{\partial y_{\beta}})\Big)\frac{\partial}{\partial y_{\alpha}}
\end{equation}
and the integral curves of $R_h$ (i.e., the integral curves of
$\rho_{\widetilde{{\mathcal A}}}^{\tau_{V}^*}(R_h)$) are just {\it
the solutions of the Hamilton equations for $h$},
$$\frac{dx^i}{dt}=\rho_0^i+\rho_{\alpha}^i\frac{\partial
H}{\partial
y_{\alpha}},\;\;\;\frac{dy_{\alpha}}{dt}=-\rho_{\alpha}^i\frac{\partial
H }{\partial
x^i}+y_{\gamma}(C_{0\alpha}^{\gamma}+C_{\beta\alpha}^{\gamma}\frac{\partial
H }{\partial y_{\beta}}),$$
for $i\in\{1,\dots,m\}$ and $\alpha\in\{1,\dots,n\}$.

Next, we will present an alternative approach in order to obtain
the Hamilton equations. For this purpose, we will use the notion
of an aff-Poisson structure on an AV-bundle which was introduced
in \cite{GGrU} (see also \cite{GGU2}).

Let $\tau_Z:Z\to M$ be an affine bundle of rank $1$ modelled on
the trivial vector bundle $\tau_{M\times\R}:M\times\R\to M$, that
is, $\tau_Z:Z\to M$ is an {\em AV-bundle } in the terminology of
\cite{GGU2}.

Then, we have an action of $\R$ on the fibres of $Z$. This action
induces a vector field $X_Z$ on $Z$ which is vertical with respect
to the projection $\tau_Z:Z\to M$.

On the other hand, there exists a one-to-one correspondence
between the space of sections of $\tau_Z:Z\to M$,
$\Gamma(\tau_Z)$, and the set
$$\{F_h\in C^{\infty}(Z)/X_Z(F_h)=-1\}.$$
In fact, if $h\in\Gamma(\tau_Z)$ and $(x^i,s)$ are local fibred
coordinates on $Z$ such that
$X_Z=\displaystyle\frac{\partial}{\partial s}$ and $h$ is locally
defined by $h(x^i)=(x^i,-H(x^i))$, then the function $F_h$ on $Z$
is locally given by
\begin{equation}\label{eq3.81}
F_h(x^i,s)=-H(x^i)-s,
\end{equation}
(for more details, see \cite{GGU2}).

Now, an {\em aff-Poisson structure } on the AV-bundle $\tau_Z:Z\to
M$ is a bi-affine map
$$\{\cdot,\cdot\}:\Gamma(\tau_Z)\times\Gamma(\tau_Z)\to
C^{\infty}(M)$$ which satisfies the following properties:
\begin{enumerate}
\item[i)] Skew-symmetric: $\{h_1,h_2\}=-\{h_2,h_1\}$.
\item[ii)] Jacobi identity:
$$\{h_1,\{h_2,h_3\}\}_V+\{h_2,\{h_3,h_1\}\}_V+\{h_3,\{h_1,h_2\}\}_V=0,$$
where $\{\cdot,\cdot\}_V$ is the affine-linear part of the
bi-affine bracket.
\item[iii)] If $h\in\Gamma(\tau_Z)$ then
$$\{h,\cdot\}:\Gamma(\tau_Z)\to
C^{\infty}(M),\;\;\;h'\mapsto\{h,h'\},$$ is an affine derivation.
\end{enumerate}
Condition iii) implies that, for each $h\in\Gamma(\tau_Z)$ the
linear part $\{h,\cdot\}_V:C^{\infty}(M)\to C^{\infty}(M)$ of the
affine map $\{h,\cdot\}:\Gamma(\tau_Z)\to C^{\infty}(M)$ defines a
vector field on $M$, which is called {\it the Hamiltonian vector
field of $h$} (see \cite{GGU2}).

In \cite{GGU2}, the authors proved that there is a one-to-one
correspondence between aff-Poisson bra\-ckets $\{\cdot,\cdot\}$ on
$\tau_Z:Z\to M$ and Poisson brackets $\{\cdot,\cdot\}_{\Pi}$ on
$Z$ which are $X_Z$-invariant, i.e., which are associated with
Poisson 2-vectors $\Pi$ on $Z$ such that ${\mathcal
L}_{X_Z}\Pi=0$. This correspondence is determined by
$$\{h_1,h_2\}\circ\tau_Z=\{F_{h_1},F_{h_2}\}_\Pi,\makebox[1cm]{for}h_1,h_2\in\Gamma(\tau_Z).$$
Note that the function $\{F_{h_1},F_{h_2}\}_\Pi$ on $Z$ is
$\tau_Z$-projectable, i.e., ${\mathcal
L}_{X_Z}\{F_{h_1},F_{h_2}\}_\Pi$ $=0$ (because the Poisson
2-vector $\Pi$ is $X_Z$-invariant).

Using this correspondence we will prove the following result.

\begin{theorem}\label{teor3.2}\cite{IMPS} Let $\tau_{\mathcal A}:{\mathcal A}\to M$ be a Lie affgebroid modelled
on the vector bundle $\tau_V:V\to M$. Denote by $\tau_{{\mathcal
A}^+}:{\mathcal A}^+\to M$ (resp., $\tau_V^*:V^*\to M$) the dual
vector bundle to ${\mathcal A}$ (resp., to $V$) and by
$\mu:{\mathcal A}^+\to V^*$ the canonical projection. Then:
\begin{enumerate}
\item[i)] $\mu:{\mathcal A}^+\to V^*$ is an AV-bundle which admits
an aff-Poisson structure. \item[ii)] If $h:V^*\to {\mathcal A}^+$
is a Hamiltonian section (that is, $h\in\Gamma(\mu)$) then the
Hamiltonian vector field of $h$ with respect to the aff-Poisson
structure is a vector field on $V^*$ whose integral curves are
just the solutions of the Hamilton equations for $h$.
\end{enumerate}
\end{theorem}
\begin{proof}i) It is clear that $\mu:{\mathcal {\mathcal A}}^+\to V^*$ is an
AV-bundle. In fact, if $\a^+\in {\mathcal A}^+_x$, with $x\in M$,
and $t\in\R$ then
$$\a^++t=\a^++t1_{\mathcal A}(x).$$
Thus, the $\mu$-vertical vector field $X_{{\mathcal A}^+}$ on
${\mathcal A}^+$ is just the vertical lift $1_{\mathcal A}^V$ of
the section $1_{\mathcal A}\in\Gamma(\tau_{{\mathcal A}^+})$.
Moreover, one may consider the Lie algebroid
$\tau_{\widetilde{{\mathcal A}}}:\widetilde{{\mathcal
A}}=({\mathcal A}^+)^*\to M$ and the corres\-ponding linear
Poisson 2-vector $\Pi_{{\mathcal A}^+}$ on ${\mathcal A}^+$. Then,
using the fact that $1_{\mathcal A}$ is a 1-cocycle of
$\tau_{\widetilde{{\mathcal A}}}:\widetilde{{\mathcal
A}}=({\mathcal A}^+)^*\to M$, it follows that the Poisson 2-vector
$\Pi_{{\mathcal A}^+}$ is $X_{{\mathcal A}^+}$-invariant.
Therefore, $\Pi_{{\mathcal A}^+}$ induces an aff-Poisson structure
$\{\cdot,\cdot\}$ on $\mu:{\mathcal A}^+\to V^*$ which is
characterized by the condition
\begin{equation}\label{eq3.82}
\{h_1,h_2\}\circ\mu=\{F_{h_1},F_{h_2}\}_{\Pi_{{\mathcal
A}^+}},\makebox[1cm]{for}h_1,h_2\in\Gamma(\mu).
\end{equation}
One may also prove this first part of the theorem using the
relation between special Lie affgebroid structures on an affine
bundle ${\mathcal A}'$ and aff-Poisson structures on the AV-bundle
$AV(({\mathcal A}')^\sharp)$ (see Theorem $23$ in \cite{GGU2}).

ii) From (\ref{eq3.81}) and (\ref{eq3.82}), we deduce that the
linear map $\{h,\cdot\}_V:C^\infty(V^*)\to C^\infty(V^*)$
associated with the affine map $\{h,\cdot\}:\Gamma(\mu)\to
C^\infty(V^*)$ (that is, the Hamiltonian vector field of $h$) is
given by
\begin{equation}\label{eq3.83}
\{h,\cdot\}_V(\varphi)\circ\mu=\{F_h,\varphi\circ\mu\}_{\Pi_{{\mathcal
A}^+}},\makebox[1cm]{for}\varphi\in C^\infty(V^*).
\end{equation}
Now, suppose that the local expression of $h$ is
\begin{equation}\label{eq3.84}
h(x^i,y_\alpha)=(x^i,-H(x^j,y_\beta),y_\alpha).
\end{equation}
On the other hand, using (\ref{eq2.3'}), we have that
\begin{equation}\label{eq3.85}
\Pi_{{\mathcal A}^+}=\rho_0^i\displaystyle\frac{\partial}{\partial
x^i}\wedge\frac{\partial}{\partial
y_0}+\rho_\alpha^i\displaystyle\frac{\partial}{\partial
x^i}\wedge\frac{\partial}{\partial y_\alpha}-C_{0\alpha}^\gamma
y_\gamma\frac{\partial}{\partial
y_0}\wedge\frac{\partial}{\partial
y_\alpha}-\displaystyle\frac{1}{2}C_{\alpha\beta}^\gamma
y_\gamma\frac{\partial}{\partial
y_\alpha}\wedge\frac{\partial}{\partial y_\beta}.
\end{equation}
Thus, from (\ref{eq3.83}), (\ref{eq3.84}) and (\ref{eq3.85}), we
conclude that the Hamiltonian vector field of $h$ is locally given
by
$$(\rho_0^i+\rho_\alpha^i\frac{\partial H
}{\partial y_{\alpha}})\frac{\partial}{\partial
x^i}+\Big(-\rho_{\alpha}^i\frac{\partial H}{\partial
x^i}+y_\gamma(C_{0\alpha}^{\gamma}+C_{\beta\alpha}^{\gamma}\frac{\partial
H}{\partial y_{\beta}})\Big)\frac{\partial}{\partial y_{\alpha}}$$
which proves our result (see (\ref{eq3.7'})).
\end{proof}

\subsection{The Legendre transformation and the equivalence
between the Lagrangian and Hamiltonian formalisms}\label{sec3.3}

Let $L:{\mathcal A}\to \R$ be a Lagrangian function and
$\Theta_L\in \Gamma((\tau^{\tau_{\mathcal
A}}_{\widetilde{{\mathcal A}}})^*)$ be the Poincar\'{e}-Cartan
$1$-section associated with $L$. We introduce {\it the extended
Legendre transformation associated with $L$} as the smooth map
$Leg_L:{\mathcal A}\to {\mathcal A}^+$ defined by
$Leg_L(\a)(\b)=\Theta_L(\a)(z),$
for $\a,\b\in {\mathcal A}_x,$ where $z\in{\mathcal
T}_\a^{\widetilde{{\mathcal A}}}{\mathcal A}$ is such that
$pr_1(z)=i_{\mathcal A}(\b),$ $pr_1:{\mathcal
T}^{\widetilde{{\mathcal A}}}{\mathcal A}\to\widetilde{{\mathcal
A}}$ being the restriction to ${\mathcal  T}^{\widetilde{{\mathcal
A}}}{\mathcal A}$ of the first canonical projection
$pr_1:\widetilde{{\mathcal A}}\times T{\mathcal
A}\to\widetilde{{\mathcal A}}$. The map $Leg_L$ is well-defined
and its local expression in fibred coordinates on ${\mathcal A}$
and  ${\mathcal A}^+$ is
\begin{equation}\label{locLegL}
Leg_L(x^i,y^\alpha)=(x^i,L-\frac{\partial L}{\partial
y^{\alpha}}y^{\alpha},\frac{\partial L}{\partial y^{\alpha}}).
\end{equation}

Thus, we can define {\it the Legendre transformation associated
with $L$}, $leg_L:{\mathcal A}\to V^*$, by $leg_L=\mu\circ Leg_L.$
From (\ref{locLegL}) and since
$\mu(x^i,y_0,y_{\alpha})=(x^i,y_{\alpha})$, we have that
\begin{equation}\label{loclegL}
leg_L(x^i,y^{\alpha})=(x^i,\displaystyle\frac{\partial L}{\partial
y^\alpha}).
\end{equation}

The maps $Leg_L$ and $leg_L$ induce the maps ${\mathcal
T}{Leg_L}:{\mathcal  T}^{\widetilde{{\mathcal A}}}{\mathcal A}\to
{\mathcal T}^{\widetilde{{\mathcal A}}}{\mathcal A}^+$ and
${\mathcal T}{leg_L}:{\mathcal T}^{\widetilde{{\mathcal
A}}}{\mathcal A}\to {\mathcal T}^{\widetilde{{\mathcal A}}}V^*$
defined by
\begin{equation}\label{LLegL}
({\mathcal T}{Leg_L})(\tilde{\b},X_\a)=(\tilde{\b},(T_\a
Leg_L)(X_\a)),\;\;({\mathcal
T}{leg_L})(\tilde{\b},X_\a)=(\tilde{\b},(T_\a leg_L)(X_\a)),
\end{equation}
for $\a\in {\mathcal A}$ and $(\tilde{\b},X_\a)\in {\mathcal
T}_\a^{\widetilde{{\mathcal A}}}{\mathcal A}$.

Now, let $\{\tilde{T}_0,\tilde{T}_{\alpha},\tilde{V}_{\alpha}\}$
(respectively,
$\{\tilde{e}_0,\tilde{e}_{\alpha},\bar{e}_0,\bar{e}_{\alpha}\}$)
be a local basis of $\Gamma(\tau_{\widetilde{{\mathcal
A}}}^{\tau_{\mathcal A}})$ as in Section \ref{Section2.4}
(respectively, of $\Gamma(\tau_{\widetilde{{\mathcal
A}}}^{\tau_{{\mathcal A}^+}})$ as in Section \ref{sec1.1.1}) and
denote by $(x^i,y^{\alpha};z^0,z^{\alpha},v^{\alpha})$
(respectively,
$(x^i,y_0,y_{\alpha};z^0,z^{\alpha},v_0,v_{\alpha})$) the
corresponding local coordinates on ${\mathcal
T}^{\widetilde{{\mathcal A}}}{\mathcal A}$ (respectively,
${\mathcal T}^{\widetilde{{\mathcal A}}}{\mathcal A}^+$). In
addition, suppose that
$(x^i,y_{\alpha};z^0,z^{\alpha},v_{\alpha})$ are local coordinates
on ${\mathcal  T}^{\widetilde{{\mathcal A}}}V^*$ as in Section
\ref{sec3.1}. Then, from (\ref{locLegL}), (\ref{loclegL}) and
(\ref{LLegL}), we deduce that the local expression of the maps
${\mathcal T}{Leg_L}$ and ${\mathcal T}leg_L$ is
\begin{equation}\label{LLegloc}
\begin{array}{rcl}
{\mathcal
T}{Leg_L}(x^i,y^\alpha;z^0,z^\alpha,v^\alpha)\kern-10pt&=&\kern-10pt(x^i,L-\displaystyle\frac{\partial
L}{\partial y^{\alpha}}y^{\alpha},\frac{\partial L}{\partial
y^\alpha};z^0,z^\alpha,z^0\rho_0^i(\frac{\partial L}{\partial
x^i}-\frac{\partial^2L}{\partial x^i\partial
y^{\alpha}}y^{\alpha})\\[8pt]&&
\kern-20pt+z^\alpha\rho_\alpha^i(\displaystyle\frac{\partial
L}{\partial x^i}-\frac{\partial^2 L}{\partial x^i\partial
y^\beta}y^{\beta})- v^\alpha\frac{\partial^2 L}{\partial
y^\alpha\partial y^\beta}y^{\beta}, z^0\rho_0^i\frac{\partial^2
L}{\partial x^i\partial
y^{\alpha}}\\[8pt]&&\kern-20pt+z^{\beta}\rho_{\beta}^i\displaystyle\frac{\partial^2L}{\partial
x^i\partial
y^{\alpha}}+v^{\beta}\displaystyle\frac{\partial^2L}{\partial
y^{\alpha}\partial y^{\beta}} ),
\end{array}\end{equation}
\begin{equation}\label{Llegloc}
\begin{array}{rcl}
 {\mathcal
T}{leg_L}(x^i,y^\alpha;z^0,z^\alpha,v^\alpha)&=&(x^i,\displaystyle\frac{\partial
L}{\partial
y^\alpha};z^0,z^{\alpha},z^0\rho_0^i\displaystyle\frac{\partial^2L}{\partial
x^i\partial
y^{\alpha}}+z^{\beta}\rho_{\beta}^i\frac{\partial^2L}{\partial
x^i\partial
y^{\alpha}}\\[8pt] &&+v^{\beta}\displaystyle\frac{\partial^2L}{\partial
y^{\alpha}\partial y^{\beta}}).  \end{array}\end{equation} Thus,
using (\ref{formas}), (\ref{PCloc}), (\ref{LLegL}) and
(\ref{LLegloc}), we can prove the following result.
\begin{theorem}\label{t3.2}\cite{IMPS}
The pair $({\mathcal  T}{Leg_L},Leg_L)$ is a morphism between the
Lie algebroids $({\mathcal  T}^{\widetilde{{\mathcal A}}}{\mathcal
A}, \linebreak \lcf\cdot,\cdot\rcf^{\tau_{\mathcal
A}}_{\widetilde{{\mathcal A}}},\rho^{\tau_{\mathcal
A}}_{\widetilde{{\mathcal A}}})$ and $({\mathcal
T}^{\widetilde{{\mathcal A}}}{\mathcal
A}^+,\lcf\cdot,\cdot\rcf^{\tau_{{\mathcal
A}^+}}_{\widetilde{{\mathcal A}}},\rho^{\tau_{{\mathcal
A}^+}}_{\widetilde{{\mathcal A}}}).$ Moreover, if $\Theta_L$ and
$\Omega_L$ (respectively, $\lambda_{\widetilde{{\mathcal A}}}$ and
$\Omega_{\widetilde{{\mathcal A}}})$ are the Poincar\'{e}-Cartan
$1$-section and $2$-section associated with $L$ (respectively, the
Liouville $1$-section and the canonical symplectic section
associated with $\widetilde{{\mathcal A}}$) then
\begin{equation}\label{pullback}
({\mathcal T}{Leg_L},Leg_L)^*(\lambda_{\widetilde{{\mathcal
A}}})=\Theta_L,\;\;\; ({\mathcal
T}{Leg_L},Leg_L)^*(\Omega_{\widetilde{{\mathcal A}}})=\Omega_L.
\end{equation}
\end{theorem}
From (\ref{loclegL}), it follows
\begin{proposition}\cite{IMPS}
The Lagrangian $L$ is regular if and only if the Legendre
transformation $leg_L:{\mathcal A}\to V^*$ is a local
diffeomorphism.
\end{proposition}

Next, we will assume that $L$ is {\it hyperregular}, that is,
$leg_L$ is a global diffeomorphism. Then, from (\ref{LLegL}) and
Theorem \ref{t3.2}, we conclude that the pair $({\mathcal
T}leg_L,leg_L)$ is a Lie algebroid isomorphism. Moreover, we can
consider the Hamiltonian section $h:V^*\to {\mathcal A}^+$ defined
by
\begin{equation}\label{h}
h=Leg_L\circ leg_L^{-1}, \end{equation} the corresponding
cosymplectic structure $(\Omega_{h},\eta)$ on the Lie algebroid
$\tau_{\widetilde{{\mathcal A}}}^{\tau_V^*}:{\mathcal
T}^{\widetilde{{\mathcal A}}}V^*\to V^*$ and the Hamiltonian
section $R_{h}\in\Gamma(\tau_{\widetilde{{\mathcal
A}}}^{\tau_V^*})$.

Using (\ref{Omegah}), (\ref{Llegloc}), (\ref{pullback}), (\ref{h})
and Theorem \ref{t3.2}, we deduce that

\begin{theorem}\label{t3.4}\cite{IMPS}
If the Lagrangian $L$ is hyperregular then the Euler-Lagrange
section $R_L$ asso\-cia\-ted with $L$ and the Hamiltonian section
$R_{h}$ associated with $h$ satisfy the following relation
$$R_{h}\circ leg_L={\mathcal  T}leg_L\circ R_L.$$

Moreover, if $\gamma:I\to {\mathcal A}$ is a solution of the
Euler-Lagrange equations associated with $L$, then $leg_L\circ
\gamma:I\to V^*$ is a solution of the Hamilton equations
associated with $h$ and, conversely, if $\bar{\gamma}:I\to V^*$ is
a solution of the Hamilton equations for $h$ then
$\gamma=leg_L^{-1}\circ\bar{\gamma}$ is a solution of the
Euler-Lagrange equations for $L$.
\end{theorem}

Note that if the local expression of the inverse of the Legendre
transformation $leg_L:{\mathcal A}\to V^*$ is
$$leg_L^{-1}(x^i,y_\beta)=(x^i,y^\alpha(x^j,y_\beta)),$$
then the Hamiltonian section $h$ is locally given by
$$h(x^i,y_\alpha)=(x^i,-H(x^j,y_\beta),y_\alpha),$$
where the function $H$ is
$$H(x^i,y_\alpha)=y_\beta
y^\beta(x^i,y_\alpha)-L(x^i,y^\beta(x^i,y_\alpha)).$$ Thus,
\begin{equation}\label{PartialsH}
\displaystyle\frac{\partial H}{\partial
x^i}_{|(x^j,y_\beta)}=-\frac{\partial L}{\partial
x^i}_{|(x^j,y^\alpha(x^j,y_\beta))},\;\;\;\;\frac{\partial
H}{\partial y_\beta}_{|(x^i,y_\alpha)}=y^\beta(x^i,y_\alpha).
\end{equation}
Therefore,
\begin{equation}\label{legL-1}
leg_L^{-1}(x^i,y_\alpha)=(x^i,\displaystyle\frac{\partial
H}{\partial y_\alpha}).
\end{equation}

\section{Affinely constrained Lagrangian systems}\label{Noho}

We start with a free Lagrangian system on a Lie affgebroid
${\mathcal A}$ of rank $n$. Now, we plug in some nonholonomic
affine constraints described by an affine subbbundle ${\mathcal
B}$ of rank $n-r$ of the bundle ${\mathcal A}$ of admissible
directions, that is, we have an affine bundle $\tau_{\mathcal
B}:{\mathcal B}\to M$ with associated vector bundle
$\tau_{U_{\mathcal B}}: U_{\mathcal B}\to M$ and the corresponding
inclusions $i_{\mathcal B}:{\mathcal B}\hookrightarrow {\mathcal
A}$ and $i_{U_{\mathcal B}}:U_{\mathcal B}\hookrightarrow V$,
$\tau_V:V\to M$ being the vector bundle associated with the affine
bundle $\tau_{\mathcal A}:{\mathcal A}\to M$. If we impose to the
admissible solution curves $\gamma(t)$ the condition to stay on
the manifold ${\mathcal B}$, we arrive at the equations $\delta
L_{\gamma(t)}=\lambda(t)$ and $\gamma(t)\in {\mathcal B}$, where
the constraint force $\lambda(t)\in {\mathcal A}^+_{\tau_{\mathcal
A}(\gamma(t))}$ is to be determined. In the standard case
(${\mathcal A}=J^1\tau$, $\tau:M\to\R$ being a fibration),
$\lambda$ takes values in the annihilator of the constraint
submanifold ${\mathcal B}$. Therefore, in the case of a general
Lie affgebroid, the natural \emph{Lagrange-d'Alembert equations}
one should pose are
$$\delta L_{\gamma(t)}\in
{\mathcal B}^\circ_{\tau_{\mathcal
B}(\gamma(t))}\makebox[1cm]{and}\gamma(t)\in {\mathcal B},$$ where
${\mathcal B}^\circ=\{\varphi\in {\mathcal
A}^+/\varphi_{|{\mathcal B}}\equiv 0\}$ is the affine annihilator
of ${\mathcal B}$ which is a vector subbundle of ${\mathcal A}^+$
with rank equal to $r$.

In more explicit terms, we look for curves $\gamma(t)\in {\mathcal
A}$ such that
\begin{itemize}
\item[-]$\gamma\in Adm({\mathcal A})$, \item[-]$\gamma(t)\in
{\mathcal B}_{\tau_{\mathcal A}(\gamma(t))}$, \item[-]there exists
$\lambda(t)\in {\mathcal B}^\circ_{\tau_{\mathcal A}(\gamma(t))}$
such that $\delta L_{\gamma(t)}=\lambda(t)$.
\end{itemize}

If $\gamma(t)$ is one of such curves, then
$(\gamma(t),\dot{\gamma}(t))$ is a curve in ${\mathcal
T}^{\widetilde{{\mathcal A}}}{\mathcal A}$ with
$i_{(\gamma(t),\dot{\gamma}(t))}\phi_0=1$. Moreover, since
$\gamma(t)$ is in ${\mathcal B}$, we have that $\dot{\gamma}(t)$
is tangent to ${\mathcal B}$, that is,
$(\gamma(t),\dot{\gamma}(t))\in{\mathcal T}^{\widetilde{{\mathcal
A}}}{\mathcal B}$ with $i_{(\gamma(t),\dot{\gamma}(t))}\phi_0=1$.
Note that, since ${\mathcal B}$ is an affine subbundle of
${\mathcal A}$, the Lie algebroid prolongation
$\tau_{\widetilde{{\mathcal A}}}^{\tau_{\mathcal B}}:{\mathcal
T}^{\widetilde{{\mathcal A}}}{\mathcal B}\to {\mathcal B}$ of the
Lie algebroid $\widetilde{{\mathcal A}}$ over the fibration
$\tau_{\mathcal B}:{\mathcal B}\to M$ is well-defined. Under some
regularity conditions (to be made precise later on), we may assume
that the above admi\-ssi\-ble curves are integral curves of a
section $X$ which we will assume that it is a SODE section
ta\-king va\-lues in ${\mathcal T}^{\widetilde{{\mathcal
A}}}{\mathcal B}$. Then, from (\ref{E-Loperator}), we deduce that
$$(i_X\Omega_L)(\b)\in\tilde{{\mathcal B}}^\circ_\b,\mbox{ for all
}\b\in{\mathcal B},$$ where $\tilde{{\mathcal B}}^\circ$ is the
vector bundle over ${\mathcal B}$ whose fibre at the point
$\b\in{\mathcal B}$ is
$$\tilde{{\mathcal B}}^\circ_\b=\{\alpha\in({\mathcal
T}^{\widetilde{{\mathcal A}}}_\b{\mathcal
A})^*/\,\alpha(\b',v_\b)=0,\forall\,(\b',v_\b)\in{\mathcal
T}^{\widetilde{{\mathcal A}}}_\b{\mathcal A}\mbox{ such that
}\b'\in{\mathcal B}_{\tau_{{\mathcal A}}(\b)}\}.$$

 Based on the previous  arguments, we may reformulate geometrically
our problem as the search for such a SODE $X$ (defined at least on
a neighborhood of ${\mathcal B}$) satisfying
$i_{X(\b)}\Omega_L(\b)\in \tilde{{\mathcal B}}^\circ_\b$,
$i_{X(\b)}\phi_0(\b)=1$ and $X(\b)\in{\mathcal
T}^{\widetilde{{\mathcal A}}}_\b{\mathcal B}$, at every point
$\b\in {\mathcal B}$.

Now, we will prove the following result.

\begin{proposition}\label{prop3.0} If $S$ is the vertical endomorphism then the
vector bundles over ${\mathcal B}$, $\tilde{{\mathcal
B}}^\circ\to{\mathcal B}$ and $S^*({\mathcal
T}^{\widetilde{{\mathcal A}}}{\mathcal B})^\circ\to{\mathcal B}$,
are equal. In other words,
$$\tilde{{\mathcal B}}^\circ_\b=S^*({\mathcal
T}_\b^{\widetilde{{\mathcal A}}}{\mathcal B})^\circ,\mbox{ for all
}\b\in{\mathcal B}.$$
\end{proposition}
\begin{proof}First, we will see that the map
$$S^*:({\mathcal
T}^{\widetilde{{\mathcal A}}}_\b{\mathcal B})^\circ\to
S^*({\mathcal T}^{\widetilde{{\mathcal A}}}_\b{\mathcal
B})^\circ$$ is a linear isomorphism.

In fact, suppose that $\alpha\in({\mathcal
T}^{\widetilde{{\mathcal A}}}_\b{\mathcal B})^\circ$ and
$S^*\alpha=0$. Then, we have that
$$\alpha(0,v_\b)=0,\mbox{ for all }v_\b\in T_\b{\mathcal A}\mbox{
such that }(T_\b\tau_{{\mathcal A}})(v_\b)=0.$$

Moreover, if $(\tilde{a},u_\b)\in{\mathcal T}^{\widetilde{\mathcal
A}}_\b{\mathcal A}$ then, since $T_\b\tau_{\mathcal
A}:T_\b{\mathcal B}\rightarrow T_{\tau_{\mathcal A}(\b)}M$ is a
linear epimorphism, it follows that there exists $v_\b\in
T_\b{\mathcal B}$ satisfying
$$(T_\b\tau_{\mathcal A})(v_\b)=\rho_{\widetilde{\mathcal
A}}(\tilde{a}).$$

Thus, using that $\alpha\in({\mathcal T}^{\widetilde{{\mathcal
A}}}_\b{\mathcal B})^\circ$, we deduce that
$$\alpha(\tilde{a},u_\b)=\alpha(\tilde{a},v_\b)+\alpha(0,u_\b-v_\b)=0.$$
Therefore, $\alpha=0$.

This implies that $S^*:({\mathcal T}^{\widetilde{{\mathcal
A}}}_\b{\mathcal B})^\circ\to S^*({\mathcal
T}^{\widetilde{{\mathcal A}}}_\b{\mathcal B})^\circ$ is a linear
isomorphism and
$$dim\,S^*({\mathcal
T}^{\widetilde{{\mathcal A}}}_\b{\mathcal B})^\circ=r.$$

On the other hand, since $T_\b\tau_{\mathcal A}:T_\b{\mathcal
B}\to T_{\tau_{\mathcal A}(\b)}M$ is a linear epimorphism, we
obtain that
$$dim\,\tilde{\mathcal B}^\circ_\b=r=dim\,S^*({\mathcal
T}^{\widetilde{{\mathcal A}}}_\b{\mathcal B})^\circ.$$ Finally, we
will prove that $S^*({\mathcal T}_\b^{\widetilde{{\mathcal
A}}}{\mathcal B})^\circ\subseteq \tilde{{\mathcal B}}^\circ_\b$.

In fact, if $\alpha\in({\mathcal T}^{\widetilde{{\mathcal
A}}}_\b{\mathcal B})^\circ$, $(\b',v_\b)\in {\mathcal
T}^{\widetilde{{\mathcal A}}}_\b{\mathcal A}$ and $\b'\in{\mathcal
B}_{\tau_{\mathcal B}(\b)}$, it follows that
$$S(\b',v_\b)=(0,v_\b'),\mbox{ with }v_\b'\in T_\b{\mathcal B}.$$
Consequently,
$$(S^*\alpha)(\b',v_\b)=0.$$

\end{proof}

Proposition \ref{prop3.0} suggests us to introduce the following
definition.

\begin{definition} A nonholonomically constrained Lagrangian
system on a Lie affgebroid ${\mathcal A}$ is a pair $(L,{\mathcal
B})$, where $L: {\mathcal A}\longrightarrow \R$ is a
$C^\infty$-\emph{Lagrangian function}, and $i_{\mathcal
B}:{\mathcal B}\hookrightarrow {\mathcal A}$ is a smooth affine
subbundle of ${\mathcal A}$, called the \emph{constraint affine
subbundle}. By a \emph{solution} of the nonholonomically
constrained Lagrangian system $(L,{\mathcal B})$ we mean a section
$X\in\Gamma(\tau_{\widetilde{{\mathcal A}}}^{\tau_{\mathcal A}})$
which is a SODE on ${\mathcal B}$ and satisfies the
Lagrange-d'Alembert equations
\begin{equation}\label{LdA}
\left\{ \begin{array}{rcl} (i_X\Omega_L)_{|{\mathcal
B}}&\in&\Gamma(\tau_{S^*({\mathcal
T}^{\widetilde{{\mathcal A}}}{\mathcal B})^\circ}),\\
(i_X\phi_0)_{|{\mathcal B}}&=&1,\\
X_{|{\mathcal B}}&\in&\Gamma(\tau_{\widetilde{{\mathcal
A}}}^{\tau_{\mathcal B}}),
\end{array} \right.
\end{equation}
where $\tau_{S^*({\mathcal T}^{\widetilde{{\mathcal A}}}{\mathcal
B})^\circ}$ is the vector bundle projection of $S^*({\mathcal
T}^{\widetilde{{\mathcal A}}}{\mathcal B})^\circ\to {\mathcal B}$.
\end{definition}

With a slight abuse of language, we will interchangeably refer to
a solution of the constrained Lagrangian system as a section or
the collection of its corresponding integral curves.

\begin{remark}{\em We want to stress that a solution of the
Lagrange-d'Alembert equations needs to be defined only over
${\mathcal B}$, but for practical purposes we consider it extended
to ${\mathcal A}$ (or just to a neighborhood of ${\mathcal B}$ in
${\mathcal A}$). We will not make any notational distinction
between a solution on ${\mathcal B}$ and any of its extensions.
Solutions which coincide on ${\mathcal B}$ will be considered as
equal. In accordance with this convention, by a SODE on ${\mathcal
B}$ we mean a section of ${\mathcal T}^{\widetilde{{\mathcal
A}}}{\mathcal B}$ which is the restriction to ${\mathcal B}$ of
some SODE defined in a neighborhood of ${\mathcal
B}$.\hfill$\diamondsuit$}
\end{remark}

Now, we will analyze the form of the Lagrange-d'Alembert equations
in local coordinates. Let $(x^i)$ be local coordinates on an open
subset $U$ of $M$ and $\{e_0,e_{\alpha}\}$ be a local basis of
sections of the vector bundle $\tau^{-1}_{\widetilde{{\mathcal
A}}}(U)\rightarrow U$ adapted to $1_{\mathcal A}$. Denote by
$(x^i,y^0,y^\alpha)$ the corresponding local coordinates on
$\widetilde{\mathcal A}$ and suppose that
$\{\mu_\alpha^ae^\alpha\}$ is a local basis of the vector bundle
$U_{\mathcal B}^\circ\to M$ and that $e_0^{\mathcal B}$ is a local
section of the affine bundle ${\mathcal B}\to M$ such that
$e_0-e_0^{\mathcal B}=a_0^\alpha e_\alpha$. Then, using that an
element $e_0(x)+y^\alpha e_\alpha(x)$ of ${\mathcal A}_x$ belongs
to ${\mathcal B}_x$ if and only if $(e_0(x)+y^\alpha
e_\alpha(x))-e_0^{\mathcal B}(x)\in(U_{\mathcal B})_x$, we deduce
that the local equations defining the constrained subbundle
${\mathcal B}$ as an affine subbundle of ${\mathcal A}$ are
$$\Psi^a=\mu^a_0+\mu^a_\alpha y^\alpha=0,\makebox[1.3cm]{for all}a\in\{1,\dots,r\},$$
where $\mu_0^a=\mu_\alpha^a a_0^\alpha$. Observe that
$d^{{\mathcal T}^{\widetilde{{\mathcal A}}}{\mathcal
A}}\Psi^a\in\Gamma((\tau_{\widetilde{{\mathcal
A}}}^{\tau_{\mathcal A}})^*)$ verifies that
$$(d^{{\mathcal
T}^{\widetilde{{\mathcal A}}}{\mathcal
A}}\Psi^a)(\b)(\tilde{\a},X_\b)=0,\makebox[1.5cm]{for
all}(\tilde{\a},X_\b)\in{\mathcal T}^{\widetilde{{\mathcal
A}}}_\b{\mathcal B},$$ and, therefore, $(d^{{\mathcal
T}^{\widetilde{{\mathcal A}}}{\mathcal A}}\Psi^a)(\b)\in({\mathcal
T}^{\widetilde{{\mathcal A}}}_\b{\mathcal B})^\circ$, for every
$\b\in {\mathcal B}$ and $a\in\{1,\dots,r\}$.

On the other hand, using (\ref{dE}) and (\ref{corTVtilde}), we
deduce that
$$d^{{\mathcal T}^{\widetilde{{\mathcal A}}}{\mathcal
A}}\Psi^a=\rho_0^i\Big(\displaystyle\frac{\partial
\mu^a_0}{\partial x^i}+y^\beta\frac{\partial \mu^a_\beta}{\partial
x^i}\Big)\phi_0+\rho_\alpha^i\Big(\displaystyle\frac{\partial
\mu^a_0}{\partial x^i}+y^\beta\frac{\partial \mu^a_\beta}{\partial
x^i}\Big)\tilde{T}^\alpha+\mu^a_\alpha\tilde{V}^\alpha,$$
$\{\phi_0,\tilde{T}^\alpha,\tilde{V}^\alpha\}$ being the dual
basis of $\{\tilde{T}_0,\tilde{T}_\alpha,\tilde{V}_\alpha\}$.

Thus, since the matrix $(\mu^a_\alpha)$ has maximum rank r, it
follows that the sections $\{d^{{\mathcal T}^{\widetilde{{\mathcal
A}}}{\mathcal A}}\Psi^a\}_{a=1,\dots,r}$ are linearly independent
and, using that $rank({\mathcal T}^{\widetilde{{\mathcal
A}}}{\mathcal B})^\circ=r$, we deduce that $\{(d^{{\mathcal
T}^{\widetilde{{\mathcal A}}}{\mathcal A}}\Psi^a)_{|{\mathcal B}}
\}_{a=1,\dots,r}$ is a local basis of sections of $({\mathcal
T}^{\widetilde{{\mathcal A}}}{\mathcal B})^\circ\to {\mathcal B}$.

Therefore, using (\ref{verend}) and the fact that $S^*:({\mathcal
T}^{\widetilde{{\mathcal A}}}_\b{\mathcal B})^\circ \to
S^*({\mathcal T}^{\widetilde{{\mathcal A}}}_\b{\mathcal B})^\circ$
is an isomorphism, for all $\b\in {\mathcal B}$, we obtain that
$\{\mu^a_\alpha\theta^\alpha\}_{a=1,\dots,r}$ is a local basis of
sections of $S^*({\mathcal T}^{\widetilde{{\mathcal A}}}{\mathcal
B})^\circ\to {\mathcal B}$, where
$\theta^\alpha=\tilde{T}^\alpha-y^\alpha\phi_0$.

Consequently, the constrained motion equations (\ref{LdA}) can be
written as
\begin{equation}\label{LdA2}
\left\{ \begin{array}{rcl} i_X\Omega_L&=&\lambda_a\mu^a_\alpha\theta^\alpha,\\
i_X\phi_0&=&1,\\
(d^{{\mathcal T}^{\widetilde{{\mathcal A}}}{\mathcal
A}}\Psi^a)(X)&=&0,
\end{array} \right.
\end{equation}
along the points of ${\mathcal B}$, where $\lambda_a$ are some
Lagrange multipliers to be determined.

Now, a section $X\in\Gamma(\tau_{\widetilde{{\mathcal
A}}}^{\tau_{\mathcal A}})$ is of the form
$X=z^0\tilde{T}_0+z^\alpha\tilde{T}_\alpha+v^\alpha\tilde{V}_\alpha$.
If we assume that the Lagrangian $L$ is regular, then the first
two equations in (\ref{LdA2}) imply that any solution $X$ has to
be a SODE, that is,
$X=\tilde{T}_0+y^\alpha\tilde{T}_\alpha+X^\alpha\tilde{V}_\alpha$
and, then, the first equation in (\ref{LdA2}) becomes (see
(\ref{PCloc}))
$$X^\alpha\displaystyle\frac{\partial^2L}{\partial y^\alpha\partial
y^\beta}+(\rho^i_0+y^\gamma\rho^i_\gamma)\frac{\partial^2L}{\partial
x^i\partial y^\beta}-\rho_\beta^i\frac{\partial L}{\partial
x^i}-(C_{0\beta}^\gamma+y^\nu C_{\nu\beta}^\gamma)\frac{\partial
L}{\partial y^\gamma}+\lambda_a\mu_\beta^a=0.$$

Thus, if $R_L$ is the Euler-Lagrange section associated with $L$,
$X=R_L-W^{\alpha\beta}\lambda_a\mu_\beta^a\tilde{V}_\alpha$ and
then, the third equation in (\ref{LdA2}) implies that
$$(d^{{\mathcal T}^{\widetilde{{\mathcal A}}}{\mathcal
A}}\Psi^a)(R_L)+\lambda_b(d^{{\mathcal T}^{\widetilde{{\mathcal
A}}}{\mathcal
A}}\Psi^a)(-W^{\alpha\beta}\mu_\beta^b\tilde{V}_\alpha)=0,\makebox[1.5cm]{for
all} a\in\{1,\dots,r\}.$$

As a consequence, we get that there exists a unique solution of
the Lagrange-d'Alembert equations if and only if the matrix
\begin{equation}\label{Cij1}
{\mathcal
C}^{ab}(\b)=(-W^{\alpha\beta}\mu_\beta^b\mu_\alpha^a)(\b)
\end{equation}
is regular, for all $\b\in {\mathcal B}$.

From (\ref{LdA2}), we deduce that the differential equations for
the integral curves of the vector field
$\rho_{\widetilde{{\mathcal A}}}^{\tau_{\mathcal B}}(X)$ are the
Lagrange-d'Alembert differential equations, which read
\begin{equation}\label{3.4'}
\left\{ \begin{array}{c}\displaystyle\frac{dx^i}{dt}=\rho_0^i+\rho_\alpha^iy^\alpha,\\[8pt]
\displaystyle\frac{d}{dt}\left(\frac{\partial L}{\partial
y^\alpha}\right)-\rho_\alpha^i\frac{\partial L}{\partial
x^i}+(C_{\alpha0}^\gamma+ C_{\alpha\beta}^\gamma
y^\beta)\frac{\partial L}{\partial y^\gamma}=-\lambda_a\mu^a_\alpha,\\[8pt]
\mu^a_0+\mu^a_\alpha y^\alpha=0, \end{array}\right. \end{equation}
\vspace{0.2cm}
 with $i\in\{1,\dots,m\}$, $\alpha\in\{1,\dots,n\}$
and $a\in\{1,\dots,r\}$.

\begin{remark}
{\em  In the above discussion, we obtained a complete set of
differential equations that determine the dynamics. Now, we will
analyze the form of the Lagrange-d'Alembert equations in terms of
adapted coordinates to ${\mathcal B}$.  Consider local coordinates
$(x^i)$ on a open set ${U}$ of $M$ and the local basis
$\{\phi_A\}_{A=1,\dots,n-r}$ of sections of $U_{\mathcal B}$.
Complete it to a local basis of sections $\{e_0,\phi_A, \phi_a\}$
of the vector bundle $\tau^{-1}_{\widetilde{{\mathcal
A}}}(U)\rightarrow U$ adapted to $1_{\mathcal A}$. Then, in
coordinates $(x^i,y^0,y^A,y^a)$ adapted to this basis, the
equations defining the constrained subbundle ${\mathcal B}$ as an
affine subbundle of ${\mathcal A}$ (respectively, as an affine
subbundle of $\widetilde{\mathcal A}$) are $y^a=0$ (respectively,
$y^0=1$ and $y^a=0$). Thus, the Lagrange-d'Alembert differential
equations read as
$$\left\{ \begin{array}{c}\displaystyle\frac{dx^i}{dt}=\rho_0^i+\rho^i_Ay^A,\\[8pt]
\displaystyle\frac{d}{dt}(\frac{\partial L}{\partial
y^A})-\rho_A^i\frac{\partial L}{\partial x^i}+(C_{A0}^\gamma+
C_{AB}^\gamma y^B)\frac{\partial
L}{\partial y^\gamma}=0,\\[8pt]
 y^a=0. \end{array}\right.$$ \vspace{0.2cm}
\hfill$\diamondsuit$}
\end{remark}

\section{Solution of Lagrange-d'Alembert equations}\label{secSol}
In this section, we will perform a precise global analysis of the
existence and uniqueness of the solution of Lagrange-d'Alembert
equations. In what follows, we will assume that the Lagrangian $L$
is regular at least in a neighborhood of ${\mathcal B}$.

\begin{definition} A constrained Lagrangian system $(L,{\mathcal B})$ is said
to be \textbf{regular} if the Lagrange-d'Alembert equations have a
unique solution.
\end{definition}

In order to characterize geometrically these nonholonomic systems
which are regular, we define the vector subbundle
$F\subset{\mathcal T}^{\widetilde{{\mathcal A}}}{\mathcal
A}_{|{\mathcal B}}\to {\mathcal B}$ whose fibre at point $\b\in
{\mathcal B}$ is $F_\b=\flat^{-1}_L(S^*({\mathcal
T}^{\widetilde{{\mathcal A}}}_\b{\mathcal B})^\circ)$, where
$\flat_L:{\mathcal T}^{\widetilde{{\mathcal A}}}{\mathcal
A}\to({\mathcal T}^{\widetilde{{\mathcal A}}}{\mathcal A})^*$ is
the vector bundle isomorphism defined by
\begin{equation}\label{flatL}
\flat_L(X)=i_X\Omega_L+\phi_0(X)\phi_0, \makebox[1.5cm]{for
all}X\in {\mathcal T}^{\widetilde{{\mathcal A}}}{\mathcal A}.
\end{equation} More explicitly,
$$F_\b=\{X\in{\mathcal T}^{\widetilde{{\mathcal A}}}_\b{\mathcal A}/exists \;\chi_a\;
s.t.\; \flat_L(X)=\chi_a\mu^a_\alpha\theta^\alpha|_\b\}.$$ From
the definition, it is clear that the rank of $F$ is $rank
(F)=rank({\mathcal T}^{\widetilde{{\mathcal A}}}{\mathcal
B})^\circ=rank({\mathcal T}^{\widetilde{{\mathcal A}}}{\mathcal
A})-rank({\mathcal T}^{\widetilde{{\mathcal A}}}{\mathcal
B})=rank({\mathcal B}^\circ)=r$. If we consider the sections
$Z_a\in\Gamma(\tau_{\widetilde{{\mathcal A}}}^{\tau_{\mathcal
A}})$ such that $\flat_L(Z_a)=\mu^a_\alpha\theta^\alpha$, with
$a\in\{1,\dots,r\}$, then $\{Z_a\}$ is a local basis of sections
of $F$. Moreover, if $R_L$ is the Euler-Lagrange section
associated with $L$, we have that
$$(i_{Z_a}\Omega_L)(R_L)+\phi_0(Z_a)\phi_0(R_L)=\mu_\alpha^a\theta^\alpha(R_L)=0$$
which implies that $\phi_0(Z_a)=0$. Therefore, $Z_a$ is completely
characterized by the conditions
$$i_{Z_a}\Omega_L=\mu^a_\alpha\theta^\alpha\makebox[1cm]{and}i_{Z_a}\phi_0=0.$$

In addition, using these two conditions, we conclude that the
local expression of $Z_a$ is the following
\begin{equation}\label{Zi}
Z_a=-W^{\alpha\beta}\mu_\beta^a\tilde{V}_\alpha,
\end{equation}
where $(W^{\alpha\beta})$ is the inverse matrix of
$(W_{\alpha\beta})$. Thus, the matrix defined in (\ref{Cij1}) is
just ${\mathcal C}^{ab}=(d^{{\mathcal T}^{\widetilde{{\mathcal
A}}}{\mathcal
A}}\Psi^a)(Z_b)=-W^{\alpha\beta}\mu_\beta^a\mu_{\alpha}^b$ and we
will denote it by ${\mathcal C}=({\mathcal C}^{ab})$.

A second important geometric object is the vector subbundle
$G\subset{\mathcal T}^{\widetilde{{\mathcal A}}}{\mathcal
A}_{|{\mathcal B}}\to {\mathcal B}$ whose annihilator at a point
$\b\in {\mathcal B}$ is $S^*({\mathcal T}^{\widetilde{{\mathcal
A}}}_\b{\mathcal B})^\circ$. We can consider the subbundle
$G^\perp\subset{\mathcal T}^{\widetilde{{\mathcal A}}}{\mathcal
A}_{|{\mathcal B}}\to {\mathcal B}$, the orthogonal to $G$ with
respect to the cosymplectic structure $(\Omega_L,\phi_0)$, which
is given by
$$G^\perp_\b=\{X\in{\mathcal
T}^{\widetilde{{\mathcal A}}}_\b{\mathcal
A}/(i_X\Omega_L+\phi_0(X)\phi_0)_{|G_\b}\equiv 0\},\mbox{ for
}\b\in{\mathcal B}.$$ Note that
$G_\b^\perp=\flat_L^{-1}(G_\b^\circ)=F_\b$, for all $\b\in
{\mathcal B}$. Moreover, we obtain
\begin{proposition}\label{propcoisot} $G_\b$ is coisotropic in
$({\mathcal T}^{\widetilde{{\mathcal A}}}_\b{\mathcal
A},\Omega_L(\b),\phi_0(\b))$, for all $\b\in {\mathcal B}$. In
other words, $G^\perp_\b\subseteq G_\b$, for $\b\in{\mathcal B}$.
\end{proposition}
\begin{proof} In fact, using (\ref{Zi}) and since $F$ (respectively, $G^\circ$) is locally generated by
$\{Z_a\}_{a=1,\dots,r}$ (respectively, $\{S^*(d^{{\mathcal
T}^{\widetilde{{\mathcal A}}}{\mathcal
A}}\Psi^a)\}_{a=1,\dots,r}$), we deduce that
$$G_\b^\perp=F_\b\subset G_\b,\makebox{ for all }\b\in {\mathcal
B}.$$
\end{proof}

Now, we introduce the section $\Lambda_L$ of the vector bundle
$\wedge^2({\mathcal T}^{\widetilde{\mathcal A}}{\mathcal A})\to
{\mathcal A}$ defined by
\begin{equation}\label{LambdaL}
\Lambda_L(\alpha,\beta)=\Omega_L(\flat_L^{-1}(\alpha),\flat_L^{-1}(\beta)),
\end{equation}
for $\alpha,\beta\in({\mathcal T}^{\widetilde{\mathcal
A}}{\mathcal A})^*$. $\Lambda_L$ is the (algebraic) Poisson
structure (on the vector bundle ${\mathcal T}^{\widetilde{\mathcal
A}}{\mathcal A}\to {\mathcal A}$) associated with the cosymplectic
structure $(\Omega_L,\phi_0)$. We will denote by
$\sharp_{\Lambda_L}:({\mathcal T}^{\widetilde{{\mathcal
A}}}{\mathcal A})^*\to{\mathcal T}^{\widetilde{{\mathcal
A}}}{\mathcal A}$ the vector bundle morphism given by
$$\sharp_{\Lambda_L}(\alpha)=i_{\alpha}\Lambda_L,\mbox{ for }\alpha\in({\mathcal T}^{\widetilde{{\mathcal A}}}{\mathcal
A})^*.$$ Note that
$$\phi_0(\sharp_{\Lambda_L}(\alpha))=\Omega_L(\flat_L^{-1}(\alpha),R_L)=0.$$
On the other hand, it is clear that
$$\alpha=\flat_L(\flat_L^{-1}(\alpha))=i_{\flat_L^{-1}(\alpha)}\Omega_L+\phi_0(\flat_L^{-1}(\alpha))\phi_0$$
and, thus
\begin{equation}\label{4.3'}
\alpha(R_L)=\phi_0(\flat_L^{-1}(\alpha)). \end{equation} Moreover,
from (\ref{LambdaL}), we have that
$$\beta(\sharp_{\Lambda_L}(\alpha))=-(i_{\flat_L^{-1}(\beta)}\Omega_L)(\flat_L^{-1}(\alpha))=-\beta(\flat_L^{-1}(\alpha))+\phi_0(\flat_L^{-1}(\beta))\phi_0(\flat_L^{-1}(\alpha))$$
which implies that (see (\ref{4.3'}))
$$\beta(\sharp_{\Lambda_L}(\alpha))=\beta(-\flat_L^{-1}(\alpha)+\alpha(R_L)R_L).$$
Therefore, we have proved that
\begin{equation}\label{sosbem}
\sharp_{\Lambda_L}(\alpha)=-\flat_L^{-1}(\alpha)+\alpha(R_L)R_L,\mbox{
for }\alpha\in({\mathcal T}^{\widetilde{{\mathcal A}}}{\mathcal
A})^*.
\end{equation}

Next, we consider the vector subbundle
$G^{\perp,\Lambda_L}\subset{\mathcal T}^{\widetilde{{\mathcal
A}}}{\mathcal A}_{|{\mathcal B}}\to {\mathcal B}$, the orthogonal
to $G$ with respect to the Poisson structure
$(\Lambda_L)_{|{\mathcal B}}$, which is given by
$$G^{\perp,\Lambda_L}_\b=\sharp_{\Lambda_L}(G^\circ_\b),\mbox{ for }\b\in{\mathcal B}.$$
From (\ref{sosbem}) and since $S^*(d^{{\mathcal
T}^{\widetilde{{\mathcal A}}}{\mathcal A}}\Psi^a)(R_L)=0$, we
deduce that
$$\sharp_{\Lambda_L}(S^*(d^{{\mathcal
T}^{\widetilde{{\mathcal A}}}{\mathcal A}}\Psi^a)_{|{\mathcal B}}
)=-Z_a.$$

Consequently, we obtain that
$$G^{\perp,\Lambda_L}_\b=G_\b^\perp=F_\b,\mbox{ for }\b\in
{\mathcal B}.$$

Now, we will consider the vector subbundle $H\subset{\mathcal
T}^{\widetilde{{\mathcal A}}}{\mathcal A}_{|{\mathcal B}}\to
{\mathcal B}$ whose fibre at point $\b\in {\mathcal B}$ is
$H_\b={\mathcal T}^{\widetilde{{\mathcal A}}}_\b{\mathcal B}\cap
G_\b$. We will denote, as above, by $H^\perp$ the orthogonal to
$H$ with respect to the cosymplectic structure $(\Omega_L,\phi_0)$
and by $H^{\perp,\Lambda_L}$ the orthogonal to $H$ with respect to
the Poisson structure $\Lambda_L$.

Using (\ref{verend}) and (\ref{PCloc}), we have that
$i_S\Omega_L=0$, that is,
$$\Omega_L(SX,Y)+\Omega_L(X,SY)=0,\makebox[1.5cm]{for
all}X,Y\in\Gamma(\tau_{\widetilde{{\mathcal A}}}^{\tau_{\mathcal
A}}),$$ or, equivalently,
\begin{equation}\label{SOmega}
i_{SX}\Omega_L=-S^*(i_X\Omega_L),\makebox[1.5cm]{for
all}X\in\Gamma(\tau_{\widetilde{{\mathcal A}}}^{\tau_{\mathcal
A}}).
\end{equation}
Next, denote by $grad\;\Psi^a$ the gradient section corresponding
to the function $\Psi^a$ with respect to the cosymplectic
structure $(\Omega_L,\phi_0)$, that is,
$grad\;\Psi^a=\flat_L^{-1}(d^{{\mathcal T}^{\widetilde{{\mathcal
A}}}{\mathcal A}}\Psi^a)$. Then, using (\ref{SOmega}) and the fact
that $S^*\phi_0=0$, we have that $i_{(S\,grad\;\Psi^a)_{|{\mathcal
B}}}\Omega_L=-i_{Z_a}\Omega_L$ and
$\phi_0(S\,grad\;\Psi^a)=\phi_0(Z_a)=0$ which implies that
\begin{equation}\label{Relation}
(S\,grad\;\Psi^a)_{|{\mathcal B}}=-Z_a,\mbox{ for all
}a\in\{1,\dots,r\}.
\end{equation}

Since the sections $\{d^{{\mathcal T}^{\widetilde{{\mathcal
A}}}{\mathcal A}}\Psi^a\}$ are linearly independent, it follows
that the sections $\{grad\;\Psi^a\}$ also are linearly
independent. Thus, using (\ref{Relation}), we conclude that the
sections $\{(grad\;\Psi^a)_{|{\mathcal B}}, Z_a\}$ are linearly
independent. On the other hand,
\begin{equation}\label{H0}
H^\circ=({\mathcal T}^{\widetilde{\mathcal A}}{\mathcal
B})^\circ+G^\circ
\end{equation}
and, therefore, the space of sections of $H^\perp$ is locally
generated by $\{(grad\;\Psi^a)_{|{\mathcal B}}, Z_a\}$. This
implies that $\{(grad\;\Psi^a)_{|{\mathcal B}}, Z_a\}$ is a local
basis of sections of $H^\perp$ and $rank(H^\perp)=corank(H)=2r$.

Now, using (\ref{H0}), we deduce that
$$H^{\perp,\Lambda_L}=({\mathcal T}^{\widetilde{{\mathcal
A}}}{\mathcal B})^{\perp,\Lambda_L}+F.$$

On the other hand, we introduce the Hamiltonian section
$X_{\Psi^a}^{\Lambda_L}$ associated with the function $\Psi^a$
with respect to $\Lambda_L$ which is defined by
\begin{equation}\label{sHLambdaL}
X_{\Psi^a}^{\Lambda_L}=-\sharp_{\Lambda_L}(d^{{\mathcal
T}^{\widetilde{{\mathcal A}}}{\mathcal A}}\Psi^a),\mbox{ for
}a\in\{1,\dots,r\}.
\end{equation}

From (\ref{sosbem}), we have that
\begin{equation}\label{bLX}
\flat_L(X_{\Psi^a}^{\Lambda_L})=d^{{\mathcal
T}^{\widetilde{{\mathcal A}}}{\mathcal A}}\Psi^a-(d^{{\mathcal
T}^{\widetilde{{\mathcal A}}}{\mathcal A}}\Psi^a)(R_L)\phi_0.
\end{equation}

Thus, using that the sections $\{d^{{\mathcal
T}^{\widetilde{{\mathcal A}}}{\mathcal A}}\Psi^a,\phi_0\}$ are
independent, it follows that the sections
$\{X_{\Psi^a}^{\Lambda_L}\}$ are also independent. Moreover, from
(\ref{bLX}), we obtain that
$$X_{\Psi^a}^{\Lambda_L}=grad\;\Psi^a-(d^{{\mathcal
T}^{\widetilde{{\mathcal A}}}{\mathcal A}}\Psi^a)(R_L)R_L$$ which
implies that $$S\,X_{\Psi^a}^{\Lambda_L}=-Z_a,\mbox{ for
}a\in\{1,\dots,r\}.$$

Therefore, $\{X_{\Psi^a}^{\Lambda_L},Z_a\}$ is a local basis of
sections of $H^{\perp,\Lambda_L}$ and
$rank(H^{\perp,\Lambda_L})=corank(H)=2r$.

The relation among these objects is described by the next result.

\begin{theorem}\label{Threg} The following properties are equivalent:
\begin{enumerate}
\item[(1)] The constrained Lagrangian system $(L,{\mathcal B})$ is
regular, \item[(2)] the matrices ${\mathcal C}=({\mathcal
C}^{ab})$ are non-singular, \item[(3)]\label{th3} ${\mathcal
T}^{\widetilde{{\mathcal A}}}_\b{\mathcal B}\cap
F_\b=\{0\},\makebox[1.5cm]{for all} \b\in {\mathcal B},$
\item[(4)]\label{th4} $H_\b\cap H_\b^{\perp}=\{0\}$, for all
$\b\in {\mathcal B}$, \item[(5)]\label{th5} $H_\b\cap
H_\b^{\perp,\Lambda_L}=\{0\}$, for all $\b\in {\mathcal B}$.
\end{enumerate}
\end{theorem}
\begin{proof} $[(1)\Leftrightarrow(2)]$ This result was proved in Section \ref{Noho}.

$[(2)\Leftrightarrow(3)]$ $(\Rightarrow)$ Suppose that ${\mathcal
C}$ is non-singular and let be $X\in {\mathcal
T}^{\widetilde{{\mathcal A}}}_\b{\mathcal B}\cap F_\b$. Thus,
$X=\sum_{a=1}^r\lambda_aZ_a(\b)$ and $(d^{{\mathcal
T}^{\widetilde{{\mathcal A}}}{\mathcal A}}\Psi^a)(\b)(X)=0$, for
all $a\in\{1,\dots,r\}$, which implies that
\[
\sum_{c=1}^r\lambda_c(d^{{\mathcal T}^{\widetilde{{\mathcal
A}}}{\mathcal A}}\Psi^a)(\b)(Z_c(\b))=\sum_{c=1}^r\lambda_c
{\mathcal C}^{ac}=0.
\]
Therefore, we deduce that $\lambda_c=0$, for all $c$, and
consequently $X=0$.

$(\Leftarrow)$ Conversely, take an arbitrary linear combination of
columns of ${\mathcal C}$ at some point $\b$ such that
$$\displaystyle\sum_{c=1}^{r}\lambda_c{\mathcal C}^{ac}(\b)=0,\mbox{ for all }a\in\{1,\dots,r\}.$$
Thus, $\sum_{c=1}^r\lambda_cZ_c(\b)\in{\mathcal
T}^{\widetilde{{\mathcal A}}}_\b{\mathcal B}$ which implies that
$\sum_{c=1}^r\lambda_cZ_c(\b)=0$, and hence $\lambda_c=0$, for all
$c\in\{1,\dots,r\}$.

$[(3)\Leftrightarrow(4)]$ $(\Rightarrow)$ Let $\b\in {\mathcal
B}$, then we have that ${\mathcal T}^{\widetilde{{\mathcal
A}}}_\b{\mathcal B}\cap F_\b={\mathcal T}^{\widetilde{{\mathcal
A}}}_\b{\mathcal B}\cap G_\b^{\perp}=\{0\}$ and ${\mathcal
T}^{\widetilde{{\mathcal A}}}_\b{\mathcal A}={\mathcal
T}^{\widetilde{{\mathcal A}}}_\b{\mathcal B}\oplus G_\b^{\perp}$
(note that $dim\, {\mathcal T}_\b^{\widetilde{\mathcal
A}}{\mathcal B}+dim\, G_\b^\perp=dim\, {\mathcal
T}_\b^{\widetilde{\mathcal A}}{\mathcal A}$). Hence, from
Proposition \ref{propcoisot}, we obtain that
$$G_\b=({\mathcal T}^{\widetilde{{\mathcal A}}}_\b{\mathcal B}\cap G_\b)\oplus
G_\b^{\perp} =H_\b\oplus G_\b^{\perp}.$$ If $X\in H_\b\cap
H_\b^{\perp}$ then, from the above decomposition, we also have
that $X\in G_\b^\perp$ and, thus, $X=0$.

$(\Leftarrow)$ Conversely, let $\b\in {\mathcal B}$ and take an
element $X\in{\mathcal T}^{\widetilde{{\mathcal A}}}_\b{\mathcal
B}\cap F_\b={\mathcal T}^{\widetilde{{\mathcal A}}}_\b{\mathcal
B}\cap G_\b^\perp\subset{\mathcal T}^{\widetilde{{\mathcal
A}}}_\b{\mathcal B}\cap G_\b=H_\b$. Since, $\Omega_L(\b)(X,Y)=0$
and $\phi_0(\b)(X)=0$, for all $Y\in H_\b$, we conclude that $X\in
H_\b^\perp$ and, therefore, $X=0$.

$[(4)\Leftrightarrow(5)]$ $(\Rightarrow)$ Let $\b\in {\mathcal
B}$, then we have that $H_\b\cap H^\perp_\b=\{0\}$ and, as a
consequence, the restriction $(\Omega_L^{H_\b},\phi_0^{H_\b})$ to
$H_\b$ of the cosymplectic structure $(\Omega_L,\phi_0)$ is a
cosymplectic structure on $H_\b$. In other words, the map
$$\flat_L^{H_\b}:H_\b\to H_\b^*,\;\;X\in
H_\b\mapsto(i_X\Omega_L(\b)+\phi_0(\b)(X)\phi_0(\b))_{|H_\b}\in
H_\b^*,$$ is a linear isomorphism. Thus, we can consider the Reeb
vector $R^H(\b)$ of the cosymplectic vector space
$(H_\b,\Omega_L^{H_\b},\phi_0^{H_\b})$ which is characterized by
the following conditions
\begin{equation}\label{RH}
i_{R^H(\b)}\Omega_L^{H_\b}=0\;\mbox{ and
}\;i_{R^H(\b)}\phi_0^{H_\b}=1.
\end{equation}
Now, we will prove that
$H_\b\cap\sharp_{\Lambda_L}(H^\circ_\b)=\{0\}$.

Suppose that $\alpha\in H^\circ_\b$ and
$\sharp_{\Lambda_L}(\alpha)\in H_\b$. Then, using (\ref{sosbem})
and the fact that $\phi_0(\b)(\sharp_{\Lambda_L}(\alpha))=0$, we
deduce that
\begin{equation}\label{sosLambL}
i_{\sharp_{\Lambda_L}(\alpha)}\Omega_L(\b)=-\alpha+\alpha(R_L(\b))\phi_0(\b).
\end{equation}
Therefore, from (\ref{RH}) and since $\alpha\in H_\b^\circ$ and
$\sharp_{\Lambda_L}(\alpha)\in H_\b$, we have that
$$0=(i_{\sharp_{\Lambda_L}(\alpha)}\Omega_L(\b))(R^H(\b))=\alpha(R_L(\b))$$
which implies that (see (\ref{sosLambL}))
$$\flat_L^{H_\b}(\sharp_{\Lambda_L}(\alpha))=0.$$
Consequently, $\sharp_{\Lambda_L}(\alpha)=0$.

$(\Leftarrow)$ Let $\b\in {\mathcal B}$ and take an element $X\in
H_\b\cap H^\perp_\b$. Then
$$(i_X\Omega_L(\b)+\phi_0(\b)(X)\phi_0(\b))_{|H_\b}=0.$$

Thus, using that $X\in H_\b$, it follows that
$$\phi_0(\b)(X)=0\;\mbox{ and }\;\alpha=i_X\Omega_L(\b)\in H_\b^\circ.$$

Therefore, $\alpha(R_L(\b))=0$ and $\alpha=\flat_L(X)$. This
implies that (see (\ref{sosbem}))
$$X=-\sharp_{\Lambda_L}(\alpha)\in\sharp_{\Lambda_L}(H_\b^\circ)=H^{\perp,\Lambda_L}_\b.$$

Consequently, $X=0$.

\end{proof}


\begin{proposition} Conditions $(3)$, $(4)$ and $(5)$ in Theorem
\ref{Threg} are equivalent, respectively, to
\begin{enumerate}
\item[(3')] ${\mathcal T}^{\widetilde{{\mathcal A}}}{\mathcal
A}_{|{\mathcal B}}={\mathcal T}^{\widetilde{{\mathcal
A}}}{\mathcal B}\oplus F,$ \item[(4')] ${\mathcal
T}^{\widetilde{{\mathcal A}}}{\mathcal A}_{|{\mathcal B}}=H\oplus
H^\perp,$ \item[(5')] ${\mathcal T}^{\widetilde{{\mathcal
A}}}{\mathcal A}_{|{\mathcal B}}=H\oplus H^{\perp,\Lambda_L}.$
\end{enumerate}
\end{proposition}
\begin{proof} The equivalence of $(3)$ and $(3')$ follows
by computing the dimension of the corresponding spaces. The ranks
of ${\mathcal T}^{\widetilde{{\mathcal A}}}{\mathcal A}$,
${\mathcal T}^{\widetilde{{\mathcal A}}}{\mathcal B}$ and $F$ are
$$\begin{array}{l}
rank({\mathcal T}^{\widetilde{{\mathcal A}}}{\mathcal A})=2\;rank({\mathcal A})+1,\\
rank({\mathcal T}^{\widetilde{{\mathcal A}}}{\mathcal B})=rank({\mathcal A})+rank({\mathcal B})+1,\\
rank(F)=rank({\mathcal B}^\circ)=rank({\mathcal A})-rank({\mathcal
B}).
\end{array}$$
Thus, $rank({\mathcal T}^{\widetilde{{\mathcal A}}}{\mathcal
A})=rank({\mathcal T}^{\widetilde{{\mathcal A}}}{\mathcal
B})+rank(F)$, and the result follows. The equivalence between
$(4)$ and $(4')$ is obvious, since we are assuming that the free
Lagrangian is regular, i.e., $(\Omega_L,\phi_0)$ is a cosymplectic
structure on ${\mathcal T}^{\widetilde{{\mathcal A}}}{\mathcal
A}$. Finally, the equivalence of $(5)$ and $(5')$ is also obvious
because we have that $rank(H^{\perp,\Lambda_L})=corank(H)$.

\end{proof}

\subsection{Projectors}

Now, we can express the constrained dynamical section in terms of
the free dynamical section by projecting to the adequate space,
either ${\mathcal T}^{\widetilde{{\mathcal A}}}{\mathcal B}$ or
$H$, according to each of the above decompositions of ${\mathcal
T}^{\widetilde{{\mathcal A}}}{\mathcal A}_{|{\mathcal B}}$. Of
course, both procedures give the same result.

\subsubsection{ Projection to ${\mathcal
T}^{\widetilde{{\mathcal A}}}{\mathcal B}$} Assuming that the
constrained Lagrangian system is regular, we have a direct sum
decomposition
$${\mathcal T}^{\widetilde{{\mathcal A}}}_\b{\mathcal A}={\mathcal
T}^{\widetilde{{\mathcal A}}}_\b{\mathcal B}\oplus
F_\b,\makebox[1.5cm]{for all}\b\in {\mathcal B},$$ where we recall
that the subbundle $F\subset{\mathcal T}^{\widetilde{{\mathcal
A}}}{\mathcal A}$ is defined by $F=\flat^{-1}_L(S^*({\mathcal
T}^{\widetilde{{\mathcal A}}}{\mathcal B})^\circ)$. We will denote
by $P$ and $Q$ the complementary projectors defined by this
decomposition, that is,
$$P_\b:{\mathcal T}^{\widetilde{{\mathcal A}}}_\b{\mathcal A}\to {\mathcal
T}^{\widetilde{{\mathcal A}}}_\b{\mathcal
B}\makebox[1cm]{and}Q_\b:{\mathcal T}^{\widetilde{{\mathcal
A}}}_\b{\mathcal A}\to F_\b,\makebox[1.5cm]{for all}\b\in
{\mathcal B}.$$

\begin{theorem}
Let $(L,{\mathcal B})$ be a regular constrained Lagrangian system
and $R_L$ be the solution of the free dynamics, i.e.,
$i_{R_L}\Omega_L=0$ and $i_{R_L}\phi_0=1$. Then, the solution of
the constrained dynamics is the SODE $R_{nh}$ obtained by
projection $R_{nh}=P({R_L}_{|{\mathcal B}})$.
\end{theorem}
\begin{proof}
Indeed, if we write $R_{nh}(\b)=R_L(\b)-Q_\b(R_L(\b))$, for $\b\in
{\mathcal B}$, then we have
$$i_{R_{nh}(\b)}\Omega_L(\b)=i_{R_L(\b)}\Omega_L(\b)-i_{Q_\b(R_L(\b))}\Omega_L(\b)=-i_{Q_\b(R_L(\b))}\Omega_L(\b)$$
 which is an element
of $S^*({\mathcal T}^{\widetilde{{\mathcal A}}}_\b{\mathcal
B})^\circ$, because $Q_\b(R_L(\b))\in F_\b$ and
$\phi_0(Q_\b(R_L(\b)))=0$. Moreover, using this last fact and
since $R_L$ is a SODE and $S(Q(R_L))=0$, we have that $R_{nh}$ is
also a SODE.
\end{proof}

Let $(x^i)$ be local coordinates on an open subset $U$ of $M$ and
$\{e_0,e_\alpha\}$ be a local basis of sections of the vector
bundle $\tau_{\widetilde{{\mathcal A}}}^{-1}(U)\to U$ adapted to
$1_{\mathcal A}$. Denote by $(x^i,y^0,y^\alpha)$ the corresponding
local coordinates on $\widetilde{\mathcal A}$ and suppose that the
equations defining the constrained subbundle ${\mathcal B}$ as an
affine subbundle of ${\mathcal A}$ are
$$\Psi^a=\mu_0^a+\mu_\alpha^ay^\alpha=0,\mbox{ for
}a\in\{1,\dots,r\}.$$ Then, the local expression of the projector
over ${\mathcal T}^{\widetilde{{\mathcal A}}}{\mathcal B}$ is
$$P_\b=Id-\sum_{1\leq a,c\leq r}{\mathcal C}_{ac}(\b)Z_c(\b)\otimes(d^{{\mathcal
T}^{\widetilde{{\mathcal A}}}{\mathcal A}}\Psi^a)(\b),$$ for all
$\b\in {\mathcal B}$, $({\mathcal C}_{ab})$ being the inverse
matrix of $({\mathcal C}^{ab})$.

Hence, if the constrained Lagrangian system is regular, the
solution of the constrained dynamics is
$$R_{nh}={R_L}_{|{\mathcal B}}-\sum_{1\leq a,c\leq r}{\mathcal
C}_{ac}\rho_{\widetilde{{\mathcal A}}}^{\tau_{\mathcal
A}}(R_L)(\Psi^a){Z_c}_{|{\mathcal B}}.$$

From the regularity of the local matrices ${\mathcal C}$ we deduce
that $(P,Q)$ may be extended (in many ways) to an open
neighborhood of ${\mathcal B}$. Therefore, $R_{nh}$ may also be
extended to an open neighborhood of ${\mathcal B}$. This fact will
be used in the following proposition.

\begin{proposition} Let $(L,{\mathcal B})$ be a regular constrained Lagrangian
system, $\Theta_L$ be the Poincar\'{e}-Cartan 1-section and
${\mathcal L}^{{\mathcal T}^{\widetilde{\mathcal A}}{\mathcal A}}$
be the Lie derivate on ${\mathcal T}^{\widetilde{\mathcal
A}}{\mathcal A}$.
\begin{enumerate}
\item[i)] If $R_{nh}$ is the solution of the constrained dynamics
$${\mathcal L}^{{\mathcal T}^{\widetilde{{\mathcal A}}}{\mathcal
A}}_{R_{nh}}\Theta_L=d^{{\mathcal T}^{\widetilde{{\mathcal
A}}}{\mathcal A}}L-{\mathcal L}^{{\mathcal
T}^{\widetilde{{\mathcal A}}}{\mathcal A}}_{Q(R_L)}\Theta_L.$$
\item[ii)] We have that
$$({\mathcal L}^{{\mathcal T}^{\widetilde{{\mathcal
A}}}{\mathcal A}}_{Q(R_L)}\Theta_L)_{|{\mathcal B}}\in
\Gamma(\tau_{S^*({\mathcal T}^{\widetilde{{\mathcal A}}}{\mathcal
B})^\circ}).$$
\end{enumerate}

\end{proposition}
\begin{proof}
$(i)$ It follows since $R_{nh}=P(R_L)=R_L-Q(R_L)$ and ${\mathcal
L}^{{\mathcal T}^{\widetilde{{\mathcal A}}}{\mathcal
A}}_{R_L}\Theta_L=d^{{\mathcal T}^{\widetilde{{\mathcal
A}}}{\mathcal A}}L$.

$(ii)$ Since $Q(R_L)=\sum_{a=1}^r\Lambda_aZ_a$, with
$\Lambda_a={\mathcal C}_{ac}\rho_{\widetilde{{\mathcal
A}}}^{\tau_{\mathcal A}}(R_L)(\Psi^c)$, and using (\ref{PCloc})
and (\ref{Zi}), we deduce that
$${\mathcal L}^{{\mathcal T}^{\widetilde{{\mathcal A}}}{\mathcal A}}_{Q(R_L)}\Theta_L=i_{\sum\Lambda_aZ_a}d^{{\mathcal
T}^{\widetilde{{\mathcal A}}}{\mathcal A}}\Theta_L+d^{{\mathcal
T}^{\widetilde{{\mathcal A}}}{\mathcal
A}}i_{\sum\Lambda_aZ_a}\Theta_L=-\sum_{a=1}^r\Lambda_ai_{Z_a}\Omega_L=-\Lambda_a\mu_\alpha^a\theta^\alpha.$$

Thus, $({\mathcal L}^{{\mathcal T}^{\widetilde{{\mathcal
A}}}{\mathcal A}}_{Q(R_L)}\Theta_L)_{|{\mathcal B}}\in
\Gamma(\tau_{S^*({\mathcal T}^{\widetilde{{\mathcal A}}}{\mathcal
B})^\circ}).$
\end{proof}

\subsubsection{Cosymplectic Projection to $H$} We have seen
that the regularity condition for the cons\-trained system
$(L,{\mathcal B})$ can be equivalently expressed by requiring that
the subbundle $H$ is a cosymplectic subbundle of $({\mathcal
T}^{\widetilde{{\mathcal A}}}{\mathcal A},\Omega_L,\phi_0)$. Thus,
the restriction $(\Omega_L^H,\phi_0^H)$ to $H$ of the cosymplectic
structure $(\Omega_L,\phi_0)$ is a cosymplectic structure on $H$.
Therefore, there exists a unique solution on $H$ of the equations
\begin{equation}\label{cosH}
i_X\Omega_L^H=0\makebox[1cm]{and}i_X\phi_0^H=1.
\end{equation}

\begin{proposition}\label{prop4.7} The solution of equations (\ref{cosH}) is
precisely the solution of the constrained dynamics $R_{nh}$.
\end{proposition}
\begin{proof} Since $R_{nh}$ is the solution of the constrained
dynamics, then $R_{nh}(\b)\in{\mathcal T}^{\widetilde{{\mathcal
A}}}_\b{\mathcal B}$, for all $\b\in {\mathcal B}$. Moreover,
using that $R_{nh}$ is a SODE, we obtain that $R_{nh}(\b)\in
G_\b$. Then, $R_{nh}(\b)\in H_\b$ and it is obvious that it
verifies (\ref{cosH}). Note that $i_{R_{nh}}\Omega_L(\b)\in
S^*({\mathcal T}^{\widetilde{{\mathcal A}}}_\b{\mathcal
B})^\circ=G^\circ_\b\subseteq H^\circ_\b$, for all $\b\in{\mathcal
B}$.
\end{proof}

On the other hand, we have a direct sum decomposition
$${\mathcal T}^{\widetilde{{\mathcal A}}}_\b{\mathcal A}=H_\b\oplus H_\b^\perp,$$
for all $\b\in {\mathcal B}$. We will denote by ${\mathcal P}$ and
${\mathcal Q}$ the complementary projectors defined by this
decomposition, that is,
$${\mathcal P}_\b:{\mathcal T}^{\widetilde{{\mathcal A}}}_\b{\mathcal A}\to H_\b\makebox[1cm]{and}{\mathcal Q}_\b:{\mathcal
T}^{\widetilde{{\mathcal A}}}_\b{\mathcal A}\to
H_\b^\perp,\makebox[1.5cm]{for all}\b\in {\mathcal B}.$$

Then, we have the following result.
\begin{theorem}
Let $(L,{\mathcal B})$ be a regular constrained Lagrangian system
and $R_L$ be the solution of the free dynamics, i.e.,
$i_{R_L}\Omega_L=0$ and $i_{R_L}\phi_0=1$. Then, the solution of
the constrained dynamics is the SODE $R_{nh}$ obtained by
projection $R_{nh}={\mathcal P}({R_L}_{|{\mathcal B}})$.
\end{theorem}
\begin{proof} If $\b\in{\mathcal B}$ and $X\in {\mathcal T}^{\widetilde{{\mathcal A}}}_\b{\mathcal
A}$ then, using that ${\mathcal Q}_\b(R_L(\b))\in H^\perp_\b$, it
follows that
\begin{equation}\label{Clave}
\begin{array}{l}
\Big(i_{{\mathcal P}_\b(R_L(\b))}
\Omega_L^H(\b)+\phi_0^H(\b)({\mathcal
P}_\b(R_L(\b)))\phi_0^H(\b)\Big)({\mathcal P}_\b(X))\\[10pt]
=-\Big(i_{{\mathcal Q}_\b(R_L(\b))}
\Omega_L(\b)+\phi_0(\b)({\mathcal
Q}_\b(R_L(\b)))\phi_0^H(\b)\Big)({\mathcal
P}_\b(X))+\phi_0^H(\b)({\mathcal P}_\b(X))\\[10pt]
=\phi_0^H(\b)({\mathcal P}_\b(X)).
\end{array}
\end{equation}
Thus, we deduce that
$$i_{{\mathcal P}({R_L}_{|{\mathcal B}})}
\Omega_L^H+\phi_0^H({\mathcal P}({R_L}_{|{\mathcal B}}
))\phi_0^H=\phi_0^H.$$ In particular, from Proposition
\ref{prop4.7}, we obtain that
\begin{equation}\label{Firstcond}
1=\phi_0^H(R_{nh})=\phi_0^H({\mathcal P}({R_L}_{|{\mathcal B}})).
\end{equation}
Therefore, using (\ref{Clave}), we conclude that
\begin{equation}\label{Secondcond}
i_{{\mathcal P}({R_L}_{|{\mathcal B}})} \Omega_L^H=0.
\end{equation}
Consequently, from (\ref{Firstcond}) and (\ref{Secondcond}), it
follows that $$R_{nh}={\mathcal P}({R_L}_{|{\mathcal B}}).$$
\end{proof}

Next, we will denote by $\{\cdot,\cdot\}_L$ the Poisson bracket on
${\mathcal A}$ induced by the algebraic Poisson structure
$\Lambda_L$ given by
$$\{\Psi,\Psi'\}_L=-(d^{{\mathcal T}^{\widetilde{\mathcal
A}}{\mathcal
A}}\Psi')(X_\Psi^{\Lambda_L})=\Omega_L(X_\Psi^{\Lambda_L},X_{\Psi'}^{\Lambda_L}),$$
for $\Psi,\Psi'\in C^\infty({\mathcal A})$.

On the other hand, we introduce the Hamiltonian section
$X_{\Psi^a}^L\in\Gamma(\tau_{\widetilde{\mathcal
A}}^{\tau_{\mathcal A}})$ associated with the function $\Psi^a$
with respect to the cosymplectic structure $(\Omega_L,\phi_0)$
which is defined by
$$X_{\Psi^a}^L=\flat_L^{-1}(d^{{\mathcal T}^{\widetilde{\mathcal
A}}{\mathcal A}}\Psi^a)-(d^{{\mathcal T}^{\widetilde{\mathcal
A}}{\mathcal A}}\Psi^a)(R_L)R_L,\mbox{ for }a\in\{1,\dots,r\},$$
where $\flat_L:{\mathcal T}^{\widetilde{\mathcal A}}{\mathcal
A}\to({\mathcal T}^{\widetilde{\mathcal A}}{\mathcal A})^*$ is the
vector bundle isomorphism defined in (\ref{flatL}). Note that,
from (\ref{sosbem}) and (\ref{sHLambdaL}), we obtain that
$$X_{\Psi^a}^L=X_{\Psi^a}^{\Lambda_L}.$$

Then, using that $\{(grad\;\Psi^a)_{|{\mathcal B}}
,Z_a=-S(grad\,\Psi^a)_{|{\mathcal B}}\}_{a=1,\dots,r}$ is a local
basis of sections of $H^\perp$, we deduce that the local
expression of the projector ${\mathcal P}$ is
$$\begin{array}{lcl}
{\mathcal P}&=&\displaystyle{Id-\sum_{1\leq a,b,c,d\leq
r}{\mathcal C}_{ab}{\mathcal
C}_{cd}\Big(\{\Psi^b,\Psi^d\}_L+(d^{{\mathcal
T}^{\widetilde{{\mathcal A}}}{\mathcal
A}}\Psi^b)(R_L)(d^{{\mathcal T}^{\widetilde{{\mathcal
A}}}{\mathcal A}}\Psi^d)(R_L)\Big) Z_a\otimes S^*(d^{{\mathcal
T}^{\widetilde{{\mathcal A}}}{\mathcal A}}\Psi^c)}\\[8pt]
&&\displaystyle{+ \sum_{1\leq a,b\leq r}{\mathcal
C}_{ab}\Big(X_{\Psi^a}^L+(d^{{\mathcal T}^{\widetilde{{\mathcal
A}}}{\mathcal A}}\Psi^a)(R_L)R_L\Big)\otimes S^*(d^{{\mathcal
T}^{\widetilde{{\mathcal A}}}{\mathcal A}}\Psi^b)-\sum_{1\leq
a,b\leq r}{\mathcal C}_{ab}Z_a\otimes(d^{{\mathcal
T}^{\widetilde{{\mathcal A}}}{\mathcal A}}\Psi^b)},
\end{array}
$$
along the points of ${\mathcal B}$.

\subsubsection{Poisson Projection to $H$} Assuming that the
constrained Lagrangian system is regular, we have a direct sum
decomposition
$${\mathcal T}^{\widetilde{{\mathcal A}}}_\b{\mathcal A}=H_\b\oplus
H^{\perp,\Lambda_L}_\b,\makebox[1.5cm]{for all}\b\in {\mathcal
B},$$ where we recall that
$H^{\perp,\Lambda_L}=\sharp_{\Lambda_L}(H^\circ)$. We will denote
by $\tilde{\mathcal P}$ and $\tilde{\mathcal Q}$ the complementary
projectors defined by this decomposition, that is,
$$\tilde{\mathcal P}_\b:{\mathcal T}^{\widetilde{{\mathcal A}}}_\b{\mathcal A}\to
H_\b\makebox[1cm]{and}\tilde{\mathcal Q}_\b:{\mathcal
T}^{\widetilde{{\mathcal A}}}_\b{\mathcal A}\to
H_\b^{\perp,\Lambda_L},\makebox[1.5cm]{for all}\b\in {\mathcal
B}.$$

\begin{theorem}
Let $(L,{\mathcal B})$ be a regular constrained Lagrangian system
and $R_L$ be the solution of the free dynamics. Then, the solution
of the constrained dynamics is the SODE $R_{nh}$ obtained by
projection $R_{nh}=\tilde{\mathcal P}({R_L}_{|{\mathcal B}})$.
\end{theorem}
\begin{proof}Suppose that $\b\in{\mathcal B}$. Since ${\mathcal Q}_\b(R_L(\b))=Q_\b(R_L(\b))\in
F_\b=G^\perp_\b=G^{\perp,\Lambda_L}_\b\subseteq
H^{\perp,\Lambda_L}_\b$, ${\mathcal P}_\b(R_L(\b))\in H_\b$ and
$R_L(\b)={\mathcal P}_\b(R_L(\b))+{\mathcal Q}_\b(R_L(\b))$, we
conclude that
$$\begin{array}{ccccccc}
P_\b(R_L(\b))&=&{\mathcal P}_\b(R_L(\b))&=&\tilde{\mathcal
P}_\b(R_L(\b))&=&R_{nh}(\b),\\
Q_\b(R_L(\b))&=&{\mathcal Q}_\b(R_L(\b))&=&\tilde{\mathcal
Q}_\b(R_L(\b)).&&
\end{array}$$

\end{proof}

 Since also
$\{Z_a, X_{\Psi^a}^{\Lambda_L}\}_{a=1,\dots,r}$ is a local basis
of $H^{\perp,\Lambda_L}$, we obtain that the local expression of
the projector $\tilde{\mathcal P}$ is
\begin{eqnarray*}
\tilde{\mathcal P}&=&Id-\sum_{1\leq a,b,c,d\leq r}{\mathcal
C}_{ab}{\mathcal C}_{cd}\{\Psi^b,\Psi^d\}_L Z_a\otimes
S^*(d^{{\mathcal T}^{\widetilde{{\mathcal {\mathcal {\mathcal
A}}}}}{\mathcal {\mathcal A}}}\Psi^c)\\
&&+ \sum_{1\leq a,b\leq r}{\mathcal
C}_{ab}X_{\Psi^a}^{\Lambda_L}\otimes S^*(d^{{\mathcal
T}^{\widetilde{{\mathcal {\mathcal A}}}}{\mathcal {\mathcal
A}}}\Psi^b)-\sum_{1\leq a,b\leq r}{\mathcal
C}_{ab}Z_a\otimes(d^{{\mathcal T}^{\widetilde{{\mathcal
A}}}{\mathcal A}}\Psi^b), \end{eqnarray*}
 along the points of
${\mathcal B}$.

\subsection{The constrained Poincar\'{e}-Cartan 2-section}
Let $(L,{\mathcal B})$ be a regular constrained Lagrangian system
on a Lie affgebroid ${\mathcal A}$ of rank $n$ with regular
Lagrangian function $L:{\mathcal A}\to \R$ and with constraint
subbundle ${\mathcal B}$ of corank $r$. The equations for the
Lagrange-d'Alembert section $R_{nh}=P({R_L}_{|{\mathcal B}})$ can
be entirely written in terms of objects in the Lie algebroid
prolongation $(\tau_{\widetilde{{\mathcal A}}}^{\tau_{\mathcal
B}}:{\mathcal T}^{\widetilde{{\mathcal A}}}{\mathcal B}\to
{\mathcal B} ,\lcf\cdot,\cdot\rcf_{\widetilde{{\mathcal
A}}}^{\tau_{\mathcal B}},\rho_{\widetilde{{\mathcal
A}}}^{\tau_{\mathcal B}})$ of the Lie algebroid
$(\tau_{\widetilde{{\mathcal A}}}:\widetilde{{\mathcal A}}\to M
,\lcf\cdot,\cdot\rcf_{\widetilde{{\mathcal
A}}},\rho_{\widetilde{{\mathcal A}}})$ over the fibration
$\tau_{\mathcal B}:{\mathcal B}\to M$. In order to do this, for
every point $\b\in {\mathcal B}$, we define
$$\omega(\b)=\Omega_L(\b)-(i_{Q_\b(R_L(\b))}\Omega_L(\b))\wedge\phi_0(\b).$$
$\omega$ is a section of the vector bundle $\wedge^2({\mathcal
T}^{\widetilde{{\mathcal A}}}{\mathcal A})^*_{|{\mathcal B}}\to
{\mathcal B}$. We also have that $\phi_0(\b)\wedge\omega^n(\b)\neq
0$, for all $\b\in {\mathcal B}$. Thus, there exists a unique
section $X$ of ${{\mathcal T}^{\widetilde{{\mathcal A}}}{\mathcal
A}}_{|{\mathcal B}}\to {\mathcal B}$ such that
\begin{equation}\label{reeb1}
i_X\omega=0\makebox[1cm]{and}i_X\phi_0=1.
\end{equation}

Note that $X$ is the solution of the constrained dynamics. In
fact, since $\phi_0(P({R_L}_{|{\mathcal B}}))=1$, it follows that
$$\begin{array}{rcl}
i_{{P}({R_L}_{|{\mathcal B}})}\omega&=&i_{{P}({R_L}_{|{\mathcal
B}})}\Omega_L+i_{{Q}({R_L}_{|{\mathcal
B}})}\Omega_L-(i_{{Q}({R_L}_{|{\mathcal
B}})}\Omega_L)({P}({R_L}_{|{\mathcal
B}})){\phi_0}_{|{\mathcal B}}\\
&=&-(i_{{Q}({R_L}_{|{\mathcal B}})}\Omega_L)({P}({R_L}_{|{\mathcal
B}})){\phi_0}_{|{\mathcal B}}.
\end{array}$$
Thus, using that $Q({R_L}_{|{\mathcal B}})$ is a section of
$F\to{\mathcal B}$ and the fact that $P({R_L}_{|{\mathcal B}})$ is
a SODE, we conclude that $i_{P({R_L}_{|{\mathcal B}})}\omega=0$
which proves that $X=P({R_L}_{|{\mathcal B}})=R_{nh}$.

Next, we get the following.

\begin{theorem} If $\tilde{\omega}$ and $\tilde{\phi_0}$ are the
restrictions of $\omega$ and $\phi_0$ to the vector subbundle
${\mathcal T}^{\widetilde{{\mathcal A}}}{\mathcal B}$ of
${\mathcal T}^{\widetilde{{\mathcal A}}}{\mathcal A}$, then the
solution $R_{nh}=P({R_L}_{|{\mathcal B}})$ of the constrained
dynamics verifies the equations
\begin{equation}\label{reeb2}
i_X\tilde{\omega}=0\makebox[1cm]{and}i_X\tilde{\phi_0}=1.
\end{equation}
Moreover, the unique SODE $X$ on ${\mathcal
T}^{\widetilde{{\mathcal A}}}{\mathcal B}$ satisfying
(\ref{reeb2}) is just $R_{nh}=P({R_L}_{|{\mathcal B}})$.
\end{theorem}
\begin{proof}
Since the section $R_{nh}=P({R_L}_{|{\mathcal B}})$ satisfies
(\ref{reeb1}), then it also verifies (\ref{reeb2}).

Now, let $X$ be a SODE on ${\mathcal T}^{\widetilde{{\mathcal
A}}}{\mathcal B}$ such that $i_X\tilde{\omega}=0$ and
$i_X\tilde{\phi_0}=1$. Then, we have that
\begin{equation}\label{proof1}
(i_X\omega)(P(Y))=0,\makebox[1.5cm]{for all}Y\mbox{ section of }
{{\mathcal T}^{\widetilde{{\mathcal A}}}{\mathcal A}}_{|{\mathcal
B}}\to {\mathcal B}.
\end{equation}

On the other hand, using that $S(Q(Y))=0$, $\phi_0(Q(Y))=0$ and
the fact that $X$ is a SODE, we obtain
\begin{equation}\label{proof2}
(i_X\omega)(Q(Y))=-(i_{Q(Y)}\Omega_L)(X)-(i_{Q(R_L)}\Omega_L)(X)\phi_0(Q(Y))+(i_{Q(R_L)}\Omega_L)(Q(Y))=0,
\end{equation}
for all $Y$ section of ${{\mathcal T}^{\widetilde{{\mathcal
A}}}{\mathcal A}}_{|{\mathcal B}}\to {\mathcal B}$. Finally, from
(\ref{proof1}) and (\ref{proof2}), we conclude that $i_X\omega=0$
which implies that $X=R_{nh}$.
\end{proof}

\begin{definition} The 2-section $\tilde{\omega}$ is said to be
the constrained Poincar\'{e}-Cartan 2-section.
\end{definition}
\begin{remark} {\em Note that $(\tilde{\omega},\tilde{\phi_0})$ is not
a cosymplectic structure on ${\mathcal T}^{\widetilde{{\mathcal
A}}}{\mathcal B}$ so that it may be another solution of the
equations
$$i_X\tilde{\omega}=0\makebox[1cm]{and}i_X\tilde{\phi_0}=1.$$\hfill$\diamondsuit$}
\end{remark}

\section{Constrained Hamiltonian Systems and the nonholonomic bracket}

\subsection{Constrained Hamiltonian Systems}\label{sec5.1}
We now pass to the Hamiltonian description of the nonholonomic
system on a Lie affgebroid. The equivalence between the Lagrangian
and the Hamiltonian description of an unconstrained system was
discussed in Section \ref{sec3.3}.
We now consider a nonholonomic system, described by a hyperregular
Lagrangian $L$ and a constrained affine subbundle $\tau_{\mathcal
B}:{\mathcal B}\to M$ of the bundle $\tau_{\mathcal A}:{\mathcal
A}\to M$. Denote by $\bar{{\mathcal B}}$ the image of ${\mathcal
B}$ under the Legendre transformation, which is a submanifold of
$V^*$. If ${\mathcal B}$ is again locally defined by $r$
independent functions $\Psi^a$, then the constraint functions on
$V^*$ describing $\bar{{\mathcal B}}$ become $\psi^a=\Psi^a\circ
leg_L^{-1}$, i.e.,
$$\psi^a(x^j,y_\alpha)=\Psi^a(x^j,\displaystyle\frac{\partial
H}{\partial y_\alpha}),$$ where the local expression of the
Hamiltonian section is
$h(x^i,y_\alpha)=(x^i,-H(x^j,y_\beta),y_\alpha)$.

A curve $\gamma:I\to V^*$ is a solution of the equations of motion
for the nonholonomic Hamiltonian system $(h,\bar{\mathcal B})$ if
$leg_L^{-1}\circ\gamma:I\to{\mathcal A}$ is a solution of the
Lagrange-d'Alambert equations for the system $(L,{\mathcal B})$.
Thus, using (\ref{PartialsH}), (\ref{legL-1}) and (\ref{3.4'}), we
deduce that
$$\gamma:I\to V^*,\;\;t\mapsto\gamma(t)=(x^i(t),y_\alpha(t)),$$
is a solution of the equations of motion if
$$\left\{ \begin{array}{c}\displaystyle\frac{dx^i}{dt}=\rho_0^i+\rho_\alpha^i\frac{\partial H}{\partial y_\alpha},\\[8pt]
\displaystyle\frac{dy_\alpha}{dt}+\rho_\alpha^i\frac{\partial
H}{\partial x^i}+(C_{\alpha0}^\gamma+
C_{\alpha\beta}^\gamma\frac{\partial H}{\partial
y_\beta})y_\gamma=-\bar{\lambda}_a{\mathcal H}^{\alpha\beta}
\frac{\partial\psi^a}{\partial y_\beta},\\[8pt]
\psi^a(x^j,y_\alpha)=0, \end{array}\right.$$

\vspace{0.2cm} with $i\in\{1,\dots,m\}$, $\alpha\in\{1,\dots,n\}$
and $a\in\{1,\dots,r\}$ and where $\bar{\lambda}_i$ are the
Lagrange multipliers and ${\mathcal H}^{\alpha\beta}$ are the
components of the inverse matrix of the regular matrix $({\mathcal
H}_{\alpha\beta})=\Big(\displaystyle\frac{\partial^2H}{\partial
y_\alpha\partial y_\beta}\Big)$. Note that
\begin{equation}\label{5.0}
(\displaystyle\frac{\partial\psi^a}{\partial y_\alpha}{\mathcal
H}^{\alpha\beta})(x^j,y_\gamma)=(\frac{\partial\Psi^a}{\partial
y^\beta}\circ leg_L^{-1})(x^j,y_\gamma)=(\mu^a_\beta\circ
leg_L^{-1})(x^i,y_\gamma).\end{equation}

An intrinsic description of the equations of motion is obtained as
follows. Consider the subbundle $\bar{G}\subset{\mathcal
T}^{\widetilde{{\mathcal A}}}V^*_{|\bar{{\mathcal
B}}}\to\bar{{\mathcal B}}$ such that
$\bar{G}_{\bar{\b}}=({\mathcal T}leg_L)(G_\b)$, with
$\bar{\b}=leg_L(\b)$ and $\b\in {\mathcal B}$. It is easy to prove
that $\bar{G}_{\bar{\b}}^\circ$ is locally generated by the
independent sections
\begin{equation}\label{mubarraa}
\bar{\mu}^a=\displaystyle\frac{\partial\psi^a}{\partial
y_\alpha}{\mathcal H}^{\alpha\beta}(\tilde{e}^\beta-\frac{\partial
H}{\partial y_\beta}\tilde{e}^0),\makebox[1.5cm]{for
all}a\in\{1,\dots,r\}.
\end{equation}

Note that, from (\ref{Llegloc}) and (\ref{5.0}), it follows that
$$({\mathcal
T}leg_L,leg_L)^*(\bar{\mu}^a)=\mu_\alpha^a\theta^\alpha,\mbox{ for
all }a.$$ Moreover, it is clear that the restriction
$\tau_{\bar{\mathcal B}}:\bar{\mathcal B}\to M$ to $\bar{\mathcal
B}$ of the vector bundle projection $\tau_V^*:V^*\to M$ is a
fibration. Thus, we can consider the Lie algebroid
$\tau_{\widetilde{\mathcal A}}^{\tau_{\bar{\mathcal B}}}:{\mathcal
T}^{\widetilde{\mathcal A}}{\bar{\mathcal B}}\to\bar{\mathcal B}$
and the Hamilton equations of motion of the nonholonomic system
can then be rewritten in intrinsic form as
\begin{equation}\label{NHeq}
\left\{ \begin{array}{rcl} (i_{\bar{X}}\Omega_h)_{|\bar{{\mathcal B}}}&\in&\Gamma(\tau_{\bar{G}^\circ}),\\
(i_{\bar{X}}\eta)_{|\bar{{\mathcal B}}}&=&1,\\
\bar{X}_{|\bar{{\mathcal
B}}}&\in&\Gamma(\tau_{\widetilde{{\mathcal
A}}}^{\tau_{\bar{{\mathcal B}}}}),
\end{array} \right.
\end{equation}
$(\Omega_h,\eta)$ being the cosymplectic structure on ${\mathcal
T}^{\widetilde{{\mathcal A}}}V^*$ defined in (\ref{Omegah}) and
(\ref{eta}) and $\tau_{\bar{G}^\circ}$ being the vector bundle
projection of ${\bar{G}^\circ}\to\bar{\mathcal B}$. It is said
that the constrained Hamiltonian system $(h,\bar{{\mathcal B}})$
is {\it regular} if the Hamilton equations have a unique solution,
which we will denote by $\bar{R}_{nh}$.

Now, denote by $\flat_h:{\mathcal T}^{\widetilde{\mathcal
A}}V^*\to({\mathcal T}^{\widetilde{\mathcal A}}V^*)^*$ the vector
bundle isomorphism given by
$$\flat_h(X)=i_X\Omega_h+\eta(X)\eta,\mbox{ for }X\in{\mathcal T}^{\widetilde{\mathcal
A}}V^*,$$ and by $\Lambda_h$ the (algebraic) Poisson structure on
${\mathcal T}^{\widetilde{\mathcal A}}V^*$ defined by
$$\Lambda_h(\alpha,\beta)=\Omega_h(\flat_h^{-1}(\alpha),\flat_h^{-1}(\beta)),\mbox{
for }\alpha,\beta\in({\mathcal T}^{\widetilde{\mathcal
A}}V^*)^*.$$

 We also can transport the vector subbundle
$F\subset{\mathcal T}^{\widetilde{{\mathcal A}}}{\mathcal
A}_{|{\mathcal B}}\to {\mathcal B}$ and obtain a subbundle
$\bar{F}\subset{\mathcal T}^{\widetilde{{\mathcal
A}}}V^*_{|\bar{{\mathcal B}}}\to\bar{{\mathcal B}}$ such that
$\bar{F}_{\bar{\b}}=({\mathcal T}leg_L)(F_\b)$, with
$\bar{\b}=leg_L(\b)$ and $\b\in {\mathcal B}$. Moreover, it
satisfies that
$$\bar{F}_{\bar{\b}}=\flat_h^{-1}(\bar{G}_{\bar{\b}}^\circ)=\bar{G}_{\bar{\b}}^\perp=\bar{G}_{\bar{\b}}^{\perp,\Lambda_h},$$
where $\bar{G}^\perp$ (respectively, $\bar{G}^{\perp,\Lambda_h}$)
denotes the orthogonal to $\bar{G}$ with respect to the
cosymplectic (respectively, Poisson) structure $(\Omega_h,\eta)$
(respectively, $\Lambda_h$). Notice that $\bar{F}$ is locally
generated by the sections
$\bar{Z}_a=X_{\bar{\mu}^a}^{\Lambda_h}\in\Gamma(\tau_{\widetilde{{\mathcal
A}}}^{\tau_{\bar{\mathcal B}}})$ characterized by the conditions
$$i_{\bar{Z}_a}\Omega_h=\bar{\mu}^a\makebox[1cm]{and}i_{\bar{Z}_a}\eta=0,$$
for all $a\in\{1,\dots,r\}$. Of course, $Z_a$ and $\bar{Z}_a$ are
$({\mathcal T}leg_L,leg_L)$-related. Moreover, its local
expression is the following
\begin{equation}\label{Zbarraa}
\bar{Z}_a=-{\mathcal
H}^{\alpha\beta}\displaystyle\frac{\partial\psi^a}{\partial
y_\beta}\bar{e}_\alpha.
\end{equation}

We will denote by $(\bar{\mathcal C}^{ab})$ the matrix which
elements are
\begin{equation}\label{Cbarraab}
\bar{\mathcal C}^{ab}=(d^{{\mathcal T}^{\widetilde{{\mathcal
A}}}V^*}\psi^a)(\bar{Z}_b)=-\bar{\mu}^b(X_{\psi^a}^{\Lambda_h})=-{\mathcal
H}^{\alpha\beta}\displaystyle\frac{\partial\psi^a}{\partial
y_\beta}\frac{\partial\psi^b}{\partial y_\alpha},
\end{equation}
for all $a,b\in\{1,\dots,r\}$, where $X_{\psi^a}^{\Lambda_h}$ is
the Hamiltonian section of $\psi^a$ with respect to the
(algebraic) Poisson structure $\Lambda_h$, that is,
$X_{\psi^a}^{\Lambda_h}=-\sharp_{\Lambda_h}(d^{{\mathcal
T}^{\widetilde{\mathcal A}}V^*}\psi^a)$,
$\sharp_{\Lambda_h}:({\mathcal T}^{\widetilde{\mathcal
A}}V^*)^*\to{\mathcal T}^{\widetilde{\mathcal A}}V^*$ being the
vector bundle morphism given by
$\sharp_{\Lambda_h}(\alpha)=i_\alpha\Lambda_h$, for
$\alpha\in({\mathcal T}^{\widetilde{\mathcal A}}V^*)^*$. It is
easy to check that $\bar{\mathcal C}^{ab}\circ leg_L={\mathcal
C}^{ab}$.


In a similar way that in the Lagrangian side, we can consider the
subbundle $\bar{H}\subset{\mathcal T}^{\widetilde{{\mathcal
A}}}V^*_{|\bar{{\mathcal B}}}\to\bar{{\mathcal B}}$ such that
$\bar{H}_{\bar{\b}}={\mathcal T}^{\widetilde{{\mathcal
A}}}_{\bar{\b}}\bar{{\mathcal B}}\cap\bar{G}_{\bar{\b}}=({\mathcal
T}leg_L)(H_\b)$, with $\bar{\b}=leg_L(\b)$ and $\b\in {\mathcal
B}$, and its orthogonal $\bar{H}^\perp$ (respectively,
$\bar{H}^{\perp,\Lambda_h}$) with respect to the cosymplectic
(respectively, Poisson) structure $(\Omega_h,\eta)$ (respectively,
$\Lambda_h$).

Using Theorem \ref{Threg} and the fact that the Lagrangian $L$ is
hyperregular, we deduce the next result.
\begin{theorem}\label{NHth} The following properties are equivalent:
\begin{enumerate}
\item[(1)] The constrained Lagrangian system $(L,{\mathcal B})$ is
regular, \item[(2)] the constrained Hamiltonian system
$(h,\bar{{\mathcal B}})$ is regular, \item[(3)] the matrices
$\bar{\mathcal C}=(\bar{\mathcal C}^{ab})$ are non-singular,
\item[(4)] ${\mathcal T}^{\widetilde{{\mathcal
A}}}_{\bar{\b}}\bar{{\mathcal B}}\cap \bar{F}_{\bar{\b}}=\{0\},$
for all $\bar{\b}\in\bar{{\mathcal B}},$ \item[(5)]
$\bar{H}_{\bar{\b}}\cap\bar{H}_{\bar{\b}}^{\perp}=\{0\}$, for all
$\bar{\b}\in\bar{{\mathcal B}}$, \item[(6)]
$\bar{H}_{\bar{\b}}\cap
\bar{H}_{\bar{\b}}^{\perp,\Lambda_h}=\{0\}$, for all $\bar{\b}\in
\bar{{\mathcal B}}$.
\end{enumerate}
\end{theorem}

Assuming the regularity of the constrained system, we have the
corresponding direct sum decompositions of ${{\mathcal
T}^{\widetilde{{\mathcal A}}}V^*}_{|\bar{{\mathcal B}}}$ (similar
to the Lagrangian case). But we will only develop the Poisson
decomposition. Since the condition $(6)$ of the above theorem, we
have the decomposition
$${\mathcal
T}^{\widetilde{{\mathcal
A}}}_{\bar{\b}}V^*=\bar{H}_{\bar{\b}}\oplus\bar{H}_{\bar{\b}}^{\perp,\Lambda_h},\makebox[1.5cm]{for
all}\bar{\b}\in\bar{{\mathcal B}},$$ and, hence, one can define
two complementary projectors
$$\bar{\tilde{\mathcal P}}_{\bar{\b}}:{\mathcal
T}^{\widetilde{{\mathcal
A}}}_{\bar{\b}}V^*\to\bar{H}_{\bar{\b}}\makebox[1cm]{and}\bar{\tilde{\mathcal
Q}}_{\bar{\b}}:{\mathcal T}^{\widetilde{{\mathcal
A}}}_{\bar{\b}}V^*\to\bar{H}_{\bar{\b}}^{\perp,\Lambda_h},\makebox[1.5cm]{for
all}\bar{\b}\in\bar{{\mathcal B}}.$$ Moreover,
$(\bar{\tilde{\mathcal P}},\bar{\tilde{\mathcal Q}})$ is
$({\mathcal T}leg_L,leg_L)$-related with $(\tilde{\mathcal
P},\tilde{\mathcal Q})$. Then, the section $\bar{\tilde{\mathcal
P}}({R_h}_{|\bar{{\mathcal B}}})$ is the unique solution of the
constrained Hamilton equations (\ref{NHeq}), where $R_h$ is the
solution of the Hamilton equations for the free dynamics, that is,
$\bar{R}_{nh}=\bar{\tilde{\mathcal P}}({R_h}_{|\bar{{\mathcal
B}}})$.

Now, we are going to construct a decomposition of ${\mathcal
T}^{\widetilde{{\mathcal A}}}{\mathcal A}^+$, which allows us to
obtain the solution of the constrained Hamiltonian system as the
projection by ${\mathcal T} \mu$ of a certain section. Given a
Hamiltonian section $h: V^*\to {\mathcal A}^+$, one can construct
an affine function $F_h:{\mathcal A}^+\to\R$ with respect to the
AV-bundle $\mu: {\mathcal A}^+\to V^*$ as follows. For each
$\varphi_x\in {\mathcal A}^+_x$, with $x\in M$, exists a unique
$F_h(\varphi_x)\in\R$ such that
\begin{equation}\label{conFh}
h(\mu(\varphi_x))-\varphi_x=F_h(\varphi_x)1_{\mathcal A}(x).
\end{equation}
The function $F_h:{\mathcal A}^+\to\R$ is locally given by
\begin{equation}\label{Fhloc}
F_h(x^i,y_0,y_\alpha)=-H(x^i,y_\alpha)-y_0. \end{equation}
 Note
that this function was introduced in another way in Section
\ref{sec3.1}. Using (\ref{lambdaE}), (\ref{Omegah}), (\ref{eta})
and (\ref{conFh}), it is easy to prove that
$$\lambda_{\widetilde{{\mathcal A}}}=({\mathcal
T}\mu,\mu)^*\lambda_h-F_h({\mathcal T}\mu,\mu)^*\eta,$$
where $\lambda_{\widetilde{{\mathcal A}}}$ is the Liouville
section associated with the Lie algebroid $\widetilde{{\mathcal
A}}$ and $({\mathcal T}\mu,\mu)$ is the Lie algebroid epimorphism
between the Lie algebroids $\tau_{\widetilde{{\mathcal
A}}}^{\tau_{{\mathcal A}^+}}:{\mathcal T}^{\widetilde{{\mathcal
A}}}{\mathcal A}^+\to {\mathcal A}^+$ and
$\tau_{\widetilde{{\mathcal A}}}^{\tau_V^*}:{\mathcal
T}^{\widetilde{{\mathcal A}}}V^*\to V^*$ defined by ${\mathcal
T}\mu(\tilde{a},X_{\varphi_x})=(\tilde{a},(T_{\varphi_x}\mu)(X_{\varphi_x}))$,
for $(\tilde{a},X_{\varphi_x})\in{\mathcal
T}_{\varphi_x}^{\widetilde{{\mathcal A}}}{\mathcal A}^+$. As
consequence, we have that
\begin{equation}\label{OmegaAh}
\Omega_{\widetilde{{\mathcal A}}}=({\mathcal
T}\mu,\mu)^*\Omega_h+\theta_h\wedge({\mathcal T}\mu,\mu)^*\eta,
\end{equation}
where $\theta_h=d^{{\mathcal T}^{\widetilde{{\mathcal
A}}}{\mathcal A}^+}F_h$ and $\Omega_{\widetilde{{\mathcal A}}}$ is
the canonical symplectic section associated with
$\widetilde{{\mathcal A}}$.

We will denote by $E\in\Gamma(\tau_{\widetilde{{\mathcal
A}}}^{\tau_{{\mathcal A}^+}})$ the section given by
\begin{equation}\label{defE}
E(\varphi_x)=(0(x),-1_{\mathcal A}^V(\varphi_x))\in{\mathcal
T}^{\widetilde{{\mathcal A}}}_{\varphi_x}{\mathcal
A}^+,\makebox[1.5cm]{for all}\varphi_x\in {\mathcal A}^+_x,
\end{equation} where $1_{\mathcal A}^V \in \frak X({\mathcal A}^+)$ is
the vertical lift of the section $1_{\mathcal
A}\in\Gamma(\tau_{{\mathcal A}^+})$. Using (\ref{Fhloc}), we
obtain that
\begin{equation}\label{eq5.9'}
ker\, {\mathcal T}_{\varphi_x}\mu=<E(\varphi_x)>,\;\;
\theta_h(\varphi_x)(E(\varphi_x))=1,
\end{equation}
for all $\varphi_x\in {\mathcal A}^+_x$, with $x\in M$.

Now, we introduce the vector subbundle $({\mathcal
T}^{\widetilde{{\mathcal A}}}V^*)^H\subset {\mathcal
T}^{\widetilde{{\mathcal A}}}{\mathcal A}^+\to {\mathcal A}^+$
whose fibre at point $\varphi_x\in {\mathcal A}^+_x$ is
$$({\mathcal
T}^{\widetilde{{\mathcal
A}}}_{\mu(\varphi_x)}V^*)^H_{\varphi_x}=\{X_{\varphi_x}\in{\mathcal
T}^{\widetilde{{\mathcal A}}}_{\varphi_x}{\mathcal
A}^+/\theta_h(\varphi_x)(X_{\varphi_x})=0\}.$$

It is not difficult to prove that $({\mathcal
T}_{\varphi_x}\mu)_{|({\mathcal T}^{\widetilde{{\mathcal
A}}}_{\mu(\varphi_x)}V^*)^H_{\varphi_x}}:({\mathcal
T}^{\widetilde{{\mathcal
A}}}_{\mu(\varphi_x)}V^*)^H_{\varphi_x}\to {\mathcal
T}^{\widetilde{{\mathcal A}}}_{\mu(\varphi_x)}V^*$ is a linear
isomorphism, for all $\varphi_x\in {\mathcal A}^+_x$. Then, we
deduce that
\begin{equation}\label{decA}
{\mathcal T}^{\widetilde{{\mathcal A}}}_{\varphi_x}{\mathcal
A}^+=({\mathcal T}^{\widetilde{{\mathcal
A}}}_{\mu(\varphi_x)}V^*)^H_{\varphi_x}\oplus <E(\varphi_x)>,
\end{equation} for all $\varphi_x\in {\mathcal A}^+_x$, $x\in M$. We will denote by ${}^H_{\varphi_x}:{\mathcal
T}^{\widetilde{{\mathcal A}}}_{\mu(\varphi_x)}V^*\to ({\mathcal
T}^{\widetilde{{\mathcal A}}}_{\mu(\varphi_x)}V^*)^H_{\varphi_x}$
the inverse map of $({\mathcal T}_{\varphi_x}\mu)_{|({\mathcal
T}^{\widetilde{{\mathcal
A}}}_{\mu(\varphi_x)}V^*)^H_{\varphi_x}}:({\mathcal
T}^{\widetilde{{\mathcal
A}}}_{\mu(\varphi_x)}V^*)^H_{\varphi_x}\to {\mathcal
T}^{\widetilde{{\mathcal A}}}_{\mu(\varphi_x)}V^*$. So, if
$X\in{\mathcal T}^{\widetilde{{\mathcal A}}}_{\varphi_x}{\mathcal
A}^+$ it follows that
\begin{equation}\label{eq5.10'}
X=({\mathcal
T}_{\varphi_x}\mu)(X_{\varphi_x})_{\varphi_x}^H+\theta_h(\varphi_x)(X_{\varphi_x})E(\varphi_x).
\end{equation}

On the other hand, consider the section
$\tilde{\eta}\in\Gamma((\tau_{\widetilde{\mathcal
A}}^{\tau_{{\mathcal A}^+}})^*)$ defined by
$$\tilde{\eta}(\varphi_x)(\tilde{a},X_{\varphi_x})=1_{\mathcal
A}(\tilde{a}),\mbox{ for }\varphi_x\in{\mathcal A}_x^+\mbox{ and
}(\tilde{a},X_{\varphi_x})\in{\mathcal
T}_{\varphi_x}^{\widetilde{\mathcal A}}{\mathcal A}^+.$$

A direct computation, using (\ref{formas}) and (\ref{Fhloc}),
proves that
$$\tilde{\eta}(X_{F_h}^{\Omega_{\widetilde{\mathcal A}}})=-1,$$
where $X_{F_h}^{\Omega_{\widetilde{\mathcal
A}}}\in\Gamma(\tau_{\widetilde{\mathcal A}}^{\tau_{{\mathcal
A}^+}})$ is the Hamiltonian section of $F_h$ with respect to the
symplectic structure $\Omega_{\widetilde{\mathcal A}}$. Thus,
since $({\mathcal T}\mu,\mu)^*\eta=\tilde{\eta}$, we deduce that
$$\eta(\mu(\varphi_x))(({\mathcal T}_{\varphi_x}\mu)(X_{F_h}^{\Omega_{\widetilde{\mathcal
A}}}(\varphi_x)))=-1.$$ Therefore, from (\ref{OmegaAh}), we obtain
that
$$\Big(i_{({\mathcal T}_{\varphi_x}\mu)(X_{F_h}^{\Omega_{\widetilde{\mathcal
A}}}(\varphi_x))}\Omega_h(\mu(\varphi_x))\Big)(({\mathcal
T}_{\varphi_x}\mu)(Z_{\varphi_x}))=0,\mbox{ for
}Z_{\varphi_x}\in{\mathcal T}_{\varphi_x}^{\widetilde{\mathcal
A}}{\mathcal A}^+.$$ This implies that
\begin{equation}\label{5.6'}
({\mathcal T}_{\varphi_x}\mu)(X_{F_h}^{\Omega_{\widetilde{\mathcal
A}}}(\varphi_x))=-R_h(\mu(\varphi_x)),\mbox{ for
}\varphi_x\in{\mathcal A}^+_x.
\end{equation}

Now, if $\alpha\in\Gamma((\tau_{\widetilde{\mathcal
A}}^{\tau_V^*})^*)$, then the Hamiltonian section
$X_\alpha^{\Lambda_h}\in\Gamma(\tau_{\widetilde{\mathcal
A}}^{\tau_V^*})$ of $\alpha$ with respect to the (algebraic)
Poisson structure $\Lambda_h$ is characterized by the following
conditions
$$i_{X_\alpha^{\Lambda_h}}\Omega_h=\alpha-\alpha(R_h)\eta\;\mbox{
and }\;i_{X_\alpha^{\Lambda_h}}\eta=0.$$ On the other hand, denote
by $\Lambda_{\widetilde{\mathcal
A}}\in\Gamma(\wedge^2(\tau_{\widetilde{\mathcal
A}}^{\tau_{{\mathcal A}^+}}))$ the (algebraic) Poisson structure
associated with the symplectic structure
$\Omega_{\widetilde{\mathcal A}}$. In other words,
$$\Lambda_{\widetilde{\mathcal
A}}(\tilde{\alpha},\tilde{\beta})=\Omega_{\widetilde{\mathcal
A}}(\tilde{\flat}^{-1}(\tilde{\alpha}),\tilde{\flat}^{-1}(\tilde{\beta})),\mbox{
for
}\tilde{\alpha},\tilde{\beta}\in\Gamma((\tau_{\widetilde{\mathcal
A}}^{\tau_{{\mathcal A}^+}})^*),$$ $\tilde{\flat}:{\mathcal
T}^{\widetilde{\mathcal A}}{\mathcal A}^+\to({\mathcal
T}^{\widetilde{\mathcal A}}{\mathcal A}^+)^*$ being the vector
bundle isomorphism induced by $\Omega_{\widetilde{\mathcal A}}$.
Note that if $\tilde{\alpha}\in\Gamma((\tau_{\widetilde{\mathcal
A}}^{\tau_{{\mathcal A}^+}})^*)$ then
$$X_{\tilde{\alpha}}^{\Omega_{\widetilde{\mathcal
A}}}=-\sharp_{\Lambda_{\widetilde{\mathcal A}}}(\tilde{\alpha}).$$

\begin{proposition}\label{prop5.1'} The prolongation of $\mu$
$${\mathcal T}\mu:{\mathcal T}^{\widetilde{\mathcal A}}{\mathcal
A}^+\to{\mathcal T}^{\widetilde{\mathcal A}}V^*$$ is a Poisson
morphism over $\mu$, that is, if
$\alpha\in\Gamma((\tau_{\widetilde{\mathcal A}}^{\tau_V^*})^*)$
then
\begin{equation}\label{Poismorph}
{\mathcal T}\mu\circ X_{({\mathcal
T}\mu,\mu)^*\alpha}^{\Omega_{\widetilde{\mathcal
A}}}=X_\alpha^{\Lambda_h}\circ\mu.
\end{equation}
\end{proposition}
\begin{proof} If $\varphi_x\in{\mathcal A}_x^+$ then a direct
computation, using (\ref{formas}) and (\ref{eta}), proves that
$$\eta(\mu(\varphi_x))[({\mathcal T}_{\varphi_x}\mu)(X_{({\mathcal
T}\mu,\mu)^*\alpha}^{\Omega_{\widetilde{\mathcal
A}}}(\varphi_x))]=0.$$ Thus, from (\ref{OmegaAh}) and
(\ref{5.6'}), it follows that

$\Big(i_{({\mathcal T}_{\varphi_x}\mu)(X_{({\mathcal
T}\mu,\mu)^*\alpha}^{\Omega_{\widetilde{\mathcal
A}}}(\varphi_x))}\Omega_h(\mu(\varphi_x))\Big)(({\mathcal
T}_{\varphi_x}\mu)(Z_{\varphi_x}))=\alpha(\mu(\varphi_x))(({\mathcal
T}_{\varphi_x}\mu)(Z_{\varphi_x}))$
\begin{flushright}
$+\Omega_{\widetilde{\mathcal A}}(\varphi_x)(X_{({\mathcal
T}\mu,\mu)^*\alpha}^{\Omega_{\widetilde{\mathcal
A}}}(\varphi_x),X_{F_h}^{\Omega_{\widetilde{\mathcal
A}}}(\varphi_x))\eta(\mu(\varphi_x))(({\mathcal
T}_{\varphi_x}\mu)(Z_{\varphi_x}))$
\end{flushright}\begin{flushright}
\hspace{2cm}$=\alpha(\mu(\varphi_x))(({\mathcal
T}_{\varphi_x}\mu)(Z_{\varphi_x}))-\alpha(\mu(\varphi_x))(R_h(\mu(\varphi_x)))
\eta(\mu(\varphi_x))(({\mathcal
T}_{\varphi_x}\mu)(Z_{\varphi_x})),$
\end{flushright}

for $Z_{\varphi_x}\in{\mathcal T}_{\varphi_x}^{\widetilde{\mathcal
A}}{\mathcal A}^+$. This implies that
$$i_{({\mathcal T}_{\varphi_x}\mu)(X_{({\mathcal
T}\mu,\mu)^*\alpha}^{\Omega_{\widetilde{\mathcal
A}}}(\varphi_x))}\Omega_h(\mu(\varphi_x))=\alpha(\mu(\varphi_x))-\alpha(\mu(\varphi_x))(R_h(\mu(\varphi_x)))
\eta(\mu(\varphi_x)).$$ Therefore, we deduce that
(\ref{Poismorph}) holds.
\end{proof}

\begin{remark} {\rm From Proposition \ref{prop5.1'}, it follows that
\begin{equation}\label{eq5.12'}
{\mathcal
T}_{\varphi_x}\mu\circ\sharp_{\Lambda_{\widetilde{\mathcal
A}}|({\mathcal T}^{\widetilde{\mathcal A}}_{\varphi_x}{\mathcal
A}^+)^*}\circ({\mathcal
T}_{\varphi_x}\mu)^*=\sharp_{\Lambda_h|({\mathcal
T}^{\widetilde{\mathcal A}}_{\mu(\varphi_x)}V^*)^*},\makebox{ for
}\varphi_x\in{\mathcal A}^+_x.
\end{equation}}
\end{remark}

Next, we consider the submanifold ${\mathcal B}^+$ of ${\mathcal
A}^+$ defined by ${\mathcal B}^+=\mu^{-1}(\bar{{\mathcal B}})$ and
the vector subbundle $H^+\subset{\mathcal T}^{\widetilde{{\mathcal
A}}}{\mathcal A}^+_{|{\mathcal B}^+}\to {\mathcal B}^+$ given by
\begin{equation}\label{defH+}
H^+_{\varphi_x}=({\mathcal
T}_{\varphi_x}\mu)^{-1}(\bar{H}_{\mu(\varphi_x)})=(\bar{H}_{\mu(\varphi_x)})^H_{\varphi_x}\oplus<E(\varphi_x)>,\makebox[1.5cm]{for
all}\varphi_x\in {\mathcal B}^+_x,
\end{equation}

where $(\bar{H}_{\mu(\varphi_x)})_{\varphi_x}^H=({\mathcal
T}_{\varphi_x}\mu)^{-1}(\bar{H}_{\mu(\varphi_x)})\cap({\mathcal
T}_{\mu(\varphi_x)}^{\widetilde{\mathcal A}}V^*)_{\varphi_x}^H$.

We will denote by
$(H^+_{\varphi_x})^{\perp,\Omega_{\widetilde{{\mathcal A}}}}$ the
orthogonal to $H^+_{\varphi_x}$ with respect to the symplectic
structure $\Omega_{\widetilde{{\mathcal A}}}(\varphi_x)$ of
${\mathcal T}^{\widetilde{{\mathcal A}}}_{\varphi_x}{\mathcal
A}^+$, for all $\varphi_x\in {\mathcal B}^+_x$. In other words,
\begin{equation}\label{eq5.13'}
(H^+_{\varphi_x})^{\perp,\Omega_{\widetilde{\mathcal
A}}}=\tilde{\flat}^{-1}((H^+_{\varphi_x})^\circ)=\sharp_{\Lambda_{\widetilde{\mathcal
A}}}((H^+_{\varphi_x})^\circ). \end{equation} Using (\ref{defH+}),
we have that
$$(H^+_{\varphi_x})^\circ=({\mathcal
T}_{\varphi_x}\mu)^*((\bar{H}_{\mu(\varphi_x)})^\circ)$$ and, from
(\ref{eq5.12'}) and (\ref{eq5.13'}), it follows that
\begin{equation}\label{TH+}({\mathcal
T}_{\varphi_x}\mu)((H^+_{\varphi_x})^{\perp,\Omega_{\widetilde{{\mathcal
A}}}})=(\bar{H}_{\mu(\varphi_x)})^{\perp,\Lambda_h},\makebox[1.5cm]{for
all}\varphi_x\in {\mathcal B}^+_x.
\end{equation}

In addition, using that $\{\bar{\mu}^a, d^{{\mathcal
T}^{\widetilde{{\mathcal A}}}V^*}\psi^a\}$ is a local basis of
$\bar{H}^\circ$, we obtain that $\{({\mathcal T}\mu,\mu)^*
\bar{\mu}^a,$ $({\mathcal T}\mu,\mu)^*(d^{{\mathcal
T}^{\widetilde{{\mathcal A}}}V^*}\psi^a)\}$ is a local basis of
$(H^+)^\circ$. This implies that $\{X_{({\mathcal T}\mu,\mu)^*
\bar{\mu}^a}^{\Omega_{\widetilde{{\mathcal A}}}},X_{({\mathcal
T}\mu,\mu)^*(d^{{\mathcal T}^{\widetilde{{\mathcal
A}}}V^*}\psi^a)}^{\Omega_{\widetilde{{\mathcal A}}}}\}$ is a local
basis of $(H^+)^{\perp,\Omega_{\widetilde{{\mathcal A}}}}$.

Note that, from Proposition \ref{prop5.1'}, we have that
\begin{equation}\label{5.9'}
\begin{array}{rcl}
({\mathcal T}_{\varphi_x}\mu)(X_{({\mathcal T}\mu,\mu)^*
\bar{\mu}^a}^{\Omega_{\widetilde{{\mathcal
A}}}}(\varphi_x))&=&X_{\bar{\mu}^a}^{\Lambda_h}(\mu(\varphi_x))=\bar{Z}_a(\mu(\varphi_x)),\\[8pt]
({\mathcal T}_{\varphi_x}\mu)(X_{({\mathcal
T}\mu,\mu)^*(d^{{\mathcal T}^{\widetilde{{\mathcal
A}}}V^*}\psi^a)}^{\Omega_{\widetilde{{\mathcal
A}}}}(\varphi_x))&=&X_{\psi^a}^{\Lambda_h}(\mu(\varphi_x)).
\end{array}\end{equation}

Then, we will prove the following result.

\begin{theorem} The following conditions are equivalent:
\begin{enumerate}
\item[(1)] The constrained Hamiltonian system $(h,\bar{{\mathcal
B}})$ is regular, \item[(2)]
$H^+_{\varphi_x}\cap(H^+_{\varphi_x})^{\perp,\Omega_{\widetilde{{\mathcal
A}}}}=\{0\}$, for all $\varphi_x\in {\mathcal B}^+_x$, with $x\in
M$.
\end{enumerate}
\end{theorem}

\begin{proof} $(1)\Rightarrow(2)$ Suppose that the constrained
Hamiltonian system is regular. Then, since ${\mathcal
B}^+=\mu^{-1}(\bar{{\mathcal B}})$ and using Theorem \ref{NHth},
we have that
$\bar{H}_{\mu(\varphi_x)}\cap\bar{H}_{\mu(\varphi_x)}^{\perp,\Lambda_h}=\{0\}$,
for all $\varphi_x\in {\mathcal B}^+_x$. If $X\in
H^+_{\varphi_x}\cap(H^+_{\varphi_x})^{\perp,\Omega_{\widetilde{{\mathcal
A}}}}$ then, from (\ref{defH+}), $\bar{X}=({\mathcal
T}_{\varphi_x}\mu)(X)\in\bar{H}_{\mu(\varphi_x)}$. Moreover, using
(\ref{TH+}), we deduce that
$\bar{X}\in(\bar{H}_{\mu(\varphi_x)})^{\perp,\Lambda_h}$ and,
therefore, $\bar{X}=0$. Thus, $X\in$ $<E(\varphi_x)>$. Therefore,
from (\ref{OmegaAh}) and (\ref{eq5.9'}),
$i_X\Omega_{\widetilde{\mathcal A}}(\varphi_x)=\lambda(({\mathcal
T}\mu,\mu)^*\eta)(\varphi_x)\in(H^+_{\varphi_x})^\circ$, with
$\lambda\in\R$. However, $(({\mathcal
T}\mu,\mu)^*\eta)(\varphi_x)\not\in (H^+_{\varphi_x})^\circ$. In
fact, there exists $\tilde{\bar{R}}_{nh}(\varphi_x)\in
H^+_{\varphi_x}$ such that $({\mathcal
T}_{\varphi_x}\mu)(\tilde{\bar{R}}_{nh}(\varphi_x))=\bar{R}_{nh}(\mu(\varphi_x))$
and, thus, $(({\mathcal
T}\mu,\mu)^*\eta)(\varphi_x)(\tilde{\bar{R}}_{nh}(\varphi_x))=\eta(\mu(\varphi_x))(\bar{R}_{nh}(\mu(\varphi_x)))=1$.
Consequently, $\lambda=0$ and $X=0$.

$(2)\Rightarrow(1)$ Suppose that
$H^+_{\varphi_x}\cap(H^+_{\varphi_x})^{\perp,\Omega_{\widetilde{{\mathcal
A}}}}=\{0\}$, for all $\varphi_x\in {\mathcal B}^+_x$ and take
$\bar{X}\in\bar{H}_{\mu(\varphi_x)}\cap\bar{H}_{\mu(\varphi_x)}^{\perp,\Lambda_h}$.
Using (\ref{TH+}), there exists $X\in
(H^+_{\varphi_x})^{\perp,\Omega_{\widetilde{{\mathcal A}}}}$ such
that $({\mathcal T}_{\varphi_x}\mu)(X)=\bar{X}$ and, since
$\bar{X}\in\bar{H}_{\mu(\varphi_x)}$, we have that $X\in
H^+_{\varphi_x}$. Thus, $X=0$, which implies that $\bar{X}=0$.
Finally, using Theorem \ref{NHth}, we conclude the result.

\end{proof}

Assuming that the constrained system is regular, the bundle
$H^+\to {\mathcal B}^+$ is a symplectic subbundle of the
symplectic bundle $({{\mathcal T}^{\widetilde{{\mathcal
A}}}{\mathcal A}^+}_{|{\mathcal B}^+}\to {\mathcal
B}^+,(\Omega_{\widetilde{{\mathcal A}}})_{|{\mathcal B}^+})$ and
we have a direct sum decomposition
$${\mathcal
T}^{\widetilde{{\mathcal A}}}_{\varphi_x}{\mathcal
A}^+=H^+_{\varphi_x}\oplus(H^+_{\varphi_x})^{\perp,\Omega_{\widetilde{{\mathcal
A}}}},\makebox[1.5cm]{for all}\varphi_x\in {\mathcal
B}^+_x,\makebox[1cm]{with}x\in M.$$

Let us denote by ${\mathcal P}^+$ and ${\mathcal Q}^+$ the
complementary projectors defined by this decomposition, that is,
$${\mathcal P}^+_{\varphi_x}:{\mathcal
T}^{\widetilde{{\mathcal A}}}_{\varphi_x}{\mathcal A}^+\to
H^+_{\varphi_x}\makebox[1cm]{and}{\mathcal
Q}^+_{\varphi_x}:{\mathcal T}^{\widetilde{{\mathcal
A}}}_{\varphi_x}{\mathcal A}^+\to
(H^+_{\varphi_x})^{\perp,\Omega_{\widetilde{{\mathcal
A}}}},\makebox[1.5cm]{for all}\varphi_x\in {\mathcal B}^+_x.$$

\begin{theorem}\label{th5.3} Let $(h,\bar{{\mathcal B}})$ be a regular constrained
Hamiltonian system. Then, the solution of the constrained dynamics
is the section $\bar{R}_{nh}$ obtained as follows
$\bar{R}_{nh}\circ\mu=-{\mathcal T}\mu({\mathcal
P}^+(X_{F_h}^{\Omega_{\widetilde{{\mathcal A}}}}))$.
\end{theorem}

\begin{proof} If $\varphi_x\in{\mathcal B}_x^+$ then, from
(\ref{defH+}) and since the map ${\mathcal
T}_{\varphi_x}\mu:{\mathcal T}^{\widetilde{{\mathcal
A}}}_{\varphi_x}{\mathcal A}^+\to {\mathcal
T}^{\widetilde{{\mathcal A}}}_{\mu(\varphi_x)}V^*$ is a linear
epimorphism, it follows that
\begin{equation}\label{5.9''}
({\mathcal
T}_{\varphi_x}\mu)(H^+_{\varphi_x})=\bar{H}_{\mu(\varphi_x)}.
\end{equation}
Thus, if $Z_{\varphi_x}\in {\mathcal T}^{\widetilde{{\mathcal
A}}}_{\varphi_x}{\mathcal A}^+$ then, using (\ref{TH+}) and
(\ref{5.9''}), we deduce that
$$({\mathcal T}_{\varphi_x}\mu)({\mathcal
P}^+_{\varphi_x}(Z_{\varphi_x}))\in\bar{H}_{\mu(\varphi_x)},\;\;({\mathcal
T}_{\varphi_x}\mu)({\mathcal
Q}^+_{\varphi_x}(Z_{\varphi_x}))\in(\bar{H}_{\mu(\varphi_x)})^{\perp,\Lambda_h}$$
and it is clear that $$({\mathcal
T}_{\varphi_x}\mu)(Z_{\varphi_x})=({\mathcal
T}_{\varphi_x}\mu)({\mathcal
P}^+_{\varphi_x}(Z_{\varphi_x}))+({\mathcal
T}_{\varphi_x}\mu)({\mathcal Q}^+_{\varphi_x}(Z_{\varphi_x})).$$
Therefore, we have proved that
$$\begin{array}{rcl}
({\mathcal T}_{\varphi_x}\mu)({\mathcal
P}^+_{\varphi_x}(Z_{\varphi_x}))&=&\bar{\tilde{\mathcal
P}}_{\mu(\varphi_x)}(({\mathcal
T}_{\varphi_x}\mu)(Z_{\varphi_x})),\\[6pt]
({\mathcal T}_{\varphi_x}\mu)({\mathcal
Q}^+_{\varphi_x}(Z_{\varphi_x}))&=&\bar{\tilde{\mathcal
Q}}_{\mu(\varphi_x)}(({\mathcal
T}_{\varphi_x}\mu)(Z_{\varphi_x})),
\end{array}$$
i.e.,
\begin{equation}\label{Pconm}
{\mathcal T}_{\varphi_x}\mu\circ{\mathcal
P}^+_{\varphi_x}=\bar{\tilde{\mathcal
P}}_{\mu(\varphi_x)}\circ{\mathcal
T}_{\varphi_x}\mu\makebox[1cm]{and}{\mathcal
T}_{\varphi_x}\mu\circ{\mathcal
Q}^+_{\varphi_x}=\bar{\tilde{\mathcal
Q}}_{\mu(\varphi_x)}\circ{\mathcal T}_{\varphi_x}\mu.
\end{equation}


Finally, from (\ref{5.6'}) and (\ref{Pconm}), we conclude that
$$\bar{R}_{nh}(\mu(\varphi_x))=\bar{\tilde{\mathcal P}}_{\mu(\varphi_x)}(R_h(\mu(\varphi_x)))=-\bar{\tilde{\mathcal P}}_{\mu(\varphi_x)}({\mathcal
T}_{\varphi_x}\mu(X_{F_h}^{\Omega_{\widetilde{{\mathcal
A}}}}(\varphi_x)))=-{\mathcal T}_{\varphi_x}\mu({\mathcal
P}^+_{\varphi_x}(X_{F_h}^{\Omega_{\widetilde{{\mathcal
A}}}}(\varphi_x))).$$
\end{proof}

Next, we will denote by $\{\cdot,\cdot\}_h$ the Poisson bracket on
$V^*$ induced by the algebraic Poisson structure $\Lambda_h$ given
by
\begin{equation}\label{eq5.17'}
\{\psi,\psi'\}_h=-(d^{{\mathcal T}^{\widetilde{\mathcal
A}}V^*}\psi')(X_\psi^{\Lambda_h})=\Omega_h(X_\psi^{\Lambda_h},X_{\psi'}^{\Lambda_h}),
\end{equation}
for $\psi,\psi'\in C^\infty(V^*)$.

Then, using (\ref{Cbarraab}), (\ref{defH+}), (\ref{5.9'}) and the
fact that $\bar{\mu}^a(\bar{Z}_b)=0$ (see (\ref{mubarraa}) and
(\ref{Zbarraa})), we deduce that the local expression of the
projector ${\mathcal P}^+$ is
\begin{equation}\label{locP+}
\begin{array}{lcl} {\mathcal P}^+&=&Id-(\bar{\mathcal
C}_{ab}\circ\mu)(\bar{\mathcal
C}_{cd}\circ\mu)(\{\psi^b,\psi^d\}_h\circ\mu) X_{({\mathcal
T}\mu,\mu)^*\bar{\mu}^a}^{\Omega_{\widetilde{{\mathcal
A}}}}\otimes({\mathcal
T}\mu,\mu)^* \bar{\mu}^c\\[8pt]
&-&(\bar{\mathcal C}_{ab}\circ\mu) X_{({\mathcal
T}\mu,\mu)^*\bar{\mu}^a}^{\Omega_{\widetilde{{\mathcal
A}}}}\otimes({\mathcal T}\mu,\mu)^*(d^{{\mathcal
T}^{\widetilde{{\mathcal A}}}V^*}\psi^b)\\[8pt]
&+&(\bar{\mathcal C}_{ab}\circ\mu) X_{({\mathcal
T}\mu,\mu)^*(d^{{\mathcal T}^{\widetilde{{\mathcal
A}}}V^*}\psi^a)}^{\Omega_{\widetilde{{\mathcal
A}}}}\otimes({\mathcal T}\mu,\mu)^* \bar{\mu}^b,
\end{array}
\end{equation}
along the points of ${\mathcal B}^+$.

Now, we will denote by $\bar{\theta}_h$ the section of $({\mathcal
T}^{\widetilde{{\mathcal A}}}V^*)^*_{|\bar{{\mathcal B}}}\to
\bar{{\mathcal B}}$ given by
$$\bar{\theta}_h(\bar{\b})=-(i_{\bar{R}_{nh}}\Omega_h)(\bar{\b}),\makebox[1.5cm]{for
all}\bar{\b}\in\bar{{\mathcal B}}.$$
Then, using (\ref{OmegaAh}), (\ref{Pconm}), Theorem \ref{th5.3}
and the facts that $Im({\mathcal P}^+)\subseteq {H}^+$,
$Im({\mathcal Q}^+)\subseteq
(H^+)^{\perp,\Omega_{\widetilde{{\mathcal A}}}}$,
$\bar{R}_{nh}\in\bar{H}$ and $i_{\bar{R}_{nh}}\eta=1$ , it is not
difficult to prove that
\begin{equation}\label{thetaconm}
\bar{\theta}_h(\mu(\varphi_x))(\bar{\tilde{\mathcal
Q}}_{\mu(\varphi_x)}({\mathcal
T}_{\varphi_x}\mu(X_{\varphi_x})))=-\theta_h(\varphi_x)({\mathcal
Q}^+_{\varphi_x}(X_{\varphi_x})),
\end{equation}
for all $X_{\varphi_x}\in{\mathcal T}^{\widetilde{{\mathcal
A}}}_{\varphi_x}{\mathcal A}^+$ and $\varphi_x\in {\mathcal
B}^+_x$, with $x\in M$. Note that
$\eta(\mu(\varphi_x))\circ{\mathcal
T}_{\varphi_x}\mu\circ{\mathcal
Q}^+_{\varphi_x}=\eta(\mu(\varphi_x))\circ \bar{\tilde{\mathcal
Q}}_{\mu(\varphi_x)}\circ {\mathcal T}_{\varphi_x}\mu=0$. We also
recall that $\theta_h=d^{{\mathcal T}^{\widetilde{{\mathcal
A}}}{\mathcal A}^+}F_h$.

Thus, using (\ref{eq5.10'}), (\ref{Pconm}) and (\ref{thetaconm}),
we obtain that
\begin{equation}\label{PQdec}
\begin{array}{lcl}
{\mathcal P}^+_{\varphi_x}(X_{\varphi_x})&=&(\bar{\tilde{\mathcal
P}}_{\mu(\varphi_x)}({\mathcal
T}_{\varphi_x}\mu(X_{\varphi_x})))^H_{\varphi_x}+\Big(\theta_h(\varphi_x)(X_{\varphi_x})\\[8pt]
&&+\bar{\theta}_h(\mu(\varphi_x))(\bar{\tilde{\mathcal
Q}}_{\mu(\varphi_x)}({\mathcal
T}_{\varphi_x}\mu(X_{\varphi_x})))\Big)E(\varphi_x),\\[8pt]
{\mathcal Q}^+_{\varphi_x}(X_{\varphi_x})&=&(\bar{\tilde{\mathcal
Q}}_{\mu(\varphi_x)}({\mathcal
T}_{\varphi_x}\mu(X_{\varphi_x})))^H_{\varphi_x}-\bar{\theta}_h(\mu(\varphi_x))(\bar{\tilde{\mathcal
Q}}_{\mu(\varphi_x)}({\mathcal
T}_{\varphi_x}\mu(X_{\varphi_x})))E(\varphi_x),
\end{array}
\end{equation}
for all $X_{\varphi_x}\in{\mathcal
T}_{\varphi_x}^{\widetilde{{\mathcal A}}}{\mathcal A}^+$, with
$\varphi_x\in {\mathcal B}^+_x$ and $x\in M$.

\subsection{The nonholonomic bracket}\label{secnhbrac}
We consider a regular nonholonomic system on a Lie aff\-ge\-broid
${\mathcal A}$ described by a hyperregular Lagrangian function
$L:{\mathcal A}\to \R$ and a constraint affine subbundle
$\tau_{\mathcal B}:{\mathcal B}\to M$ of the bundle
$\tau_{\mathcal A}:{\mathcal A}\to M$. We will denote by
$(h,\bar{\mathcal B})$ the corresponding regular constrained
Hamiltonian system, by $\bar{h}$ the restriction to $\bar{\mathcal
B}$ of $h$, by ${\mathcal B}^+$ the submanifold of ${\mathcal
A}^+$ given by ${\mathcal B}^+=\mu^{-1}(\bar{\mathcal B})$ and by
$${\mathcal P}^+_{\varphi_x}:{\mathcal
T}^{\widetilde{{\mathcal A}}}_{\varphi_x}{\mathcal A}^+\to
H^+_{\varphi_x}\makebox[1cm]{and}{\mathcal
Q}^+_{\varphi_x}:{\mathcal T}^{\widetilde{{\mathcal
A}}}_{\varphi_x}{\mathcal A}^+\to
(H^+_{\varphi_x})^{\perp,\Omega_{\widetilde{{\mathcal A}}}},$$ the
corresponding complementary projectors, for $\varphi_x\in
{\mathcal B}^+_x$.

We have that ${\mathcal B}^+$ is the total space of an AV-bundle
over $\bar{\mathcal B}$. In fact, the affine bundle projection is
the restriction $\mu_{{\mathcal B}^+}:{\mathcal
B}^+\to\bar{{\mathcal B}}$ to ${\mathcal B}^+$ of the canonical
projection $\mu:{\mathcal A}^+\to V^*$.

Note that the restriction $\tau_{{\mathcal B}^+}:{\mathcal B}^+\to
M$ of the projection $\tau_{{\mathcal A}^+}:{\mathcal A}^+\to M$
is a fibration. Thus, one may consider the prolongation ${\mathcal
T}^{\widetilde{\mathcal A}}{\mathcal B}^+$ of the Lie algebroid
$\widetilde{\mathcal A}$ over $\tau_{{\mathcal B}^+}:{\mathcal
B}^+\to M$. ${\mathcal T}^{\widetilde{\mathcal A}}{\mathcal B}^+$
is a Lie algebroid over ${\mathcal B}^+$.

Now, suppose that
$$\bar{h}',\bar{h}'':\bar{{\mathcal B}}\to
{\mathcal B}^+\in\Gamma(\mu_{{\mathcal B}^+})$$ are two sections
of the AV-bundle $\mu_{{\mathcal B}^+}:{\mathcal
B}^+\to\bar{{\mathcal B}}$ and that
$$h',h'':V^*\to {\mathcal A}^+\in\Gamma(\mu)$$
are arbitrary extensions to $V^*$ of $\bar{h}'$ and $\bar{h}''$,
respectively.

We will denote by
$$F_{h'},F_{h''}:{\mathcal A}^+\to\R$$
the affine functions associated with the sections $h'$ and $h''$,
respectively, and by $X_{F_{h'}}^{\Omega_{\widetilde{{\mathcal
A}}}}$ and $X_{F_{h''}}^{\Omega_{\widetilde{{\mathcal A}}}}$ the
corresponding Hamiltonian sections. Then, we have that
$$\begin{array}{l}
h'(\mu(\varphi_x))-\varphi_x=F_{h'}(\varphi_x)1_{\mathcal
A}(x),\;\;\;\;h''(\mu(\varphi_x))-\varphi_x=F_{h''}(\varphi_x)1_{\mathcal
A}(x),\\[6pt]
i_{X_{F_{h'}}^{\Omega_{\widetilde{{\mathcal
A}}}}}\Omega_{\widetilde{{\mathcal A}}}=d^{{\mathcal
T}^{\widetilde{{\mathcal A}}}{\mathcal
A}^+}F_{h'},\makebox[2cm]{}i_{X_{F_{h''}}^{\Omega_{\widetilde{{\mathcal
A}}}}}\Omega_{\widetilde{{\mathcal A}}}=d^{{\mathcal
T}^{\widetilde{{\mathcal A}}}{\mathcal A}^+}F_{h''},
\end{array}$$
for $\varphi_x\in{\mathcal A}_x^+$.

Moreover, we will prove the following results.

\begin{lemma}\label{lema5.5} The real function on ${\mathcal B}^+$
$$\Omega_{\widetilde{\mathcal A}}({\mathcal P}^+({X_{F_{h'}}^{\Omega_{\widetilde{{\mathcal
A}}}}}_{|{\mathcal B}^+}),{\mathcal
P}^+({X_{F_{h''}}^{\Omega_{\widetilde{{\mathcal A}}}}}_{|{\mathcal
B}^+}))$$ does not depend on the chosen extensions $h'$ and $h''$
of $\bar{h}'$ and $\bar{h}''$, respectively.
\end{lemma}
\begin{proof}
If $\varphi_x\in{\mathcal B}_x^+$ then ${\mathcal
Q}_{\varphi_x}^+(X_{F_{h'}}^{\Omega_{\widetilde{{\mathcal
A}}}}(\varphi_x))\in(H_{\varphi_x}^+)^{\perp,\Omega_{\widetilde{\mathcal
A}}}$ and, thus,
$$\Omega_{\widetilde{\mathcal A}}(\varphi_x)({\mathcal Q}^+_{\varphi_x}(X_{F_{h'}}^{\Omega_{\widetilde{{\mathcal
A}}}}(\varphi_x)),{\mathcal
P}^+_{\varphi_x}(X_{F_{h''}}^{\Omega_{\widetilde{{\mathcal
A}}}}(\varphi_x)))=0$$ which implies that
\begin{equation}\label{sin+}
\Omega_{\widetilde{\mathcal A}}({\mathcal
P}^+({X_{F_{h'}}^{\Omega_{\widetilde{{\mathcal A}}}}}_{|{\mathcal
B}^+}),{\mathcal P}^+({X_{F_{h''}}^{\Omega_{\widetilde{{\mathcal
A}}}}}_{|{\mathcal B}^+}))=\Omega_{\widetilde{\mathcal
A}}({X_{F_{h'}}^{\Omega_{\widetilde{{\mathcal A}}}}}_{|{\mathcal
B}^+},{\mathcal P}^+({X_{F_{h''}}^{\Omega_{\widetilde{{\mathcal
A}}}}}_{|{\mathcal B}^+})).
\end{equation}

Now, let $h_1':V^*\to{\mathcal A}^+$ be another extension of
$\bar{h}':\bar{\mathcal B}\to{\mathcal B}^+$. Then, it is clear
that
$$(F_{h'}-F_{h_1'})_{|{\mathcal B}^+}=0.$$
Therefore, since ${\mathcal
P}^+_{\varphi_x}(X_{F_{h''}}^{\Omega_{\widetilde{\mathcal
A}}}(\varphi_x))\in H^+_{\varphi_x}\subseteq {\mathcal
T}_{\varphi_x}^{\widetilde{\mathcal A}}{\mathcal B}^+$, we obtain
that
$$\Omega_{\widetilde{\mathcal A}}(\varphi_x)(X_{F_{h'}}^{\Omega_{\widetilde{{\mathcal
A}}}}(\varphi_x)-X_{F_{h'_1}}^{\Omega_{\widetilde{{\mathcal
A}}}}(\varphi_x),{\mathcal
P}^+_{\varphi_x}(X_{F_{h''}}^{\Omega_{\widetilde{{\mathcal
A}}}}(\varphi_x)))=\rho_{\widetilde{\mathcal A}}^{\tau_{{\mathcal
A}^+}}({\mathcal
P}^+_{\varphi_x}(X_{F_{h''}}^{\Omega_{\widetilde{\mathcal
A}}}(\varphi_x)))(F_{h'}-F_{h'_1})=0.$$ This proves that
$$\Omega_{\widetilde{\mathcal A}}({\mathcal
P}^+({X_{F_{h'}}^{\Omega_{\widetilde{{\mathcal A}}}}}_{|{\mathcal
B}^+}),{\mathcal P}^+({X_{F_{h''}}^{\Omega_{\widetilde{{\mathcal
A}}}}}_{|{\mathcal B}^+}))=\Omega_{\widetilde{\mathcal
A}}({\mathcal P}^+({X_{F_{h'_1}}^{\Omega_{\widetilde{{\mathcal
A}}}}}_{|{\mathcal B}^+}),{\mathcal
P}^+({X_{F_{h''}}^{\Omega_{\widetilde{{\mathcal A}}}}}_{|{\mathcal
B}^+})).$$ Consequently, using that $\Omega_{\widetilde{\mathcal
A}}$ is skew-symmetric, we deduce the result.
\end{proof}

\begin{lemma}\label{lema5.6} The real function on ${\mathcal B}^+$
$$\Omega_{\widetilde{\mathcal A}}({\mathcal P}^+({X_{F_{h'}}^{\Omega_{\widetilde{{\mathcal
A}}}}}_{|{\mathcal B}^+}),{\mathcal
P}^+({X_{F_{h''}}^{\Omega_{\widetilde{{\mathcal A}}}}}_{|{\mathcal
B}^+}))$$ is basic with respect to the affine bundle projection
$\mu_{{\mathcal B}^+}:{\mathcal B}^+\to\bar{\mathcal B}$.
\end{lemma}
\begin{proof}Denote by $\theta_{h'}$ the section of $({\mathcal
T}^{\widetilde{\mathcal A}}{\mathcal A}^+)^*\to{\mathcal A}^+$
defined by
$$\theta_{h'}=d^{{\mathcal T}^{\widetilde{\mathcal A}}{\mathcal
A}^+}F_{h'}.$$ Then, from (\ref{sin+}), it follows that
$$\Omega_{\widetilde{\mathcal A}}({\mathcal P}^+({X_{F_{h'}}^{\Omega_{\widetilde{{\mathcal
A}}}}}_{|{\mathcal B}^+}),{\mathcal
P}^+({X_{F_{h''}}^{\Omega_{\widetilde{{\mathcal A}}}}}_{|{\mathcal
B}^+}))={\theta_{h'}}_{|{\mathcal
B}^+}({X_{F_{h''}}^{\Omega_{\widetilde{\mathcal A}}}}_{|{\mathcal
B}^+})-{\theta_{h'}}_{|{\mathcal B}^+}({\mathcal
Q}^+({X_{F_{h''}}^{\Omega_{\widetilde{\mathcal A}}}}_{|{\mathcal
B}^+})).$$ Moreover, if $\Pi_{{\mathcal A}^+}$ is the linear
Poisson structure on ${\mathcal A}^+$ induced by the Lie algebroid
${\widetilde{\mathcal A}}$, we have that (see (\ref{2.2'}))
$\theta_{h'}(X_{F_{h''}}^{\Omega_{\widetilde{\mathcal
A}}})=\{F_{h'},F_{h''}\}_{\Pi_{{\mathcal A}^+}}$. In addition, the
vertical bundle of the fibration $\mu:{\mathcal A}^+\to V^*$ is
generated by the vertical lift $1_{\mathcal A}^V$ of the section
$1_{\mathcal A}\in\Gamma(\tau_{{\mathcal A}^+})$. Since
$1_{\mathcal A}$ is a 1-cocycle of $\widetilde{\mathcal A}$, we
deduce that $1_{\mathcal A}^V$ is an infinitesimal Poisson
automorphism of $\Pi_{{\mathcal A}^+}$ and
$$\begin{array}{lcl}
1_{\mathcal A}^V (\{F_{h'},F_{h''}\}_{\Pi_{{\mathcal
A}^+}})&=&\{1_{\mathcal A}^V(F_{h'}),F_{h''}\}_{\Pi_{{\mathcal
A}^+}}+\{F_{h'},1_{\mathcal A}^V(F_{h''})\}_{\Pi_{{\mathcal
A}^+}}\\[6pt]
&=&-\{1,F_{h''}\}_{\Pi_{{\mathcal
A}^+}}-\{F_{h'},1\}_{\Pi_{{\mathcal A}^+}}=0.
\end{array}$$

This implies that the function
$\theta_{h'}(X_{F_{h''}}^{\Omega_{\widetilde{\mathcal A}}})$ is
basic with respect to the fibration $\mu:{\mathcal A}^+\to V^*$.
Thus, the function
$\theta_{h'}(X_{F_{h''}}^{\Omega_{\widetilde{\mathcal
A}}})_{|{\mathcal B}^+}$ is basic with respect to the fibration
$\mu_{{\mathcal B}^+}:{\mathcal B}^+\to\bar{\mathcal B}$.

On the other hand, from (\ref{PQdec}), we deduce that
$$\begin{array}{rcl}\theta_{h'}(\varphi_x)({\mathcal
Q}^+_{\varphi_x}(X_{F_{h''}}^{\Omega_{{\widetilde{\mathcal
A}}}}(\varphi_x)))&=&\theta_{h'}(\varphi_x)((\bar{\tilde{\mathcal
Q}}_{\mu(\varphi_x)}({\mathcal
T}_{\varphi_x}\mu(X_{F_{h''}}^{\Omega_{\widetilde{\mathcal
A}}}(\varphi_x))))_{\varphi_x}^H)\\[6pt]
&-&\bar{\theta}_h(\mu(\varphi_x))(\bar{\tilde{\mathcal
Q}}_{\mu(\varphi_x)}({\mathcal
T}_{\varphi_x}\mu(X_{F_{h''}}^{\Omega_{\widetilde{\mathcal
A}}}(\varphi_x))))\theta_{h'}(\varphi_x)(E(\varphi_x)),\end{array}$$
for $\varphi_x\in{\mathcal B}^+_x$.

Now, using (\ref{conFh}) and (\ref{defE}), one proves that
\begin{equation}\label{thetah'E}
\theta_{h'}(\varphi_x)(E(\varphi_x))=-1_{\mathcal
A}^V(\varphi_x)(F_{h'})
=-\displaystyle\frac{d}{dt}_{|t=0}(F_{h'}(\varphi_x)-t)=1,\mbox{
for }\varphi_x\in{\mathcal A}_x^+.
\end{equation}

Furthermore, proceeding as in Section \ref{sec5.1} (see
(\ref{5.6'})), we have that there exists a section
$R_{h''}\in\Gamma(\tau_{\widetilde{\mathcal A}}^{\tau_V^*})$ such
that
$$({\mathcal T}_{\varphi_x}\mu)(X_{F_{h''}}^{\Omega_{\widetilde{\mathcal
A}}}(\varphi_x))=-R_{h''}(\mu(\varphi_x)),\mbox{ for all
}\varphi_x\in{\mathcal A}^+_x.$$ Therefore,
$$\begin{array}{rcl}\theta_{h'}(\varphi_x)({\mathcal
Q}^+_{\varphi_x}(X_{F_{h''}}^{\Omega_{{\widetilde{\mathcal
A}}}}(\varphi_x)))&=&-\theta_{h'}(\varphi_x)((\bar{\tilde{\mathcal
Q}}_{\mu(\varphi_x)}(R_{h''}(\mu(\varphi_x))))_{\varphi_x}^H)\\[6pt]
&&+\bar{\theta}_h(\mu(\varphi_x))(\bar{\tilde{\mathcal
Q}}_{\mu(\varphi_x)}(R_{h''}(\mu(\varphi_x)))),\mbox{ for
}\varphi_x\in{\mathcal B}^+_x.
\end{array}$$

Consequently, it only remains to prove that the function
$\theta_{h'}(\bar{\tilde{\mathcal Q}}(R_{h''}))^H:{\mathcal
B}^+\to\R$ defined by
$$\theta_{h'}(\bar{\tilde{\mathcal Q}}(R_{h''}))^H(\varphi_x)=\theta_{h'}(\varphi_x)((\bar{\tilde{\mathcal
Q}}_{\mu(\varphi_x)}(R_{h''}(\mu(\varphi_x))))_{\varphi_x}^H)$$ is
$\mu_{{\mathcal B}^+}$-projectable or, equivalently, that
$${\mathcal L}^{{\mathcal T}^{\widetilde{{\mathcal A}}}{\mathcal A}^+}_{E_{|{\mathcal B}^+}}
(\theta_{h'}(\bar{\tilde{\mathcal Q}}(R_{h''}))^H)=0.$$ In fact,
we will see that if $Z\in\Gamma(\tau_{\widetilde{{\mathcal
A}}}^{\tau_V^*})$ and $Z^H$ is the section of ${\mathcal
T}^{\widetilde{{\mathcal A}}}{\mathcal A}^+$ given by
$Z^H(\varphi_x)=(Z(\mu(\varphi_x)))^H_{\varphi_x}$, then
\begin{equation}\label{liethetaZ}
{\mathcal L}^{{\mathcal T}^{\widetilde{{\mathcal A}}}{\mathcal
A}^+}_E(\theta_{h'}(Z^H))=0.
\end{equation}

Note that, from (\ref{thetah'E}), it follows that $${\mathcal
L}^{{\mathcal T}^{\widetilde{{\mathcal A}}}{\mathcal A}^+}_E
\theta_{h'}=d^{{\mathcal T}^{\widetilde{{\mathcal A}}}{\mathcal
A}^+}( {\mathcal L}^{{\mathcal T}^{\widetilde{{\mathcal
A}}}{\mathcal A}^+}_E F_{h'})=d^{{\mathcal
T}^{\widetilde{{\mathcal A}}}{\mathcal A}^+}(\theta_{h'}(E))= 0.$$
Thus,
$${\mathcal L}^{{\mathcal
T}^{\widetilde{{\mathcal A}}}{\mathcal
A}^+}_E(\theta_{h'}(Z^H))=\theta_{h'}(\lcf
E,Z^H\rcf_{\widetilde{{\mathcal A}}}^{\tau_{{\mathcal A}^+}}).$$
Now, since $d^{{\mathcal T}^{\widetilde{{\mathcal A}}}{\mathcal
A}^+}\theta_{h}=0$ and using that $\theta_h(E)=1$ and the fact
that $\theta_h(Z^H)=0$, we obtain that
\begin{equation}\label{partevertical}
\theta_{h}(\lcf E,Z^H\rcf_{\widetilde{{\mathcal
A}}}^{\tau_{{\mathcal A}^+}})=0.
\end{equation}

In addition,
$${\mathcal T}\mu\circ E=0,\;\;\;{\mathcal T}\mu\circ
Z^H=Z\circ\mu$$ and the pair $({\mathcal T}\mu,\mu)$ is a morphism
between the Lie algebroids ${\mathcal T}^{\widetilde{\mathcal
A}}{\mathcal A}^+\to{\mathcal A}^+$ and ${\mathcal
T}^{\widetilde{\mathcal A}}V^*\to V^*$. This implies that
\begin{equation}\label{partehorizontal}
{\mathcal T}\mu\circ \lcf E,Z^H\rcf_{\widetilde{{\mathcal
A}}}^{\tau_{{\mathcal A}^+}}=0.
\end{equation}

Therefore, from (\ref{eq5.10'}), (\ref{partevertical}) and
(\ref{partehorizontal}), we deduce that $\lcf
E,Z^H\rcf_{\widetilde{{\mathcal A}}}^{\tau_{{\mathcal A}^+}}=0$
and (\ref{liethetaZ}) holds.

\end{proof}

From Lemmas \ref{lema5.5} and \ref{lema5.6}, we have that there
exists a real function $\{\bar{h}',\bar{h}''\}_{nh}\in
C^{\infty}(\bar{{\mathcal B}})$ which is characterized by the
following condition
\begin{equation}\label{eq5.26'}
\{\bar{h}',\bar{h}''\}_{nh}\circ\mu_{{\mathcal
B}^+}=\Omega_{\widetilde{{\mathcal A}}}({\mathcal
P}^+({X_{F_{h'}}^{\Omega_{\widetilde{{\mathcal A}}}}}_{|{\mathcal
B}^+}),{\mathcal P}^+({X_{F_{h''}}^{\Omega_{\widetilde{{\mathcal
A}}}}}_{|{\mathcal B}^+})). \end{equation}
 This function is called
\emph{the nonholonomic bracket of the sections $\bar{h}'$ and
$\bar{h}''$} and the resultant map
$$\{\cdot,\cdot\}_{nh}:\Gamma(\mu_{{\mathcal
B}^+})\times\Gamma(\mu_{{\mathcal B}^+})\to C^\infty(\bar{\mathcal
B})$$ is called \emph{the nonholonomic bracket associated with the
regular constrained Hamiltonian system $(h,\bar{\mathcal B})$}.

\begin{theorem} $a)$ The nonholonomic bracket $\{\cdot,\cdot\}_{nh}:\Gamma(\mu_{{\mathcal
B}^+})\times\Gamma(\mu_{{\mathcal B}^+})\to C^\infty(\bar{\mathcal
B})$ associated with the system $(h,\bar{\mathcal B})$ is an
almost-aff-Poisson bracket on the AV-bundle $\mu_{{\mathcal
B}^+}:{\mathcal B}^+\to\bar{\mathcal B}$, that is,
\begin{enumerate}
\item $\{\cdot,\cdot\}_{nh}$ is a bi-affine map,
\item $\{\cdot,\cdot\}_{nh}$ is skew-symmetric and
\item If $\bar{h}'\in\Gamma(\mu_{{\mathcal B}^+})$ then
$$\{\bar{h}',\cdot\}_{nh}:\Gamma(\mu_{{\mathcal B}^+})\to
C^\infty(\bar{{\mathcal
B}}),\;\;\;\bar{h}''\mapsto\{\bar{h}',\bar{h}''\}_{nh},$$ is an
affine derivation.
\end{enumerate}

$b)$ If $\bar{f}\in C^\infty(\bar{\mathcal B})$ is an observable
and $\{\bar{h},\cdot\}_{nh}^{al}:C^\infty(\bar{{\mathcal B}})\to
C^\infty(\bar{{\mathcal B}})$ is the linear map associated with
the affine map $\{\bar{h},\cdot\}_{nh}:\Gamma(\mu_{{\mathcal
B}^+})\to C^\infty(\bar{{\mathcal B}})$, then
$$\dot{\bar{f}}=\{\bar{h},\bar{f}\}_{nh}^{al}$$
where $\dot{\bar{f}}$ is the evolution of $\bar{f}$ along the
solutions of the constrained Hamilton equations.
\end{theorem}
\begin{proof} $a)$ It is clear that $\{\cdot,\cdot\}_{nh}$ is
skew-symmetric.

Next, we will prove that $\{\cdot,\cdot\}_{nh}$ is a bi-affine
map. Suppose that $\bar{h}'$ and $\bar{h}''$ are sections of the
AV-bundle $\mu_{{\mathcal B}^+}:{\mathcal B}^+\to\bar{\mathcal B}$
and that $\bar{f}'$ and $\bar{f}''$ are real functions on
$\bar{\mathcal B}$. If $h',h''\in\Gamma(\mu)$ are arbitrary
extensions of $\bar{h}'$ and $\bar{h}''$, respectively, and
$f',f''\in C^\infty(V^*)$ are arbitrary extensions of $\bar{f}'$
and $\bar{f}''$, respectively, then we can consider the real
functions on ${\mathcal B}^+$ given by
$$\begin{array}{rcl}
\Omega_{\widetilde{\mathcal
A}}({X_{F_{h'}}^{\Omega_{\widetilde{\mathcal A}}}}_{|{\mathcal
B}^+},{\mathcal
P}^+({X_{(f''\circ\mu)}^{\Omega_{\widetilde{\mathcal
A}}}}_{|{\mathcal B}^+}))&=&-\Omega_{\widetilde{\mathcal
A}}({X_{(f''\circ\mu)}^{\Omega_{\widetilde{\mathcal
A}}}}_{|{\mathcal B}^+},{\mathcal
P}^+({X_{F_{h'}}^{\Omega_{\widetilde{\mathcal A}}}}_{|{\mathcal
B}^+})),\\[6pt]
\Omega_{\widetilde{\mathcal
A}}({X_{(f'\circ\mu)}^{\Omega_{\widetilde{\mathcal
A}}}}_{|{\mathcal B}^+},{\mathcal
P}^+({X_{F_{h''}}^{\Omega_{\widetilde{\mathcal A}}}}_{|{\mathcal
B}^+}))&=&-\Omega_{\widetilde{\mathcal
A}}({X_{F_{h''}}^{\Omega_{\widetilde{\mathcal A}}}}_{|{\mathcal
B}^+},{\mathcal
P}^+({X_{(f'\circ\mu)}^{\Omega_{\widetilde{\mathcal
A}}}}_{|{\mathcal
B}^+})),\\[6pt]
\Omega_{\widetilde{\mathcal
A}}({X_{(f'\circ\mu)}^{\Omega_{\widetilde{\mathcal
A}}}}_{|{\mathcal B}^+},{\mathcal
P}^+({X_{(f''\circ\mu)}^{\Omega_{\widetilde{\mathcal
A}}}}_{|{\mathcal B}^+}))&=&-\Omega_{\widetilde{\mathcal
A}}({X_{(f''\circ\mu)}^{\Omega_{\widetilde{\mathcal
A}}}}_{|{\mathcal B}^+},{\mathcal
P}^+({X_{(f'\circ\mu)}^{\Omega_{\widetilde{\mathcal
A}}}}_{|{\mathcal B}^+})).
\end{array}$$

Now, proceeding as in the proof of Lemma \ref{lema5.5}, we deduce
that these functions don't depend on the chosen extensions
$h',h''$ and $f',f''$ of $\bar{h}',\bar{h}''$ and
$\bar{f}',\bar{f}''$, respectively.

Thus, we can introduce the linear maps
$$\widetilde{\{\bar{h}',\cdot\}_{nh}^{al}}:C^\infty(\bar{\mathcal B})\to C^\infty({\mathcal B}^+),
\;\;\bar{f}''\mapsto
\widetilde{\{\bar{h}',\bar{f}''\}_{nh}^{al}}=\Omega_{\widetilde{\mathcal
A}}({X_{F_{h'}}^{\Omega_{\widetilde{\mathcal A}}}}_{|{\mathcal
B}^+},{\mathcal
P}^+({X_{(f''\circ\mu)}^{\Omega_{\widetilde{\mathcal
A}}}}_{|{\mathcal B}^+})),$$
$$\widetilde{\{\cdot,\bar{h}''\}_{nh}^{la}}:C^\infty(\bar{\mathcal B})\to C^\infty({\mathcal B}^+),
\;\;\bar{f}'\mapsto
\widetilde{\{\bar{f}',\bar{h}''\}_{nh}^{la}}=\Omega_{\widetilde{\mathcal
A}}({X_{(f'\circ\mu)}^{\Omega_{\widetilde{\mathcal
A}}}}_{|{\mathcal B}^+},{\mathcal
P}^+({X_{F_{h''}}^{\Omega_{\widetilde{\mathcal A}}}}_{|{\mathcal
B}^+}))$$
\hspace{8.35cm}$=-\widetilde{\{\bar{h}'',\bar{f}'\}_{nh}^{al}},$

and the bilinear map
$$\begin{array}{rcl}
\widetilde{\{\cdot,\cdot\}_{nh}^{ll}}:C^\infty(\bar{\mathcal
B})\times C^\infty(\bar{\mathcal B})&\to& C^\infty({\mathcal
B}^+)\\[6pt]
(\bar{f}',\bar{f}'')&\mapsto&
\widetilde{\{\bar{f}',\bar{f}''\}_{nh}^{ll}}=\Omega_{\widetilde{\mathcal
A}}({X_{(f'\circ\mu)}^{\Omega_{\widetilde{\mathcal
A}}}}_{|{\mathcal B}^+},{\mathcal
P}^+({X_{(f''\circ\mu)}^{\Omega_{\widetilde{\mathcal
A}}}}_{|{\mathcal B}^+})).\end{array}$$

Using that
$$F_{(h''+f'')}=F_{h''}+(f''\circ\mu),$$
we obtain that
$$\{\bar{h}',\bar{h}''+\bar{f}''\}_{nh}\circ\mu_{{\mathcal B}^+}=\{\bar{h}',\bar{h}''\}_{nh}\circ\mu_{{\mathcal
B}^+}+\widetilde{\{\bar{h}',\bar{f}''\}_{nh}^{al}}.$$ This implies
that there exists ${\{\bar{h}',\bar{f}''\}_{nh}^{al}}\in
C^\infty(\bar{\mathcal B})$ such that
$$\{\bar{h}',\bar{f}''\}_{nh}^{al}\circ\mu_{{\mathcal B}^+}=\widetilde{\{\bar{h}',\bar{f}''\}_{nh}^{al}}$$
and, in addition,
\begin{equation}\label{Affder}
\{\bar{h}',\bar{h}''+\bar{f}''\}_{nh}=\{\bar{h}',\bar{h}''\}_{nh}
+ {\{\bar{h}',\bar{f}''\}_{nh}^{al}}.
\end{equation}
Since
$\widetilde{\{\cdot,\bar{h}''\}_{nh}^{la}}=-\widetilde{\{\bar{h}'',\cdot\}_{nh}^{al}}$,
we have that there exists ${\{\bar{f}',\bar{h}''\}_{nh}^{la}}\in
C^\infty(\bar{\mathcal B})$ such that
$$\{\bar{f}',\bar{h}''\}_{nh}^{la}\circ\mu_{{\mathcal B}^+}=\widetilde{\{\bar{f}',\bar{h}''\}_{nh}^{la}}$$
and
\begin{equation}\label{Affizq}
\{\bar{h}'+\bar{f}',\bar{h}''\}_{nh}=\{\bar{h}',\bar{h}''\}_{nh} +
{\{\bar{f}',\bar{h}''\}_{nh}^{la}}.
\end{equation}

On the other hand, from (\ref{Affder}) and (\ref{Affizq}), it
follows that
$$\{\bar{h}'+\bar{f}',\bar{h}''+\bar{f}''\}_{nh}\circ\mu_{{\mathcal B}^+}=(\{\bar{h}',\bar{h}''\}_{nh}
+\{\bar{h}',\bar{f}''\}_{nh}^{al}
+\{\bar{f}',\bar{h}''\}_{nh}^{la})\circ\mu_{{\mathcal B}^+}+
\widetilde{\{\bar{f}',\bar{f}''\}_{nh}^{ll}}.$$

Therefore, there exists ${\{\bar{f}',\bar{f}''\}_{nh}^{ll}}\in
C^\infty(\bar{\mathcal B})$ such that
$$\{\bar{f}',\bar{f}''\}_{nh}^{ll}\circ\mu_{{\mathcal B}^+}=\widetilde{\{\bar{f}',\bar{f}''\}_{nh}^{ll}}$$
and
$$\{\bar{h}'+\bar{f}',\bar{h}''+\bar{f}''\}_{nh}=\{\bar{h}',\bar{h}''\}_{nh}
+\{\bar{h}',\bar{f}''\}_{nh}^{al}
+\{\bar{f}',\bar{h}''\}_{nh}^{la}+
{\{\bar{f}',\bar{f}''\}_{nh}^{ll}}.$$ This proves that the map
$\{\cdot,\cdot\}_{nh}$ is bi-affine.

Moreover, the linear map associated with the affine map
$\{\bar{h}',\cdot\}_{nh}$ is just the map
$\{\bar{h}',\cdot\}_{nh}^{al}:C^\infty(\bar{\mathcal B})\to
C^\infty(\bar{\mathcal B})$ and we have that
$$\{\bar{h}',\bar{f}\}_{nh}^{al}\circ\mu_{{\mathcal
B}^+}=-d^{{\mathcal T}^{\widetilde{\mathcal A}}{\mathcal
A}^+}(f\circ\mu)_{|{\mathcal B}^+}({\mathcal
P}^+({X_{F_{h'}}^{\Omega_{\widetilde{\mathcal A}}}}_{|{\mathcal
B}^+})),$$ for $\bar{f}\in C^\infty(\bar{\mathcal B})$, where
$f\in C^\infty(V^*)$ is an arbitrary extension of $\bar{f}$. Thus,
since the pair $({\mathcal T}\mu,\mu)$ is a Lie algebroid
morphism, it follows that
\begin{equation}\label{mumorph}
\{\bar{h}',\bar{f}\}_{nh}^{al}(\mu_{{\mathcal
B}^+}(\varphi_x))=-(({\mathcal T}\mu,\mu)^*(d^{{\mathcal
T}^{\widetilde{\mathcal A}}V^*}f))(\varphi_x)({\mathcal
P}^+(X_{F_{h'}}^{\Omega_{\widetilde{\mathcal A}}})(\varphi_x)).
\end{equation}
In addition, proceeding as in Section \ref{sec5.1} (see
(\ref{5.6'})), we deduce that there exists a section
$R_{h'}\in\Gamma(\tau_{\widetilde{\mathcal A}}^{\tau_V^*})$ such
that
\begin{equation}\label{siempre}
({\mathcal
T}_{\varphi_x}\mu)(X_{F_{h'}}^{\Omega_{\widetilde{\mathcal
A}}}(\varphi_x))=-R_{h'}(\mu_{{\mathcal B}^+}(\varphi_x)),\mbox{
for }\varphi_x\in{\mathcal B}^+_x,
\end{equation}
and, using (\ref{Pconm}), we obtain that
\begin{equation}\label{P+Pbar}
({\mathcal T}_{\varphi_x}\mu)({\mathcal P}^+_{\varphi_x}(
X_{F_{h'}}^{\Omega_{\widetilde{\mathcal
A}}}(\varphi_x)))=-\bar{\tilde{{\mathcal P}}}_{\mu_{{\mathcal
B}^+}(\varphi_x)}( R_{h'}(\mu_{{\mathcal B}^+}(\varphi_x))),\mbox{
for }\varphi_x\in{\mathcal B}^+_x.
\end{equation}
Therefore, from (\ref{mumorph}), (\ref{P+Pbar}) and since
${\mathcal T}^{\widetilde{\mathcal A}}\bar{\mathcal B}$ is a Lie
subalgebroid of ${\mathcal T}^{\widetilde{\mathcal A}}V^*$, we
conclude that
\begin{equation}\label{alnh}
\{\bar{h}',\bar{f}\}_{nh}^{al}=\rho_{\widetilde{\mathcal
A}}^{\tau_{\bar{\mathcal B}}}(\bar{\tilde{\mathcal P}}
(R_{h'}))(\bar{f}).
\end{equation}
Consequently, $\{\bar{h}',\cdot\}_{nh}^{al}$ is a vector field on
$\bar{\mathcal B}$ and $\{\bar{h}',\cdot\}_{nh}$ is an affine
derivation.

$b)$ It follows using (\ref{alnh}) and the fact that the solution
of the constrained Hamilton equations is the section
$\bar{\tilde{\mathcal P}}({R_h}_{|\bar{\mathcal B}})$.

\end{proof}

Now, let $\bar{h}',\bar{h}''$ be two sections of the AV-bundle
$\mu_{{\mathcal B}^+}:{\mathcal B}^+\to\bar{\mathcal B}$ and
$h',h''\in\Gamma(\mu)$ be two arbitrary extensions to $V^*$ of
$\bar{h}'$ and $\bar{h}''$, respectively. Then, we may consider
the nonholonomic bracket $\{\bar{h}',\bar{h}''\}_{nh}$ of the
sections $\bar{h}'$ and $\bar{h}''$ and the aff-Poisson bracket
$\{h',h''\}$ of the sections $h'$ and $h''$ defined in
(\ref{eq3.82}). Using (\ref{2.2'}), (\ref{eq3.81}), (\ref{locP+}),
(\ref{sin+}), (\ref{eq5.26'}) and (\ref{siempre}), we obtain that
$\{\bar{h}',\bar{h}''\}_{nh}$ and $\{h',h''\}$ are locally related
by the following condition
$$\begin{array}{rcl}
\{\bar{h}',\bar{h}''\}_{nh}&=&\Big(\{h',h''\}-\bar{\mathcal
C}_{ab}\bar{\mathcal
C}_{cd}\{\psi^b,\psi^d\}_h\,\bar{\mu}^c(R_{h''})\bar{\mu}^a(R_{h'})-\bar{\mathcal
C}_{ab}\{\psi^b,h''\}_V\bar{\mu}^a(R_{h'})\\[6pt]
&&+\bar{\mathcal
C}_{ab}\{\psi^a,h'\}_V\bar{\mu}^b(R_{h''})\Big)_{|\bar{\mathcal
B}},
\end{array}$$

where $\{\cdot,h''\}_V:C^\infty(V^*)\to C^\infty(V^*)$ is the
linear map associated with the affine map
$\{\cdot,h''\}:\Gamma(\mu)\to C^\infty(V^*)$ (see (\ref{eq3.83})).
Thus, if the local expressions of the sections $h'$ and $h''$ are
$h'(x^i,y_\alpha)=(x^i,-H'(x^j,y_\beta),y_\alpha)$ and
$h''(x^i,y_\alpha)=(x^i,-H''(x^j,y_\beta),y_\alpha)$, then, using
(\ref{Omegahloc}), (\ref{rh}), (\ref{eq3.81}), (\ref{eq3.82}),
(\ref{eq3.83}), (\ref{eq3.85}), (\ref{mubarraa}) and
(\ref{eq5.17'}), we deduce that $$\begin{array}{l}
\kern-5pt\{\bar{h}',\bar{h}''\}_{nh}\kern-2pt=\rho_0^i\displaystyle\frac{\partial(H'\kern-4pt-\kern-2ptH'')}{\partial
x^i} +\rho^i_\alpha(\displaystyle\frac{\partial H'}{\partial
x^i}\frac{\partial H''}{\partial y_\alpha}-\frac{\partial
H'}{\partial y_\alpha}\frac{\partial H''}{\partial
x^i})+C_{0\alpha}^\gamma
y_\gamma\frac{\partial(H'\kern-4pt-\kern-2ptH'')}{\partial
y_\alpha}- C_{\alpha\beta}^\gamma
y_\gamma\displaystyle\frac{\partial H'}{\partial
y_\alpha}\frac{\partial H''}{\partial
y_\beta}\\[10pt]
\kern-5pt+\bar{\mathcal C}_{ab}\bar{\mathcal
C}_{cd}\Big[\rho_\alpha^i\Big(\displaystyle\frac{\partial
\psi^b}{\partial x^i}\frac{\partial \psi^d}{\partial
y_\alpha}-\frac{\partial \psi^b}{\partial y_\alpha}\frac{\partial
\psi^d}{\partial x^i}\Big)+C_{\alpha\beta}^\gamma
y_\gamma\displaystyle\frac{\partial \psi^b}{\partial
y_\beta}\frac{\partial \psi^d}{\partial
y_\alpha}\Big]\frac{\partial \psi^c}{\partial y_\gamma}{\mathcal
H}^{\gamma\beta}\frac{\partial (H''-H)}{\partial
y_\beta}\frac{\partial \psi^a}{\partial y_\theta}{\mathcal
H}^{\theta\nu}\frac{\partial (H'-H)}{\partial y_\nu}\\[10pt]
\kern-5pt+\bar{\mathcal C}_{ab}\Big[
\rho_0^i\displaystyle\frac{\partial \psi^b}{\partial
x^i}-\rho_\alpha^i\Big(\frac{\partial \psi^b}{\partial
y_\alpha}\frac{\partial H''}{\partial x^i}-\frac{\partial
\psi^b}{\partial x^i}\frac{\partial H''}{\partial
y_\alpha}\Big)+\frac{\partial \psi^b}{\partial
y_\alpha}y_\gamma(C_{0\alpha}^\gamma
-C_{\alpha\beta}^\gamma\frac{\partial H''}{\partial
y_\beta})\Big]\frac{\partial \psi^a}{\partial y_\gamma}{\mathcal
H}^{\gamma\beta}\frac{\partial (H'-H)}{\partial y_\beta}\\[10pt]
\kern-5pt-\bar{\mathcal C}_{ab}\Big[
\rho_0^i\displaystyle\frac{\partial \psi^a}{\partial
x^i}-\rho_\alpha^i\Big(\frac{\partial \psi^a}{\partial
y_\alpha}\frac{\partial H'}{\partial x^i}-\frac{\partial
\psi^a}{\partial x^i}\frac{\partial H'}{\partial
y_\alpha}\Big)+\frac{\partial \psi^a}{\partial
y_\alpha}y_\gamma(C_{0\alpha}^\gamma
-C_{\alpha\beta}^\gamma\frac{\partial H'}{\partial
y_\beta})\Big]\frac{\partial \psi^b}{\partial y_\gamma}{\mathcal
H}^{\gamma\beta}\frac{\partial (H''-H)}{\partial y_\beta}.
\end{array}$$

\section{Examples}

\subsection{Lagrangian systems with linear nonholonomic
constraints on a Lie algebroid}

In this section, we will discuss the particular case of Lagrangian
systems with linear nonholonomic constraints on a Lie algebroid
(see \cite{CoLeMaMa}).

First of all, we will recall an standard construction which will
be useful in the sequel.

Let $\tau_E:E\to N$ be a Lie algebroid over a manifold $N$ with
Lie algebroid structure $(\lcf\cdot,\cdot\rcf_E,\rho_E)$. Then,
the vector bundle $\tilde{\tau}_E:E\times\R\to N$ admits a natural
Lie algebroid structure $(\lcf\cdot,\cdot\rcf_{E\times\R},$ $
\rho_{E\times\R})$ given by
$$\begin{array}{rcl}
\lcf(X,f),(Y,g)\rcf_{E\times\R}&=&(\lcf
X,Y\rcf_E,\rho_E(X)(g)-\rho_E(Y)(f)),\\[4pt]
\rho_{E\times\R}(X,f)&=&\rho_E(X),
\end{array}$$
for $(X,f),(Y,g)\in\Gamma(\tau_E)\times C^\infty(N)$.

Next, suppose that $\tau_V:V\to M$ is a Lie algebroid over a
manifold $M$ with Lie algebroid structure
$(\lcf\cdot,\cdot\rcf_V,\rho_V)$. Then, the affine bundle
$\tau_{\mathcal A}=\tau_V:{\mathcal A}=V\to M$ is a Lie affgebroid
over $M$. In fact, the bidual bundle $\tau_{\widetilde{\mathcal
A}}: \widetilde{\mathcal A}\to M$ may be identified with the Lie
algebroid $\tilde{\tau}_V:V\times\R\to M$. In addition, it is easy
to prove that the prolongation ${\mathcal T}^{\widetilde{\mathcal
A}}{\mathcal A}$ of $\widetilde{\mathcal A}$ over the fibration
$\tau_{\mathcal A}:{\mathcal A}\to M$ is isomorphic to the Lie
algebroid $\tilde{\tau}_V^{\tau_V}:{\mathcal T}^VV\times\R\to V$.

We will denote by $S:{\mathcal T}^VV\to{\mathcal T}^VV$ the
vertical endomorphism on ${\mathcal T}^VV$ (see \cite{LMM,EMar})
and by $E_0$ the section $(0,1)$ of
$\tilde{\tau}_V^{\tau_V}:{\mathcal T}^VV\times\R\to V$.

Now, let $L:V\to\R$ be a Lagrangian function on $V$ and
$E_L=\Delta(L)-L\in C^\infty(V)$ be the Lagrangian energy (here,
$\Delta$ is the Liouville vector field of $V$). Then, the dual
bundle to $\tau_{\widetilde{\mathcal A}}^{\tau_{\mathcal
A}}:{\mathcal T}^{\widetilde{\mathcal A}}{\mathcal A}\to{\mathcal
A}$ is isomorphic to the vector bundle
$(\tilde{\tau}_V^{\tau_V})^*:({\mathcal T}^VV)^*\times\R\to V$
and, under this identification, the 1-cocycle $\phi_0$ is the
section $(0,1)$ and
\begin{equation}\label{Na-a}
\Omega_L=\omega_L+(0,1)\wedge d^{{\mathcal T}^VV}E_L,
\end{equation}
where $\omega_L:V\to \wedge^2({\mathcal T}^VV)^*$ is the
Poincar\'{e}-Cartan 2-section associated with the Lagrangian $L$
on $V$ (\cite{LMM,EMar}).

Next, suppose that $\tau_U:U\to M$ is a vector subbundle of
$\tau_V:V\to M$. Then, $\tau_{\mathcal B}=\tau_U:{\mathcal B}=U\to
M$ is an affine subbundle of $\tau_{\mathcal A}:{\mathcal A}\to M$
and we can consider the affine constrained Lagrangian system
$(L,{\mathcal B})$ on ${\mathcal A}$. Moreover, if $\tilde{X}$ is
a section of $\tau_{\widetilde{\mathcal A}}^{\tau_{\mathcal A}}=
\tilde{\tau}_V^{\tau_V}:{\mathcal T}^{\widetilde{\mathcal
A}}{\mathcal A}\cong {\mathcal T}^VV\times\R\to V$ then, from
(\ref{Na-a}), it follows that $\tilde{X}$ is a solution of the
Lagrange-d'Alembert equations for the affine nonholonomic system
$(L,{\mathcal B})$ if and only if
\begin{equation}\label{Lnheqs}
\begin{array}{c}
Y=\tilde{Y}_{|U}=(\tilde{X}-E_0)_{|U}\mbox{ is a SODE on }U,\\[4pt]
(i_Y\omega_L-d^{{\mathcal
T}^VV}E_L)_{|U}\in\Gamma(\tau_{S^*({\mathcal
T}^VU)^\circ}),\\[4pt]
Y_{|U}\in\Gamma(\tau_V^{\tau_U}).
\end{array}\end{equation}

(\ref{Lnheqs}) are just the Lagrange-d'Alembert equations for the
linear nonholonomic system $(L,U)$ in the terminology of
\cite{CoLeMaMa}.

Thus, the affine constrained Lagrangian system $(L,{\mathcal B})$
on ${\mathcal A}$ is regular if and only if the linear
nonholonomic system $(L,U)$ on $V$ is regular in the sense of
\cite{CoLeMaMa}.

Therefore, using the results in Section \ref{secSol} of this
paper, one directly deduces some results, for linear nonholonomic
systems on Lie algebroids, which were obtained in \cite{CoLeMaMa}
(see Section 3 in \cite{CoLeMaMa}).

Now, assume that the Lagrangian function $L:{\mathcal A}=V\to\R$
is hyperregular (that is, the Legendre transformation
$leg_L:{\mathcal A}=V\to V^*$ is a global diffeomorphism) and that
the affine constrained Lagrangian system $(L,{\mathcal B})$ is
regular. Then, the space ${\mathcal A}^+$ may be identified with
the product manifold $V^*\times\R$ and, under this identification,
we have that:
\begin{enumerate}
\item The vector bundle projection $\tau_{{\mathcal
A}^+}:{\mathcal A}^+\to M$ is the map
$\tilde{\tau}_V^*:V^*\times\R\to M$ given by
$$\tilde{\tau}_V^*(\alpha_x,t)=\tau_V^*(\alpha_x)=x,\mbox{ for
}(\alpha_x,t)\in V^*_x\times\R\mbox{ and }x\in M.$$
\item The fibration $\mu:{\mathcal A}^+\to V^*$ is the canonical
projection $pr_1:V^*\times\R\to V^*$ on the first factor and,
thus, the spaces $\Gamma(\mu)$ and $C^\infty(V^*)$ are isomorphic.
\item The Hamiltonian section $h=Leg_L\circ
leg_L^{-1}\in\Gamma(\mu)$ is the Hamiltonian energy $H\in
C^\infty(V^*)$ of the free system, that is, $H=E_L\circ
leg_L^{-1}:V^*\to\R$.
\item The canonical aff-Poisson bracket on the trivial AV-bundle $pr_1:V^*\times\R\to
V^*$ is just the linear Poisson bracket on $V^*$ induced by the
Lie algebroid structure on $V$.
\end{enumerate}

Next, let $\bar{\mathcal B}=\bar{U}$ be the Hamiltonian
constrained submanifold of $V^*$, that is, $\bar{\mathcal
B}=\bar{U}=leg_L(U)$. Then, the constrained AV-bundle
$\mu_{{\mathcal B}^+}:{\mathcal B}^+\to\bar{\mathcal B}$ may be
identified with the trivial AV-bundle
$pr_1:\bar{U}\times\R\to\bar{U}$. Therefore, the nonholonomic
bracket associated with the nonholonomic system $(L,{\mathcal B})$
is an almost-Poisson bracket on $\bar{U}$, i.e., a $\R$-bilinear
map
$$\{\cdot,\cdot\}_{nh}:C^\infty(\bar{U})\times
C^\infty(\bar{U})\to C^\infty(\bar{U})$$ which is skew-symmetric
and a derivation in each argument with respect to the standard
product of functions ($\{\cdot,\cdot\}_{nh}$ doesn't satisfy, in
general, the Jacobi identity).

Finally, since the restriction of the Legendre transformation to
$U$
$$leg_U={leg_L}_{|U}:U\to\bar{U}$$
is a global diffeomorphism, one may consider the corresponding
nonholonomic bracket on $U$, which we also denote by
$\{\cdot,\cdot\}_{nh}$, defined by
$$\{f,g\}_{nh}=\{f\circ leg_U^{-1},g\circ leg_U^{-1}\}_{nh}\circ
leg_U,\mbox{ for }f,g\in C^\infty(U).$$ This bracket was
introduced in \cite{CoLeMaMa} and its properties were discussed in
this paper (see Section 3.5 in \cite{CoLeMaMa}).

\subsection{Standard affine nonholonomic Lagrangian systems}

Let $\tau:M\to\R$ be a fibration and ${\mathcal A}=J^1\tau$ be the
1-jet bundle of local sections of $\tau:M\to\R$. Then, as we know
(see Section \ref{secaff}), the affine bundle $\tau_{\mathcal
A}=\tau_{1,0}:{\mathcal A}=J^1\tau\to M$ is a Lie affgebroid
modelled on the Lie algebroid $\tau_V=(\pi_M)_{|V\tau}:V=V\tau\to
M$. Moreover, the bidual bundle $\widetilde{\mathcal A}$ may be
identified with the standard Lie algebroid $\pi_M:TM\to M$ and,
under this identification, the 1-cocycle $1_{\mathcal A}$ of
$\widetilde{\mathcal A}$ is just the closed 1-form $\tau^*(dt)$.

Note that if $(t,q^\alpha)$ are local coordinates on $M$ which are
adapted to the fibration $\tau$ then
$\{e_0=\displaystyle\frac{\partial}{\partial
t},e_\alpha=\frac{\partial}{\partial q^\alpha}\}$ is a local basis
of sections of the Lie algebroid $\tau_{\widetilde{\mathcal
A}}=\pi_M:\widetilde{\mathcal A}=TM\to M$ such that $1_{\mathcal
A}(e_0)=1$ and $1_{\mathcal A}(e_\alpha)=0$. Furthermore, we have
that
\begin{equation}\label{Eqsstan}
\rho_{\widetilde{\mathcal
A}}(e_0)=\displaystyle\frac{\partial}{\partial
t},\;\;\;\rho_{\widetilde{\mathcal
A}}(e_\alpha)=\frac{\partial}{\partial q^\alpha},\;\;\;\lcf
e_0,e_\alpha\rcf_{\widetilde{\mathcal A}}=\lcf
e_\alpha,e_\beta\rcf_{\widetilde{\mathcal A}}=0,
\end{equation}
for all $\alpha$ and $\beta$.

On the other hand, it is easy to prove that the prolongation
$\tau_{\widetilde{\mathcal A}}^{\tau_{\mathcal A}}:{\mathcal
T}^{\widetilde{\mathcal A}}{\mathcal A}\to{\mathcal A}$ of
$\widetilde{\mathcal A}$ over the fibration $\tau_{\mathcal
A}:{\mathcal A}\to M$ is isomorphic to the standard Lie algebroid
$\pi_{J^1\tau}:T(J^1\tau)\to J^1\tau$. Thus, the space
$\Gamma(\tau_{\widetilde{\mathcal A}}^{\tau_{\mathcal A}})$ may be
identified with the set of vector fields on $J^1\tau$.

Now, let $L:J^1\tau\to\R$ be a regular Lagrangian function. If
$\tau_1=\tau\circ\tau_{1,0}:J^1\tau\to\R$ is the canonical
projection then, under the above identifications, the 1-cocycle
$\phi_0$ of the Lie algebroid $\tau_{\widetilde{\mathcal
A}}^{\tau_{\mathcal A}}:{\mathcal T}^{\widetilde{\mathcal
A}}{\mathcal A}\to{\mathcal A}$ is the closed 1-form
$\eta=\tau_1^*(dt)$ on $J^1\tau$. Moreover, the
Poincar\'{e}-Cartan 2-section $\Omega_L$ associated with $L$ is
just the Poincar\'{e}-Cartan 2-form on $J^1\tau$.

Next, suppose that $\tau_{\mathcal B}:{\mathcal B}\to M$ is an
affine subbundle of $\tau_{\mathcal A}=\tau_{1,0}:{\mathcal
A}=J^1\tau\to M$. Then, we can consider the affine constrained
Lagrangian system $(L,{\mathcal B})$.

In addition, we have that a vector field $X$ on $J^1\tau$ is a
solution of the Lagrange-d'Alembert equations if and only if
\begin{equation}\label{DynEqs}
(i_X\Omega_L)_{|{\mathcal B}}\in\Gamma(\tau_{S^*(T{\mathcal
B})^\circ}),\;\;(i_X\eta)_{|{\mathcal B}}=1,\;\;X_{|{\mathcal
B}}\in{\mathcal X}({\mathcal B}).
\end{equation}

These equations were considered in \cite{LMdM} (see also
\cite{LeMaMa2}). Note that the first two equations imply that $X$
is a SODE on ${\mathcal B}$.

From (\ref{DynEqs}), it follows that the affine nonholonomic
Lagrangian system $(L,{\mathcal B})$ on ${\mathcal A}$ is regular
if and only if it is regular in the sense of \cite{LMdM}. Thus,
using the results in Section \ref{secSol}, one directly deduces
some results, for standard affine nonholonomic Lagrangian systems
which were obtained in \cite{LMdM} (see also \cite{LeMaMa2}).

On the other hand, we have that the space ${\mathcal A}^+$ may be
identified with the cotangent bundle $T^*M$ to $M$ and, under this
identification, the fibration $\mu:{\mathcal A}^+\to V^*$ is just
the dual map $i^*_{V\tau}:T^*M\to V^*\tau$ of the canonical
inclusion $i_{V\tau}:V\tau\to TM$.

Next, assume that the Lagrangian function $L:{\mathcal
A}=J^1\tau\to\R$ is hyperregular and that the affine constrained
Lagrangian system $(L,{\mathcal B})$ is regular. Then, we can
consider the Hamiltonian constrained submanifold $\bar{\mathcal
B}=leg_L({\mathcal B})$ of $V^*\tau$ and the constrained AV-bundle
$\mu_{{\mathcal B}^+}:{\mathcal
B}^+=(i^*_{V\tau})^{-1}(\bar{\mathcal B})\to\bar{\mathcal B}$.

Now, denote by $h:V^*\tau\to T^*M$ the Hamiltonian section induced
by the hyperregular Lagrangian $L$ ($h=Leg_L\circ leg_L^{-1}$) and
suppose that $\bar{h}'$ and $\bar{h}''$ are two sections of the
AV-bundle $\mu_{{\mathcal B}^+}:{\mathcal B}^+\to\bar{\mathcal B}$
and that $h',h'':V^*\tau\to T^*M\in\Gamma(i_{V\tau}^*)$ are
arbitrary extensions of $\bar{h}'$ and $\bar{h}''$, respectively.
Then, we will obtain the local expression of the nonholonomic
bracket of $\bar{h}'$ and $\bar{h}''$. For this purpose, we take
canonical local coordinates $(t,q^\alpha,p_\alpha)$ (respectively,
$(t,q^\alpha,p,p_\alpha)$) on $V^*\tau$ (respectively, $T^*M$)
such that
$$\begin{array}{rcl}
h(t,q^\alpha,p_\alpha)&=&(t,q^\alpha,-H(t,q^\beta,p_\beta),p_\alpha),\\[6pt]
h'(t,q^\alpha,p_\alpha)&=&(t,q^\alpha,-H'(t,q^\beta,p_\beta),p_\alpha),\\[6pt]
h''(t,q^\alpha,p_\alpha)&=&(t,q^\alpha,-H''(t,q^\beta,p_\beta),p_\alpha)
\end{array}$$

and the local equations defining $\bar{\mathcal B}$ as a
submanifold of $V^*\tau$ are
$$\psi^a(t,q^\alpha,p_\alpha)=0,\;\;a\in\{1,\dots,r\}.$$

Then, using the general local expression of the nonholonomic
bracket (see Section \ref{secnhbrac}) and (\ref{Eqsstan}), we
deduce that
$$\begin{array}{rcl}
\{\bar{h}',\bar{h}''\}_{nh}&=&\Big\{\displaystyle\frac{\partial
H'}{\partial t}-\frac{\partial H''}{\partial
t}+\displaystyle\frac{\partial H'}{\partial
q^\alpha}\frac{\partial H''}{\partial p_\alpha}-\frac{\partial
H'}{\partial p_\alpha}\frac{\partial
H''}{\partial q^\alpha}\\[10pt]
&+& \bar{C}_{ab}\bar{C}_{cd}\Big( \displaystyle\frac{\partial
\psi^b}{\partial q^\alpha}\frac{\partial \psi^d}{\partial
p_\alpha}-\frac{\partial\psi^b}{\partial
p_\alpha}\frac{\partial\psi^d}{\partial
q^\alpha}\Big)\frac{\partial\psi^c}{\partial p_\gamma}{\mathcal
H}^{\gamma\beta}\frac{\partial(H''-H)}{\partial
p_\beta}\frac{\partial \psi^a}{\partial p_\theta}{\mathcal
H}^{\theta\nu}\frac{\partial(H'-H)}{\partial p_\nu}\\[10pt]
&+&\bar{\mathcal
C}_{ab}\Big[\displaystyle\frac{\partial\psi^b}{\partial t}+\Big(
\displaystyle\frac{\partial \psi^b}{\partial
q^\alpha}\frac{\partial H''}{\partial
p_\alpha}-\frac{\partial\psi^b}{\partial p_\alpha}\frac{\partial
H''}{\partial q^\alpha}\Big)\Big]\frac{\partial\psi^a}{\partial
p_\gamma}{\mathcal H}^{\gamma\beta}\frac{\partial(H'-H)}{\partial
p_\beta}\\[10pt]
&-&\bar{\mathcal
C}_{ab}\Big[\displaystyle\frac{\partial\psi^a}{\partial t}+\Big(
\displaystyle\frac{\partial \psi^a}{\partial
q^\alpha}\frac{\partial H'}{\partial
p_\alpha}-\frac{\partial\psi^a}{\partial p_\alpha}\frac{\partial
H'}{\partial q^\alpha}\Big)\Big]\frac{\partial\psi^b}{\partial
p_\gamma}{\mathcal H}^{\gamma\beta}\frac{\partial(H''-H)}{\partial
p_\beta} \Big\}_{|\bar{\mathcal B}}.
\end{array}$$

Finally, we remark that the linear-linear part of the bi-affine
map $\{\cdot,\cdot\}_{nh}:\Gamma(\mu_{{\mathcal
B}^+})\times\Gamma(\mu_{{\mathcal B}^+})\to C^\infty(\bar{\mathcal
B})$ is an almost-Poisson bracket on $\bar{\mathcal B}$. This
bracket was considered in \cite{CLMM}.

\subsection{A homogeneous rolling ball without sliding
on a rotating table with time-dependent angular
velocity}\label{example}

\cite{BKMM,CLMM,LM,NF} Consider the vector bundle
$\tau_{\widetilde{\mathcal A}}: \widetilde{\mathcal A}\rightarrow
M$ where $\widetilde{\mathcal A}=T\R^3\times \R^3$, $M=\R^3$ and
$\tau_{\widetilde{\mathcal A}}$ is the composition of the
projection over the first factor with the canonical projection
$\pi_{\R^3}:T\R^3\to\R^3$. Denote by $(t, x, y)$ the coordinates
on $M$. A global basis of sections of $\tau_{\widetilde{\mathcal
A}}$ can be construct as follows
\[
\Big\{ e_0=(\frac{\partial}{\partial t}, 0),\;
e_1=(\frac{\partial}{\partial x}, 0),\;
e_2=(\frac{\partial}{\partial y}, 0),\; e_3=(0, u_1),\;  e_4=(0,
u_2), \;e_5=(0, u_3) \Big\},
\]
where $u_1, u_2, u_3:\R^3\to \R^3$ are the constant maps $u_1(t,
x, y)=(1,0,0)$, $u_2(t, x, y)=(0,1,0)$ and $u_3(t, x, y)=(0,0,1)$.
Consider the  projection  $\rho_{\widetilde{\mathcal A}}:
T\R^3\times \R^3\to T\R^3$ over the first factor and the Lie
bracket on the space of sections $\Gamma(
\tau_{\widetilde{\mathcal A}})$ where the only non-zero Lie
brackets are
\[
\lcf e_4, e_3\rcf_{\widetilde{\mathcal A}}=e_5,\quad  \lcf e_5,
e_4\rcf_{\widetilde{\mathcal A}}=e_3 \;\;\mbox{ and }\;\; \lcf
e_3, e_5\rcf_{\widetilde{\mathcal A}}=e_4.
\]
Thus, $(\lcf \cdot, \cdot\rcf_{\widetilde{\mathcal
A}},\rho_{\widetilde{\mathcal A}})$ induces a Lie algebroid
structure on $\widetilde{\mathcal A}=T\R^3\times\R^3$. Denote by
$(t, x, y; \dot{t}, \dot{x}, \dot{y},$ $\omega_x,$ $ \omega_y,
\omega_z)$ the coordinates on $\widetilde{\mathcal A}$ induced by
$\{e_0,e_1,e_2,e_3,e_4,e_5\}$.

Moreover, $\phi: T\R^3\times\R^3\to \R$ given by $\phi(t, x, y;
\dot{t}, \dot{x}, \dot{y}, \omega_x, \omega_y, \omega_z)=\dot{t}$
is a 1-cocycle in the co\-rres\-ponding Lie algebroid cohomology
and, then, it induces a Lie affgebroid structure over ${\mathcal
A}=\phi^{-1}\{1\}\equiv \R\times T\R^2\times \R^3$. Note that the
Lie affgebroid structure on ${\mathcal A}= \R\times T\R^2\times
\R^3$ is a special type of Lie affgebroid structure called Atiyah
affgebroid structure (see Section 9.3.1 in \cite{IMPS} for a
general construction). Moreover, the affine bundle $\tau_{\mathcal
A}: \R\times T\R^2\times \R^3\rightarrow \R^3$ is modelled on the
vector bundle $\tau_V:V=\phi^{-1}\{0\}\equiv \R\times T\R^2\times
\R^3\to \R^3$. Thus, $(t, x, y; \dot{x}, \dot{y}, \omega_x,
\omega_y, \omega_z)$ may be considered as local coordinates on
${\mathcal A}$ and $V$.

In this case, from (\ref{tildeT}), we have that a basis of
sections of ${\mathcal T}^{\widetilde{{\mathcal A}}}{\mathcal A}$
is
\begin{eqnarray*}
&\Big\{\tilde{T}_0=(e_0, \displaystyle\frac{\partial}{\partial
t}), \; \tilde{T}_1=(e_1, \frac{\partial}{\partial x}), \;
\tilde{T}_2=(e_2, \frac{\partial}{\partial y}),\;
\tilde{T}_3=(e_3, 0), \quad \tilde{T}_4=(e_4,0), \;
\tilde{T}_5=(e_5, 0),\\
&\tilde{V}_1=(0, \displaystyle\frac{\partial}{\partial \dot{x}}),
\; \tilde{V}_2=(0, \frac{\partial}{\partial \dot{y}}),\;
\tilde{V}_3=(0, \frac{\partial}{\partial w_x}), \; \tilde{V}_4=(0,
\frac{\partial}{\partial w_y}), \; \tilde{V}_5=(0,
\frac{\partial}{\partial w_z})\Big\}.&
\end{eqnarray*}

Consider now the following mechanical problem. A (homogeneous)
sphere of radius $r>0$,  mass $m$ and inertia about any axis
$k^2,$ rolls without sliding on a horizontal table which rotates
with time-dependent angular velocity $\Omega(t)$ about a vertical
axis through one of its points. Apart from the constant
gravitational force, no other external forces are assumed to act
on the sphere. Therefore, the Lagrangian of the system corresponds
to the kinetic energy. Moreover,  observe that the kinetic energy
may be expressed as a Lagrangian $L: {\mathcal A}\to \R$:
\[
L(t, x, y;  \dot{x}, \dot{y}, \omega_x, \omega_y,
\omega_z)=\frac{1}{2}(m\dot{x}^2+m\dot{y}^2 + k^2(\omega_x^2 +
\omega_y^2 + \omega_z^2)), \] where $(\omega_x, \omega_y,
\omega_z)$ are the components of the angular velocity of the
sphere (see \cite{CLMM}, for more details).

After some straightforward calculations using (\ref{PCloc}), we
deduce that the Poincar\'{e}-Cartan sections associated with $L$
are given by:
\begin{eqnarray*}
\Theta_L&=& -L \phi_0+m\dot{x} \tilde{T}^1+m\dot{y}
\tilde{T}^2+k^2 \omega_x\tilde{T}^3 +k^2\omega_y
\tilde{T}^4+k^2\omega_z \tilde{T}^5,\\
\Omega_L&=& d^{{\mathcal T}^{\widetilde{{\mathcal A}}}{\mathcal
A}}L\wedge \phi_0+m \tilde{T}^1\wedge \tilde{V}^1 +m
\tilde{T}^2\wedge \tilde{V}^2+k^2\tilde{T}^3\wedge \tilde{V}^3
+k^2\tilde{T}^4\wedge \tilde{V}^4+k^2\tilde{T}^5\wedge
\tilde{V}^5\\
&&+k^2\omega_x \tilde{T}^5\wedge \tilde{T}^4+k^2\omega_y
\tilde{T}^3\wedge \tilde{T}^5+k^2\omega_z \tilde{T}^4\wedge
\tilde{T}^3.
\end{eqnarray*}

Since the ball is rolling without sliding on a rotating table then
the system is subjected to the affine constraints:
\[
\begin{array}{lcr}
 \Psi^1=\Omega(t) y+\dot{x}-rw_y,\\
 \Psi^2=-\Omega(t) x+\dot{y} + r\omega_x,
 \end{array}
 \]
which define an affine subbundle ${\mathcal B}$ of ${\mathcal A}$.
Then, we have that \begin{eqnarray*}
 d^{{\mathcal T}^{\widetilde{{\mathcal
A}}}{\mathcal A}}\Psi^1&=&\Omega'(t)y\phi_0+\Omega(t)\tilde{T}^2
+\tilde{V}^1-r\tilde{V}^4,\\
d^{{\mathcal T}^{\widetilde{{\mathcal A}}}{\mathcal A}}\Psi^2&=&
-\Omega'(t)x\phi_0-\Omega(t)\tilde{T}^1 +\tilde{V}^2+r\tilde{V}^3.
\end{eqnarray*}

Thus, the Lagrange-d'Alembert equations (see (\ref{3.4'})) for the
system determined by $(L, {\mathcal B})$ are:
$$\left\{ \begin{array}{rcl}
m\ddot{x}&=&-\lambda_1,\\
m\ddot{y}&=&-\lambda_2,\\
k^2\dot{\omega}_x&=&-r\lambda_2,\\
k^2\dot{\omega}_y&=&r\lambda_1,\\
k^2\dot{\omega}_z&=&0,
\end{array}\right.$$
and the constraints $\Psi^1=0$ and $\Psi^2=0$.

In this case, a  basis of sections of vector subbundle $F\subset
{\mathcal T}^{\widetilde{{\mathcal A}}}{\mathcal A}_{|{\mathcal
B}}\to {\mathcal B}$ is given by
\[\Big\{
Z_1=-\frac{1}{m}\tilde{V}_1+\frac{r}{k^2}\tilde{V}_4\; ,\;
Z_2=-\frac{1}{m}\tilde{V}_2-\frac{r}{k^2}\tilde{V}_3\Big\}
\]
and then, the matrix ${\mathcal C}$ is
\[
{\mathcal C}=\left(\begin{array}{cc}
-\displaystyle\frac{1}{m}-\frac{r^2}{k^2}&0\\
0&-\displaystyle\frac{1}{m}-\frac{r^2}{k^2}
\end{array}\right).
\]

The projector $P$ over ${\mathcal T}^{\widetilde{{\mathcal
A}}}{\mathcal B}$ is
\[
P=Id+\frac{mk^2}{k^2+mr^2}Z_1\otimes d^{{\mathcal
T}^{\widetilde{{\mathcal A}}}{\mathcal A}}\Psi^1 +
\frac{mk^2}{k^2+mr^2}Z_2\otimes d^{{\mathcal
T}^{\widetilde{{\mathcal A}}}{\mathcal A}}\Psi^2.
\]
Since the solution of the unconstrained dynamics is
\[
R_L=\tilde{T}_0+\dot{x} \tilde{T}_1+ \dot{y} \tilde{T}_2+\omega_x
\tilde{T}_3+\omega_y \tilde{T}_4+ \omega_z \tilde{T}_5,
\]
then the solution of the nonholonomic problem will be
\begin{eqnarray*}
R_{nh}&=&P(R_L)=\tilde{T}_0+\dot{x} \tilde{T}_1+ \dot{y}
\tilde{T}_2+\omega_x \tilde{T}_3+\omega_y \tilde{T}_4+ \omega_z
\tilde{T}_5\\
&&+\frac{mk^2}{k^2+mr^2}(\Omega'(t)y+\Omega(t)\dot{y})Z_1-\frac{mk^2}{k^2+mr^2}(\Omega'(t)x+\Omega(t)\dot{x})Z_2.\\
\end{eqnarray*}
Thus, \begin{eqnarray*} \rho^{\tau _{\mathcal
A}}_{\widetilde{\mathcal
A}}(R_{nh})\kern-3pt&\kern-5pt=\kern-5pt&\kern-3pt\frac{\partial}{\partial
t}+ \dot{x}\frac{\partial}{\partial x} +
\dot{y}\frac{\partial}{\partial y}
-\frac{k^2}{k^2+mr^2}(\Omega'(t)y+\Omega(t)\dot{y})\frac{\partial}{\partial
\dot{x}}
 +\frac{k^2}{k^2+mr^2}(\Omega'(t)x+\Omega(t)\dot{x})\frac{\partial}{\partial
\dot{y}}\\
&&
 +\frac{mr}{k^2+mr^2}(\Omega'(t)x+\Omega(t)\dot{x})\frac{\partial}{\partial
\omega_x}
+\frac{mr}{k^2+mr^2}(\Omega'(t)y+\Omega(t)\dot{y})\frac{\partial}{\partial
\omega_y} \end{eqnarray*}

and, therefore, the equations of motion of the nonholonomic
problem are $$\left\{ \begin{array}{rcl}
\ddot{x}&=&-\displaystyle\frac{k^2}{k^2+mr^2}(\Omega'(t)y+\Omega(t)\dot{y}),\\[8pt]
\ddot{y}&=&\displaystyle\frac{k^2}{k^2+mr^2}(\Omega'(t)x+\Omega(t)\dot{x}),\\[8pt]
\dot{\omega}_x&=&\displaystyle\frac{mr}{k^2+mr^2}(\Omega'(t)x+\Omega(t)\dot{x}),\\[8pt]
\dot{\omega}_y&=&\displaystyle\frac{mr}{k^2+mr^2}(\Omega'(t)y+\Omega(t)\dot{y}),\\[6pt]
\dot{\omega}_z&=&0.
\end{array}\right.$$

Now, we take the coordinates $(t,x,y;p_t,p_x,p_y,u_x,u_y,u_z)$ on
${\mathcal A}^+$ and the corresponding coordinates
$(t,x,y;p_x,p_y,u_x,u_y,u_z)$ on $V^*$. Then, we obtain that the
extended Legendre transformation $Leg_L:{\mathcal A}\to {\mathcal
A}^+$ and the Legendre transformation $leg_L:{\mathcal A}\to V^*$
associated with $L$ are given by
\[\begin{array}{rcl}
Leg_L(t,x,y;\dot x,\dot y,\omega_x,\omega_y,\omega_z)&=&
(t,x,y;-L,m\dot x,m\dot y,k^2\omega_x, k^2\omega_y,
k^2\omega_z),\\
leg_L(t,x,y;\dot x,\dot y,\omega_x,\omega_y,\omega_z)&=&
(t,x,y;m\dot x,m\dot y,k^2\omega_x, k^2\omega_y, k^2\omega_z).
\end{array}
\]

Thus, it is easy to see that the image of ${\mathcal B}$ under the
Legendre transformation, which will be denoted by $\bar{\mathcal
B}$, is defined by the vanishing of the functions
\[
\begin{array}{l}
\psi ^1=\Omega(t)y+\displaystyle \frac{1}{m}p_x-\frac{r}{k^2}u_y,\\[5pt]
\psi ^2=-\Omega(t)x+\displaystyle \frac{1}{m}p_y+\frac{r}{k^2}u_x.
\end{array}
\]

Moreover, the induced Hamiltonian section $h=Leg_L\circ
leg_L^{-1}:V^*\to{\mathcal A}^+$ is given by
\[
h (t,x,y;p_x,p_y,u_x,u_y,u_z)=(t,x,y;-H,p_x,p_y,u_x,u_y,u_z),
\]
where the function $H$ is
$$H(t,x,y;p_x,p_y,u_x,u_y,u_z)=\displaystyle
\frac{1}{2}\Big(\frac{1}{m}(p_x^2+p_y^2)+\frac{1}{k^2}(u_x^2+u_y^2+u_z^2)\Big).$$

A curve $\gamma :I\to V^*$,
$t\mapsto(t,x(t),y(t);p_x(t),p_y(t),u_x(t),u_y(t),u_z(t))$, is a
solution of the equations of motion for the nonholonomic
Hamiltonian system $(h,\bar{\mathcal B})$ if it satisfies the
equations $$\left\{\begin{array}{rcl}
\dot{p_x}&=&-\bar\lambda_1,\\
\dot{p_y}&=&-\bar\lambda_2,\\
\dot{u}_x&=&-r\bar \lambda_2,\\
\dot{u}_y&=&r\bar\lambda_1,\\
\dot{u}_z&=&0,
\end{array}\right.$$
and the constraints $\psi^1=0$ and $\psi^2=0$.

In this particular example, the expression of the nonholonomic
bracket is given by
$$\begin{array}{l}
 \{\bar{h}',\bar{h}''\}_{nh}=\frac{\partial(H'-H'')}{\partial t}
 +\Big ( \frac{\partial H'}{\partial
x}\frac{\partial H''}{\partial p_x}-\frac{\partial H'}{\partial
p_x}\frac{\partial H''}{\partial x} \Big ) +\Big ( \frac{\partial
H'}{\partial y}\frac{\partial H''}{\partial p_y}-\frac{\partial
H'}{\partial p_y}\frac{\partial H''}{\partial y}  \Big )
\\[14pt]
\kern30pt- u_x\Big (  \frac{\partial H'}{\partial
u_z}\frac{\partial H''}{\partial u_y}-\frac{\partial H'}{\partial
u_y}\frac{\partial H''}{\partial u_z} \Big ) - u_y\Big (
\frac{\partial H'}{\partial u_x}\frac{\partial H''}{\partial
u_z}-\frac{\partial H'}{\partial u_z}\frac{\partial H''}{\partial
u_x} \Big ) -
 u_z\Big (  \frac{\partial H'}{\partial
u_y}\frac{\partial H''}{\partial u_x}-\frac{\partial H'}{\partial
u_x}\frac{\partial H''}{\partial u_y} \Big )\\[14pt]
\kern30pt- \frac{k^2m}{k^2+r^2m} \Big [ \Omega '(t)y-\frac{1}{m}
\frac{\partial H''}{\partial x}+\Omega (t) \frac{\partial
H''}{\partial p_y} +\frac{r}{k^2}(u_z \frac{\partial H''}{\partial
u_x}-u_x \frac{\partial H''}{\partial u_z})\Big ] \Big [
\frac{\partial (H'-H)}{\partial p_x}   -r \frac{\partial
(H'-H)}{\partial u_y }\Big ]
\\[14pt]
\kern30pt- \frac{k^2m}{k^2+r^2m} \Big [- \Omega '(t)x-\Omega (t)
\frac{\partial H''}{\partial p_x}-\frac{1}{m} \frac{\partial
H''}{\partial y} - \frac{r}{k^2}(u_y \frac{\partial H''}{\partial
u_z}-u_z \frac{\partial H''}{\partial u_y})\Big ] \Big [
\frac{\partial (H'-H)}{\partial p_y}   + r  \frac{\partial
(H'-H)}{\partial u_x }\Big ]\\[14pt]
\kern30pt+ \frac{k^2m}{k^2+r^2m} \Big [ \Omega '(t)y-\frac{1}{m}
\frac{\partial H'}{\partial x}+\Omega (t) \frac{\partial
H'}{\partial p_y} +\frac{r}{k^2}(u_z \frac{\partial H'}{\partial
u_x}-u_x \frac{\partial H'}{\partial u_z})\Big ] \Big [
\frac{\partial (H''-H)}{\partial p_x}   -r \frac{\partial
(H''-H)}{\partial u_y }\Big ]
\\[14pt]
\kern30pt+ \frac{k^2m}{k^2+r^2m} \Big [- \Omega '(t)x-\Omega (t)
\frac{\partial H'}{\partial p_x}-\frac{1}{m} \frac{\partial
H'}{\partial y} - \frac{r}{k^2}(u_y \frac{\partial H'}{\partial
u_z}-u_z \frac{\partial H'}{\partial u_y})\Big ] \Big [
\frac{\partial (H''-H)}{\partial p_y}   + r  \frac{\partial
(H''-H)}{\partial u_x }\Big ],
\end{array}$$
for sections $h'$ and $h''$ with associated functions $H'$ and
$H''$, respectively.

\section{Conclusions and Future work}
\label{conclusions} We have developed a general geometrical
setting for nonholonomic mechanical systems in the context of Lie
affgebroids.  We list the main results obtained in this paper:

\begin{itemize}

\item The notion of regularity of a nonholonomic mechanical system
with affine constraints on a Lie affgebroid was elucidated and
characterized in geometrical terms.

\item In the regular case, the solution of the nonholonomic problem
was obtained by projecting the unconstrained one using several
decompositions of the prolongation of the Lie affgebroid along the
affine constraint subbundle.

\item The hamiltonian formalism for nonholonomic systems is completely
analyzed and a  nonholonomic bracket was defined.

\item Several examples were discussed showing the versatility of
this geometric framework.

\end{itemize}
In a forthcoming paper we will study the reduction of the Lie
affgebroid nonholonomic dyna\-mics under symmetry and we will
obtain a Lie affgebroid version of the momentum equation
introduced in \cite{CoLeMaMa} for Lie algebroids.

Other goal we have proposed is to develop a geometric formalism
for vakonomic Mechanics and optimal control theory on Lie
affgebroids.

In addition, we will explore the construction of geometric
integrators for mechanical systems on Lie affgebroids and, in
particular, for nonholonomic systems.

\section*{Acknowledgments}
This work has been partially supported by MEC (Spain) Grants MTM
2006-03322, MTM 2004-7832, project ``Ingenio Mathematica" (i-MATH)
No. CSD 2006-00032 (Consolider-Ingenio 2010) and S-0505/ESP/0158
of the CAM. D. Iglesias wants to thank MEC for a Research Contract
``Juan de la Cierva".

\end{document}